\documentclass[a4paper,12pt,reqno]{amsart}

\usepackage{hyperref}

\usepackage{amsmath,amssymb,mathrsfs,mathtools}
\usepackage{graphicx}
\usepackage{multicol}
\usepackage{enumerate}
\usepackage{dsfont}

\newtheorem{theorem}{Theorem}[section]
\newtheorem{lemma}[theorem]{Lemma}
\newtheorem{proposition}[theorem]{Proposition}

\newtheorem{remark}[theorem]{Remark}
\newtheorem{definition}[theorem]{Definition}

\numberwithin{equation}{section}

\newcommand{\CO}{\mathbb{C}}
\newcommand{\RE}{\mathbb{R}}
\newcommand{\Gar}{\Gamma_{\text{reg}}}

\newcommand{\diag}{\operatorname{diag}}
\newcommand{\off}{\operatorname{off}}

\renewcommand{\Im}{\operatorname{Im}}
\renewcommand{\Re}{\operatorname{Re}}

\newcommand{\ve}{\varepsilon}
\newcommand{\ID}{\mathbb{I}}

\newcommand{\x}{{\bf x}}
\newcommand{\y}{{\bf y}}

\renewcommand{\r}{{\bf r}}
\newcommand{\p}{{\bf p}}
\renewcommand{\k}{{\bf k}}
\newcommand{\q}{{\bf q}}

\newcommand{\D}{\mathcal D}
\newcommand{\F}{\mathcal F}
\newcommand{\B}{\mathcal B}

\newcommand{\NA}{\mathbb N}
\newcommand{\de}{\delta}
\newcommand{\ga}{\gamma}
\newcommand{\la}{\lambda}

\newcommand{\Ga}{\Gamma}

\newcommand{\pv}{\mathbf{p}}
\newcommand{\yv}{\mathbf{y}}

\newcommand{\beq}{\begin{equation}}
\newcommand{\eeq}{\end{equation}}

\newcommand{\be}{\begin{equation}}
\newcommand{\ee}{\end{equation}}

\renewcommand{\Im}{\operatorname{Im}\,}
\renewcommand{\Re}{\operatorname{Re}\,}

\newcommand{\n}{\noindent}
\newcommand{\lf}{\left}
\newcommand{\ri}{\right}
\newcommand{\vs}{\vspace{0.5cm}}
\newcommand{\f}{\frac}
\renewcommand{\leq}{\leqslant}
\renewcommand{\geq}{\geqslant}

\newcommand{\Bou}{\mathcal{B}}

\newcommand{\reg}{\operatorname{reg}}
\newcommand{\sym}{\operatorname{sym}}
\newcommand{\Fou}{\mathcal{F}}

\newcommand{\lambdamin}{{\lambda_{\operatorname{max}}}}

\newcommand{\genf}{f}

\title[Three-Body Hamiltonian with Regularized Zero-Range Interactions in Dimension Three \hspace{5pt}]{Three-Body Hamiltonian with Regularized Zero-Range Interactions in Dimension Three}

\author[G. Basti]{Giulia Basti}
\address[G. Basti]{Gran Sasso Science Institute, 
Viale Francesco Crispi 7,  67100 L'Aquila, Italy}
\email{giulia.basti@gssi.it} 

\author[C. Cacciapuoti]{Claudio Cacciapuoti}
\address[C. Cacciapuoti]{DiSAT, Sezione di Matematica, Universit\`a dell'Insubria, Via Valleggio 11, 22100 Como, Italy}
\email{claudio.cacciapuoti@uninsubria.it	}

\author[D. Finco]{Domenico Finco}
\address[D. Finco]{Facolt\`a di Ingegneria, Universit\`a Telematica
Internazionale Uninettuno,  Corso Vittorio Emanuele II 39, 00186 Roma, Italy}
\email{d.finco@uninettunouniversity.net}

\author[A. Teta]{Alessandro Teta}
\address[A. Teta]{Dipartimento di Matematica G. Castelnuovo, Sapienza Universit\`a di Roma,  Piazzale  Aldo Moro, 5, 00185 Roma, Italy}
\email{teta@mat.uniroma1.it}

\date{}

\thanks{
The research that led to the present paper was partially supported by  
a grant of the group GNFM of INdAM }

\begin{document}

\begin{abstract}
We study the Hamiltonian for a system of three identical bosons in dimension three interacting via zero-range forces. In order to avoid the fall to the center phenomenon emerging in the standard Ter-Martirosyan--Skornyakov (TMS) Hamiltonian, known as Thomas effect,  we develop in detail a suggestion given in a seminal  paper of Minlos and Faddeev in 1962 and we construct a regularized version of the TMS Hamiltonian which is self-adjoint and bounded from below. The regularization is given by an effective three-body force, acting only at short distance, that reduces to zero the strength of the interactions when the positions of the three particles coincide. The analysis is based on the construction of a suitable quadratic form which is shown to be closed and bounded from below. Then, domain and action of the corresponding  Hamiltonian are completely  characterized and a regularity result for the elements of the domain is given. Furthermore, we show that the Hamiltonian is the norm resolvent limit of Hamiltonians with rescaled non local interactions, also called separable potentials, with a suitably renormalized coupling constant. 
\end{abstract}

\maketitle

\begin{footnotesize}
 
 \emph{MSC 2020: 
81Q10; 
81Q15; 
70F07; 
46N50. 
}  
 \end{footnotesize}

\vspace{1cm}

\section{Introduction}\label{secIntro}

In a system of nonrelativistic quantum particles at low temperature, the thermal wavelength is typically much larger than the range of the two-body interactions and therefore  the details of the interactions are irrelevant.  In these conditions, the effective behavior of the system is well described by a Hamiltonian with zero-range forces, where the only physical parameter characterizing the interaction is the scattering length. 

\n
The mathematical construction of such Hamiltonians as self-adjoint and, possibly, lower-bounded operators is straightforward in dimension one since standard perturbation theory of quadratic forms can be used. Moreover, the Hamiltonian can be obtained as the resolvent limit of approximating Hamiltonians with rescaled two-body smooth potentials (see \cite{BCFT} for the case of three particles and  \cite{GHL} for $n$ bosons). On the contrary, in dimensions two and three the interaction is too singular and more refined techniques are required for the construction. The two-dimensional case is well understood (\cite{DFT, DR}, see also \cite{GL} for applications to the Fermi polaron model), and it has been recently shown \cite{GH} that the Hamiltonian is the norm resolvent limit of Hamiltonians with rescaled smooth potentials and with a suitably renormalized coupling constant. 

\n
In dimension three, the problem is more subtle due to the fact that a   natural construction in the case of $n \geq 3$ particles, obtained following the  analogy with the one particle case, leads to the so-called TMS Hamiltonian \cite{tms} which is  symmetric but not self-adjoint. Furthermore, all its self-adjoint extensions are unbounded from below.  
Such instability property, known as Thomas effect, can be seen as  a fall to the center phenomenon, and it  is due to the fact that the interaction becomes too strong and attractive when (at least) three particles are very close to each other. This phenomenon does not occur in dimension two because the singularity of the wave function at the coincidence hyperplane is  a mild logarithmic one.

\n
The Thomas effect was first noted by Danilov \cite{dan} and then rigorously analyzed by Minlos and Faddeev \cite{mf1,mf2}, and it makes the Hamiltonian unsatisfactory from the physical point of view (for some recent mathematical contributions see, e.g., \cite{BFT,BT,FT,miche} with references therein). For other approaches to the construction of many-body contact interactions in $\RE^3$, we refer to  \cite{AHKS,  MS2, TH}.

\n
We note that a different situation occurs in the case (which is not considered here) of  a system made of two species of fermions interacting via zero-range forces, where it happens that for certain regime of the mass ratio the TMS Hamiltonian is in fact self-adjoint and bounded from below (for  mathematical results in this direction see, e.g.,  \cite{miche1,CDFMT,CDFMT1,FT2,miche2,m1,m2,m4,MS,MS1}). 
For a general mathematical approach to the construction of singularly perturbed self-adjoint operators in a Hilbert space, we refer to \cite{P1,P2}.

\n
Inspired by a suggestion contained in \cite{mf1}, in this paper we propose a regularized version of the TMS Hamiltonian for  a system of  three bosons and we prove that it is self-adjoint and bounded from below. Furthermore,  we show that the Hamiltonian is the norm  resolvent limit of approximating Hamiltonians with rescaled non-local interactions, also called separable potentials, and with a suitably renormalized coupling constant. 

\n
We stress that a more interesting problem from the physical point of view would be the approximation in norm resolvent sense by a sequence of Hamiltonians with (local) rescaled potentials as in \cite{GH} for the two dimensional case and \cite{BCFT, GHL} for the one-dimensional case. Such a result is more difficult to prove, and we plan to approach it in a forthcoming work. 

\n 
We also believe that our approach and results can be generalized to the case of three different particles. For the case of a system made of $N$ bosons in interaction with another particle, see \cite{ferretti-teta}. 

\n
In the rest of this section we introduce our Hamiltonian at heuristic level and discuss some of its properties. 
Let us consider  a system of three identical bosons with masses $1/2$ in the center of mass reference frame and let $\x_1, \x_2$ and $ \x_3 = -\x_1 - \x_2$ be the Cartesian coordinates of the particles. Let us introduce the Jacobi coordinates  $ \r_{23} \equiv \x, \; \r_1 \equiv \y$, namely

\[
\x  = \x_2 - \x_3, \;\;\;\;\;\; \y =\f{1}{2} (\x_2 + \x_3) - \x_1.
\]

\n
The other two pairs of  Jacobi coordinates in position space are $\, \r_{31}= -\f{1}{2}\x +\y,  \;\r_2 = - \f{3}{4} \x -\f{1}{2} \y\,$ and $\, \r_{12}= - \f{1}{2} \x -\y , \; \r_3 =   \f{3}{4} \x - \f{1}{2} \y$.  
Due to the symmetry constraint, the Hilbert space of states is 

\beq\label{hspa}
L^2_{\text{sym}}(\RE^6) = \Big\{ \psi \in L^2(\RE^6)\,  \;\; \text{s.t.}\;\;  \psi(\x,\y)= \psi(-\x,\y)= 
\psi\Big( \f{1}{2} \x + \y,  \f{3}{4} \x-\f{1}{2} \y   \Big) \Big\} .
\eeq

\n
Notice that the symmetry conditions in \eqref{hspa} correspond to the exchange of particles $2,3$ and $3,1$ and they also imply the symmetry under the exchange of particles $1,2$, i.e., $ \psi(\x,\y)= \psi\big( \f{1}{2} \x - \y, - \f{3}{4} \x -\f{1}{2} \y   \big)$. 

\n
The formal Hamiltonian describing the three boson system in the Jacobi coordinates reads

\beq\label{Hfor}
 H_0 + \mu \delta (\x) +\mu \delta( \y - \x/2) + \mu \delta( \y + \x/2)
\eeq

\n
where $ \mu \in \RE$ is a coupling constant and $H_0$ is the free Hamiltonian  

\beq\label{fham}
H_0 = - \Delta_{\x} - \f{3}{4} \Delta_{\y}. 
\eeq

\n
Our aim is to construct a rigorous version of \eqref{Hfor} as a self-adjoint, and possibly bounded from below, operator in $L^2_{\text{sym}}(\RE^6)$. In other words, we want to define a self-adjoint  perturbation of the free Hamiltonian \eqref{fham} supported by the coincidence hyperplanes

\[
\pi_{23} \!= \!\{\x_2 = \x_3\}\!=\! \{\x=0\}, \;\;\;\;
\pi_{31} \!= \!\{\x_3 = \x_1\}\!=\! \{   \y= \x/2 \}, \;\;\;\;
\pi_{12} \!= \!\{\x_1 = \x_2\}\!=\! \{\y = - \x/2 \}.
\]

\n
Following the analogy with the one particle case \cite{Albeverio}, a natural attempt is to define the TMS operator acting  as the free Hamiltonian outside the hyperplanes  and characterized by a (singular) boundary condition on each hyperplane. Specifically, on $\pi_{23}$ one imposes

\beq\label{bc230}
\psi(\x,\y) = \f{\xi(\y)}{x} + \beta \,\xi(\y) + o(1) , \;\;\;\; \text{for}\;\; x\rightarrow 0\, \;\; \text{and}\;\;\y \neq 0
\eeq

\n
where $x:= |\x|$, $\xi$ is a function depending on $\psi$ and 
\be\label{scatlen}
 \mathfrak{a}:= - \beta^{-1}  \in \RE
\ee
 has the physical meaning of two-body scattering length (and it can be related to $\mu$ via  a renormalization procedure). Notice that, due to the symmetry constraint,  \eqref{bc230} implies the analogous boundary conditions on $\pi_{31}$ and $\pi_{12}$. 

\n
As already recalled, the TMS operator  defined in this way is symmetric but not  self-adjoint and  its self-adjoint  extensions are all unbounded from below. Therefore, the natural problem arises of figuring out if and how one can modify the boundary condition \eqref{bc230} to obtain a bounded from below Hamiltonian.  In a comment on this point, at the end of the paper \cite{mf1}  the authors claim that it is possible to find another physically reasonable realization of $\tilde{H}$ as self-adjoint and bounded from below operator. They also affirm that the recipe consists in the replacement 

\beq\label{replace}
\beta \, \xi (\y) \; \rightarrow \; \beta  \, \xi (\y) + (K \xi)(\y)
\eeq
in the boundary condition \eqref{bc230}, where $K$ is a convolution  operator in the Fourier space with a kernel $K(\p-\p')$ satisfying

\be\label{defk}
K(\p) \sim \f{\gamma}{\,p^2} , \;\;\;\;\; \text{for} \;\;\;\; p \rightarrow \infty 
\ee

\n
with $p=|\p|$ and the positive constant $\gamma$ sufficiently large. 
The authors do not explain the reason of their assertion neither they clarify the physical meaning of the boundary condition \eqref{replace}.  
They only conclude: 
\lq\lq A detailed development of this point of view is not presented here because of lack of space\rq\rq  \hspace{0.05 cm} and, strangely enough, their idea has never been developed in the literature. 

\n
Almost 20 years later, Albeverio, H{\o}egh-Krohn and Wu \cite{AHKW} have proposed an apparently different recipe to obtain a bounded from below Hamiltonian, i.e., the replacement 

\[
\beta \, \xi(\y) \;\rightarrow \; \beta \, \xi(\y) + \f{\gamma}{y} \xi(\y)
\]
with $y=|\y|$, in the boundary condition \eqref{bc230}, where again the positive constant $\gamma$ is chosen sufficiently large. Also, the proof of this statement has been postponed to a forthcoming paper which has never been published. 
Even if it has not been explicitly noted by the authors of \cite{AHKW}, it is immediate to realize  that the two proposals contained in \cite{mf1} and  \cite{AHKW} essentially coincide in the sense that in \cite{mf1}  the term added in the boundary condition is the  Fourier transform of the term added in \cite{AHKW}.   It is also important to stress that, according to the claim in \cite{mf1}, only  the asymptotic behavior of $K(\p)$ for $|\p| \rightarrow \infty$ (see \eqref{defk}) is relevant to obtain a lower-bounded Hamiltonian. Correspondingly, it must be sufficient to require  only  the asymptotic behavior  $\gamma y^{-1} + O(1)$ for $y \rightarrow 0$ for the boundary condition in position space  in \cite{AHKW}.

\n
The above considerations suggest to define our formal regularized TMS Hamiltonian $\tilde{H}_{\text{reg}}$  as  an operator in $L^2_{\text{sym}}(\RE^6)$ acting as the free Hamiltonian outside the hyperplanes  and characterized by the following boundary condition on $\pi_{23}$

\beq\label{bc23}
\psi(\x,\y) = \f{\xi(\y)}{x} + (\Gamma_{\text{reg}} \xi ) (\y) + o(1) , \;\;\;\; \text{for}\;\; x\rightarrow 0\, \;\; \text{and}\;\;\y \neq 0
\eeq

\n
where  $\Gar$ is defined by

\beq\label{gamreg}
(\Gamma_{\text{reg}} \xi ) (\y): =  \Gamma_{\text{reg}}(y) \xi  (\y) = \lf(\beta + \f{\gamma}{y} \, \theta(y)\ri) \xi(\y), \;\;\;\;\;\;\;\; \beta \in \RE, \;\;\;\; \gamma >0
\eeq
and $\theta$ is a real cutoff function. 
Due to the symmetry constraint, \eqref{bc23} implies the boundary condition on $\pi_{31}$ and $\pi_{12}$
\[
\psi(\x,\y)= \psi\Big( \f{1}{2} \x - \y, - \f{3}{4} \x -\f{1}{2} \y   \Big) = 
\f{\xi(-\x)}{| \y - \x/2 |} + (\Gar \xi ) (-\x) + o(1), \;\;\;\; \text{for}\;\;  \lf| \y -\x/2\ri|\rightarrow 0 , \;\;\x \neq 0
\]
\[
\psi(\x,\y)= \psi\Big( \f{1}{2} \x + \y, \f{3}{4} \x -\f{1}{2} \y   \Big) = 
\f{\xi(\x)}{| \y +\x/2 |} + (\Gar \xi ) (\x) + o(1), \;\;\;\; \text{for}\;\; \lf| \y +\x/2 \ri| \rightarrow 0 ,\;\;\x \neq 0. 
\]

\n
We assume different hypothesis on $\theta$ depending on the situation. The first possible hypothesis is
\be \label{cod}
\theta\in L^\infty (\RE^+), \qquad |\theta(r) -1| \leq c \, r \quad  \text{for some }c>0.\tag{H1}
\ee
The simplest choice satisfying \eqref{cod} is the characteristic function
\be\label{cafu}
  \quad \theta(r) =  \left\{ \begin{aligned} &1 \qquad & r \leq b \\
						 &  0 & r > b
			         \end{aligned}\right. , \;\;\;\;\;\;\;\; b>0 .
\ee
The second possible  hypothesis requires some minimal smoothness 
\be \label{cos}
\theta \in C^2_{\text{bd}}(\RE^+)=\{ f: \,\RE^+ \rightarrow \RE^+, \, \text{with} \, f \in C^2(\RE^+)\, \text{and} \, f, f', f'' \, \text{bounded} \} ,  \qquad \theta(0)=1. \tag{H2}
\ee
Examples satisfying \eqref{cos} are $\theta(r)= e^{-r/b}$ or $\theta\in C_0^\infty (\RE^+)$  such that $\theta(r) = 1$ for $r \leq b$, $b>0$. 

\n
We stress that \eqref{cos} implies \eqref{cod}.
Hypothesis \eqref{cos} will be used only in Sect. \ref{s:src}, where we study the approximation with separable potentials, and in  Appendix \ref{app:A2}.

\n
Note that the crucial point is the behavior of 
$\theta$ at the origin, which represents the minimal requirement for the regularization 
of the dynamics at short distances. The support of the function $\theta$ is not relevant, in particular a simple choice would be $\theta(r)=1$. 

\n
Needless to say, the operator $\tilde{H}_{\text{reg}}$ is only formally defined since its domain and action are not clearly specified. 
Our aim is to construct  an operator  which represents the rigorous counterpart of  $\tilde{H}_{\text{reg}}$  using a quadratic form method. 
The main idea of the construction has been announced and outlined in \cite{FT},  where a more detailed historical account of the problem is given.  We also mention the recent paper \cite{miche}, where the construction is approached using the theory of self-adjoint extensions.

\n
Let us make some comments on the formal operator $\tilde{H}_{\text{reg}}$.

\n
As we already remarked,  the singular behavior of  $(\Gar \xi)(\y)$ for $y \rightarrow 0$ corresponds in the Fourier space to a convolution operator whose kernel has the asymptotic behavior \eqref{defk}. It will be clear in the course of the proofs in Sect. \ref{secQuad} that such a behavior  is  chosen in order to compensate the singular behavior of the off-diagonal term appearing in the quadratic form. In this sense, one can say that the singularity of $(\Gar \xi)(\y)$ for $y \rightarrow 0$ is the  minimal one required to obtain a self-adjoint and bounded from below Hamiltonian.

\n
Concerning the physical meaning of our regularization, we recall that we have replaced the parameter  $\beta$ in \eqref{bc230} with $\Gar$ in \eqref{bc23}. By analogy with the definition \eqref{scatlen}, we can introduce an effective, position-dependent scattering length 
\[
 \mathfrak{a}_{\text{eff}}(y) := - \Gar ^{-1} (y)
\]
which can be interpreted as follows. For simplicity, let us fix  $\beta>0$ and choose the cutoff \eqref{cafu}. Consider the zero-range  interaction between the particles $2, 3$ which takes place when $\x_2=\x_3$, i.e., for $\x=0$. In these conditions, the coordinate $y$ is the distance between the  third particle $1$ and the common position of particles $2, 3$. Then one has 
\[
 \mathfrak{a}_{\text{eff}}(y)=\mathfrak{a} \;\; \;\; \text{if} \;\; \;\; y>b , \;\;\;\;\;\; \mathfrak{a}_{\text{eff}}(y)=\f{\mathfrak{a} y}{y-\gamma \mathfrak{a}} \;\;\;\; \text{if}\;\;\;\;  y \leq b ,
\]
i.e.,  the effective scattering length associated with the interaction of particles $2, 3$ is equal to $\mathfrak{a}$ if the third particle $1$ is at a distance larger than $b$ while for distance smaller than $b$ the scattering length depends on the position of the particle $1$ and it decreases to zero, i.e., the interaction vanishes,  when the distance goes to zero. In other words, we introduce a three-body interaction which is a common procedure in certain low-energy approximations in nuclear physics. Such three-body interaction reduces to zero the two-body interaction when the third particle approaches the common position of the first two. This is precisely the mechanism that prevents in our model the fall to the center phenomenon, i.e., the Thomas effect.

\n
The paper is organized as follows. 

\n
In Sect. \ref{secMain}, starting from the formal Hamiltonian  $\tilde{H}_{\text{reg}}$, we construct
a quadratic form which is the initial point of our analysis and we formulate our main results. 

\n
In Sect. \ref{secQuad}, we prove that the quadratic form is closed and bounded from below for any $\gamma$ larger than a threshold explicitly given.

\n
In Sect. \ref{secHam}, we characterize  the self-adjoint and bounded from below Hamiltonian  $H$ uniquely associated with the quadratic form which is the
rigorous counterpart of  $\tilde{H}_{\text{reg}}$.

\n
In Sect. \ref{secAppr}, we introduce a sequence of approximating Hamiltonians $H_{\ve}$ with rescaled separable potentials and a renormalized coupling constant and we prove a uniform lower bound on the spectrum. 

\n
In Sect. \ref{s:src}, we  show that the Hamiltonian $H$ is the norm resolvent limit of the sequence of approximating Hamiltonians $H_{\ve}$. 

\n
In the Appendix, we prove a technical regularity result for the elements of the domain of $H$.

\n
In conclusion, we collect here some of the notation frequently used throughout the paper.

\n
- $\x$ is a vector in $\RE^3$ and $x=|\x|$.

\n
- $\hat{f} \equiv \mathcal F f$ is the Fourier transform of $f$.

\n
- For a linear operator $A$ acting in position space, we denote by $\hat{A} = \mathcal F A \mathcal F^{-1}$  the corresponding operator in the Fourier space.

\n
- $H^s(\RE^n)$ denotes the standard Sobolev space of order $s>0$ in $\RE^n$. 

\n
- $\|\cdot\|$ and $(\cdot, \cdot)$ are the norm and the scalar product in  $L^2(\RE^n)$, $\|\cdot\|_{L^p}$ is the norm in $L^p(\RE^n)$, with $p\neq 2$,  and $\|\cdot\|_{H^s}$ is the norm in $H^s(\RE^n)$. It will  be clear from the context if $n=3$ or $n=6$.

\n
- $f \big|_{\pi_{ij}} \in H^{s} (\RE^3)$ is the trace of $f \in H^{3/2 + s}(\RE^6)$, for any $s>0$.

 \n
-  $\Bou(\mathcal K, \mathcal H)$ is the Banach space of the linear bounded operators from  $\mathcal K$ to  $\mathcal H$, where $\mathcal K$,  $\mathcal H$ are Hilbert spaces and $\Bou( \mathcal H)=\Bou(\mathcal H, \mathcal H)$.

\n
- $c$ will denote numerical constant whose value may change from line to line.

\vs

\section{Construction of the Quadratic Form and Main Results}\label{secMain}

Here, we describe a heuristic procedure to construct the quadratic form $  (\psi, \tilde{H}_{\text{reg}} \psi)$ associated with the formal Hamiltonian $ \tilde{H}_{\text{reg}}$ defined in the introduction.  

\n
Since we  mainly work in the Fourier space,   we introduce the coordinates $ \k_{23} \equiv \k,\;$    $ \k_1 \equiv \p$, conjugate variables of $\x,\y$
\[
\k=\f{1}{2} (\p_2 - \p_3) , \;\;\;\;\;\; \p= \f{1}{3} (\p_2 + \p_3) - \f{2}{3} \p_1
\]
where $\p_1, \p_2$ and $\p_3= -\p_1 - \p_2$ are the momenta of the particles. 
The other two pairs of Jacobi coordinates in momentum space are $\, \k_{31}=  - \f{1}{2} \k + \f{3}{4} \p$,   $\,\k_2=    -\k -\f{1}{2} \p\;$ and $\;\k_{12} = - \f{1}{2} \k - \f{3}{4} \p$,  $\,\k_3 = \k- \f{1}{2} \p$.
In the Fourier space, the Hilbert space of states is equivalently written as
\beq\label{hspaf}
L^2_{\text{sym}}(\RE^6) = \Big\{ \psi \in L^2(\RE^6)\,  \;\; \text{s.t.}\;\; \hat{\psi}(\k,\p)= \hat{\psi}(-\k, \p)= \hat{\psi}\Big(\f{1}{2}\k + \f{3}{4} \p, \k -\f{1}{2} \p \Big) \Big\}. 
\eeq

\n
The symmetry conditions in \eqref{hspaf} correspond to the exchange of particles $2,3$ and $3,1$ and they also imply the symmetry under the exchange of particles $1,2$, i.e.,  $\psi(\k,\p)= \hat{\psi}\big(\f{1}{2}\k - \f{3}{4} \p, -\k -\f{1}{2} \p \big)$.
Moreover the free Hamiltonian is

\[
\hat{H}_0 = k^2 + \f{3}{4}p^2.
\]

\n
We also introduce the \lq\lq potential\rq\rq \hspace{0.05cm} produced by the \lq\lq charge density\rq\rq \hspace{0.05cm} $\xi$ distributed on the hyperplane $\pi_{23}$ by
\beq\label{pot23}
\big( \widehat{\mathcal G_{23}^{\la}} \hat{\xi} \big) (\k,\p):= \sqrt{\f{2}{\pi}} \, \f{\hat{\xi}(\p)}{k^2 + \f{3}{4} p^2 +\la}
,\;\;\;\;\;\;\;\;\la>0
\eeq
and one can verify that the function $ \mathcal G_{23}^{\la} \xi $ satisfies   the equation 
\beq\label{dist23}
\Big( (H_0 +\la) \mathcal G_{23}^{\la} \xi \Big) (\x,\y) = 4 \pi \, \xi (\y)  \,\delta (\x)
\eeq
in distributional sense. Analogously, we have
\begin{align}
\big( \widehat{\mathcal G_{31}^{\la}} \hat{\xi} \big)(\k,\p)&:= \sqrt{\f{2}{\pi}} \, \f{\hat{\xi}(-\k - \f{1}{2} \p  )}{k^2 + \f{3}{4} p^2 +\la},
&\Big( (H_0 +\la) \mathcal G_{31}^{\la} \xi \Big) (\x,\y) &= 4 \pi \, \xi (-\x)  \,\delta \Big(\f{1}{2} \x - \y \Big), \label{dist31}\\
\big( \widehat{\mathcal G_{12}^{\la}} \hat{\xi} \big) (\k,\p)&:= \sqrt{\f{2}{\pi}} \, \f{\hat{\xi}(\k - \f{1}{2} \p  )}{k^2 + \f{3}{4} p^2 +\la}, &\Big( (H_0 +\la) \mathcal G_{12}^{\la} \xi \Big) (\x,\y) &= 4 \pi \, \xi (\x)  \,\delta \Big(\f{1}{2} \x + \y \Big)\,\label{dist12}
\end{align}
 and the potential produced by the three charge densities is
\beq\label{potxi}
\big( \widehat{\mathcal G^{\la}} \hat{\xi} \big) (\k,\p) := \sum_{i<j} \big( \widehat{\mathcal G_{ij}^{\la}} \hat{ \xi} \big) (\k,\p) =
\sqrt{\f{2}{\pi}} \, \, \f{\hat{\xi}(\p)  + \hat{\xi}(\k - \f{1}{2} \p) + \hat{\xi}(-\k - \f{1}{2} \p  )    }{k^2 + \f{3}{4} p^2 +\la}. 
\eeq
Note that the function $ \widehat{\mathcal G^{\la}} \hat{\xi}$ is symmetric under the exchange of particles; hence, it belongs to $L^2_{\text{sym}}(\RE^6)$, see Eq. \eqref{hspaf}.

\n
These potentials  exhibit the same singular behavior required in the boundary conditions \eqref{bc23}. Indeed, we have for $ x \rightarrow 0, \; \y \neq0$
\begin{align}\label{G23x}
\Big(\mathcal G_{23}^{\la}  \xi \Big) (\x,\y)&=  \sqrt{\f{2}{\pi}} \, \f{1}{(2\pi)^3} \int\!\!d\k d\p \, e^{i \k \cdot \x + i \p \cdot \y} \,  \f{\hat{\xi}(\p)}{k^2 + \f{3}{4} p^2 +\la} \nonumber\\
&=  \sqrt{\f{2}{\pi}} \, \int\!\!  d\p \, e^{ i \p \cdot \y} \,  \hat{\xi}(\p) \, \f{ e^{- \sqrt{ \f{3}{4} p^2 +\la}\, \, x}}{4 \pi \,x}\nonumber\\
&= \f{\xi(\y)}{x} - \f{1}{(2\pi)^{3/2}} \int\!\!d\p\, e^{i \p \cdot \y} \, \sqrt{\f{3}{4} p^2 +\la}  \,\, \hat{\xi}(\p)  + o(1).
\end{align}
 
 \n
Taking into account of the contribution of the other two charge densities, we have for  $x \rightarrow 0, \; \y \neq 0$
\begin{align}\label{abG}
\Big(\mathcal G^{\la} \xi \Big) (\x,\y)=&\f{\xi(\y)}{x} - \f{1}{(2\pi)^{3/2}} \int\!\!d\p\, e^{i \p \cdot \y} \, \sqrt{\f{3}{4} p^2 +\la}  \,\, \hat{\xi}(\p) + o(1) \nonumber\\
& +  \f{1}{(2\pi)^{3/2}}  \int\!\!  d\p \, e^{i\p \cdot \y} \,\f{1}{2 \pi^2} \! \int\!\!d\k\, e^{i \k \cdot \x} \, \f{\hat{\xi}( \k-\f{1}{2} \p) + \hat{\xi}( -\k-\f{1}{2} \p) }{ k^2 + \f{3}{4} p^2 +\la}  \nonumber\\
=& \f{\xi(\y)}{x} - \f{1}{(2\pi)^{3/2}} \int\!\!d\p\, e^{i \p \cdot \y}\! \left(\!
\sqrt{\f{3}{4} p^2 \!+\! \la}  \,\, \hat{\xi}(\p) - \f{1}{\pi^2} \!\!\int\!\! d\p' \f{ \hat{\xi}(\p')}{p^2 + p'^2 + \p\cdot \p' +\la} \!\right) +o(1) .
\end{align}
The asymptotic behavior \eqref{abG} suggests to represent an element  $\psi$ satisfying \eqref{bc23} as
\beq\label{decops}
\psi= w^{\la} + \mathcal G^{\la} \xi
\eeq

\n
where $w^{\la}$ is a smooth function. Using \eqref{abG}, the boundary condition \eqref{bc23} is rewritten as

\beq\label{bcpos}
\big(\Gar \xi\big)(\y) + \f{1}{(2\pi)^{3/2}} \int\!\!d\p \, e^{i\p \cdot \y} \! \left(\!
\sqrt{\f{3}{4} p^2 \!+\!\la}  \,\, \hat{\xi}(\p) - \f{1}{\pi^2} \!\!\int\!\! d\p' \f{ \hat{\xi}(\p')}{p^2 + p'^2 + \p\cdot \p' +\la} \!\right) \!= w^{\la}(0,\y)\,
\eeq
where $w^{\la}(0,\cdot) \equiv w^{\la} \big|_{\pi_{23}}$ denotes the trace of $w^{\la}$ on the hyperplane $\pi_{23}$.  

\n
We are now ready to construct the energy form $  (\psi, \tilde{H}_{\text{reg}} \psi)$ associated with the formal  Hamiltonian $\tilde{H}_{\text{reg}}$ defined in the Introduction. 
For $\ve>0$, let us consider the domain $D_{\ve}= \{ (\x,\y) \in \RE^6 \;\; s.t.\;\; |\x_i - \x_j|> \ve \; \; \forall \; i\neq j\}$. Taking into account that   $\tilde{H}_{\text{reg}}$ acts as the free Hamiltonian in $D_{\ve}$, 
the decomposition  \eqref{decops} and the fact that $(H_0 + \la) \mathcal G^{\la} \xi =0$ in $D_{\ve}$, we have
\begin{align}\label{bobou}
 (\psi, \tilde{H}_{\text{reg}} \psi) &= \lim_{\ve \rightarrow 0} \int_{D_{\ve}} \!\!\!d\x d\y\, \overline{\psi(\x,\y)} (H_0 \psi) (\x,\y) \nonumber\\
&= \lim_{\ve \rightarrow 0} \int_{D_{\ve}} \!\!\!d\x d\y\, \overline{ (w^{\la} (\x,\y)  + \mathcal G^{\la} \xi (\x,\y) ) } \, \Big( (H_0 + \la) (     w^{\la}  + \mathcal G^{\la} \xi) \Big)  (\x,\y)      -\la\|\psi\|^2_{L^2}    \nonumber\\
&= (w^{\la}, (H_0 + \la) w^{\la}) - \la \|\psi\|^2_{L^2} + (\mathcal G^{\la} \xi, (H_0 + \la) w^{\la} ).
\end{align}

\n
Using \eqref{dist23}, \eqref{dist31}, \eqref{dist12} and the symmetry properties of $w^{\la}$, the last term in \eqref{bobou} reduces to
\begin{align}\label{intpar}
(\mathcal G^{\la} \xi, (H_0 + \la) w^{\la} ) &= 4\pi \int\!\!d\x d\y  \Big( \xi(\y) \delta (\x) + \xi (-\x) \delta( \y - \x/2) + \xi(\x) \delta(  \y +\x/2) \Big) w^{\la}(\x, \y)\nonumber\\
&= 12 \pi \int\!\! d \y \, \overline{\xi(\y)} \,w^{\la} (0,\y).
\end{align}
By \eqref{bobou}, \eqref{intpar} and the boundary condition \eqref{bcpos}, we finally arrive at the definition of the following quadratic form

\begin{definition}\label{def}
\begin{align}\label{f1}
 F (\psi) &  := \F^\lambda (w^{\la}) -  \la \|\psi\|^2 + 12 \pi \;\Phi^{\la} (\xi) , \\
\F^\la (w^{\la}) &:=  
\int \!\! d\k \,d\p \, \Big( k^2 + \f{3}{4} p^2+\lambda \Big) |\hat{w}^{\lambda} (\k,\p)|^2 ,\label{f2} \\
\!\Phi^{\la} (\xi) &:=\! \int\!\! d \p \, \overline{\hat{\xi} (\p)} \, \! \left(\! \sqrt{\f{3}{4} p^2 \!+\!\la}  \,\, \hat{\xi}(\p) - \f{1}{\pi^2} \!\!\int\!\! d\p' \f{ \hat{\xi}(\p')}{p^2 + p'^2 + \p\cdot \p' +\la} + (\hat{\Gamma}_{\textup{reg}} \hat{\xi})(\p) \!\right)\label{f3}
\end{align}
where $(\hat{\Gamma}_{\textup{reg}} \hat{\xi})(\p)$ is the Fourier transform of the function defined in Eq. \eqref{gamreg}.
We define the quadratic form $F$ on the domain 
\beq\label{f4}
\D( F) := \Big\{\psi\in L^2_{\textup{sym}}(\RE^6)\,|\, \psi=w^{\lambda}+\mathcal{G}^{\lambda} \xi,\, \, w^{\lambda} \in H^1(\RE^6), \; \xi\in H^{1/2}(\RE^3) \Big\}.
\eeq
\end{definition}

\vs

\begin{remark}
From the explicit expression of the potential \eqref{potxi}, one immediately sees that for any $\xi \in H^{1/2}(\RE^3)$
\[
\mathcal{G}^{\lambda} \xi   \in L^2(\RE^6), \;\;\;\;\;\; \text{and} \;\;\;\;\;    \;           \mathcal{G}^{\lambda} \xi \notin H^1(\RE^6)  \;\;\;\;\;\;\text{for} \;\;\;\;\; \xi \neq 0. 
\]
Therefore, we have $\D(F) \supset H^1(\RE^6)$ and, for fixed $\lambda$,  the decomposition $\psi=w^{\lambda}+\mathcal{G}^{\lambda} \xi$ is unique. 
\end{remark}

\n
In the rest of the paper, we assume Definition \ref{def} as the starting point of our rigorous analysis. Let us conclude this section collecting the main results we  prove in the  paper. First, we show that for any $\gamma>\gamma_c,$ where
\beq\label{eq:gamma0}
\gamma_c= \f{\sqrt{3}}{ \pi} \left( \f{4 \pi}{3 \sqrt{3} }- 1 \right) \simeq 0.782,
\eeq 
the quadratic form $F$ on the domain $\D(F)$ is closed and bounded from below. This is the content of the next theorem whose proof is presented in Sect. \ref{secQuad}.

\begin{theorem}\label{thm:main}
	{Assume \eqref{cod} and let $\gamma>\gamma_c$}, then: 
\begin{itemize}
\item[(i)]  there exists $\la_0>0$ such that for all $\lambda>\lambda_0$ the quadratic form $\Phi^\la$ in $L^2(\RE^3)$ defined in \eqref{f3},  is 
coercive and closed on the domain $\mathcal D (\Phi^\la) = H^{1/2}(\RE^3)$;
\item[(ii)]
the quadratic form $F, \D(F)$ on $L_{\text{sym}}^2(\RE^6)$ introduced in Definition \ref{def} is bounded from below and closed.
\end{itemize}
\end{theorem}

\n
Theorem \ref{thm:main} implies that $F,\D(F)$ defines a self-adjoint and bounded from below Hamiltonian $H, \D(H)$ in $L^2_\mathrm{sym}(\RE^6)$. If we denote by 
 $\Gamma^\lambda,\D(\Gamma^\lambda)$  the positive,  self-adjoint operator  associated with the quadratic form $\Phi^\la$, $\mathcal D (\Phi^\la) = H^{1/2}(\RE^3)$, then  
 domain and action of the Hamiltonian  are characterized in the following proposition.
 
\begin{theorem} \label{prop:ham} {Under the same assumptions as in Theorem \ref{thm:main}}, we have
\begin{align}
\D(H) &= \Big\{\psi \in \D(F) \, | \, w^{\la} \in H^2(\RE^6) , \, \xi \in \D(\Gamma^{\la}) ,\, \Gamma^{\la} \xi = w^{\la}\big|_{\pi_{23}} \Big\},   \label{domH}\\
H\psi&= H_0 \,w^{\la} - \la \, \mathcal G^{\la} \xi. \label{azH} 
\end{align}
\end{theorem}

\n
The proof of Theorem \ref{prop:ham} is deferred to Sect. \ref{secHam}.

\n
The next question we  address is the approximation through a regularized Hamiltonian $H_\ve,\D(H_\ve)$ with non-local interactions, also known as separable potentials. In order to define the approximating model, we need to first introduce some notation. 

\n
Let $\chi \in L^2(\RE^3, (1+x) d\x) \cap  L^1(\RE^3, (1+x) d\x)$, spherically symmetric, real valued, nonnegative and  such that $\int d\x \,  \chi(x) = 1$. Moreover,  set  
\begin{equation}\label{ellellprime}
\ell :=4\pi (\chi, (-\Delta)^{-1}\chi) = 4\pi \int d\k \, \frac{|\hat \chi (k)|^2}{k^2},\quad \ell' :=4\pi \int d\k \, \frac{|\hat \chi' (k)|^2}{k^2} ,
\quad
	\gamma_0=3\pi\sqrt{\frac{\ell\ell'}{2}}.
\end{equation}
For all $\ve>0$, we define the scaled function $\chi_\ve$ as 
\begin{equation}\label{chiep}
\chi_\ve(x) = \frac1{\ve^3} \chi(x/\ve)
\end{equation}
and the operator $g_\ve$ on $L^2(\RE^3)$
\begin{equation}\label{gve}
g_\ve := -4\pi \frac\ve\ell \Big( \ID+ \frac{\ve}{\ell} \Gamma_{\reg}\Big)^{-1} 
\end{equation}
with $\Gamma_{\reg}$ given in  \eqref{gamreg}.
Then, the approximating Hamiltonian $H_\ve,\D(H_\ve)$ on $L^2_{\sym}(\RE^6)$ is  defined as 
\be \label{afa2}
H_\ve = H_0  + \sum_{j=0}^2 S^j \big(|\chi_\ve\rangle \langle \chi_\ve|   \otimes g_\ve \big)   {S^j}^*,
 \qquad \D (H_\ve)  = H^2(\RE^6) \cap L^2_{\sym}(\RE^6)
\ee
where $S$ is  the permutation operator exchanging the triple of labels $(1,2,3)$ in the triple  $(2,3,1)$.  So that, 
\be\label{S}
S \hat \phi(\k,\p)  =   \hat \phi\Big(\frac{3}{4} \p - \frac{1}{2}\k , -\frac{1}{2}\p - \k\Big).
\ee
 Taking into account that $\chi_{\ve} \to \delta\,$ for $\ve \to 0$ in distributional sense, one sees that the three interaction terms in \eqref{afa2} for $\ve \to 0$ formally converge to zero-range interactions supported on the hyperplanes $\pi_{23}$, $\pi_{31}$, $\pi_{12}$.  
 
 \n
 Moreover, the operator $g_{\ve}$ plays the role of renormalized coupling constant. We also note that  in position space $g_{\ve}$ reduces to the multiplication operator by $g_{\ve}(y)$ which for $\ve$ small behaves as 
\[
g_{\ve}(y)= - 4 \pi \f{\ve}{\ell} + 4 \pi \f{\ve^2}{\ell^2} \Big( \beta + \f{\gamma}{y} \theta(y) \Big) + O(\ve^3).
\]
In particular, if we assume for simplicity that $\theta$ is the characteristic function \eqref{cafu} then  we find 
that for $y>b$ we have the standard behavior required to approximate a point interaction in dimension three with scattering length $ - \beta^{-1}$ (see  \cite{Albeverio}, chapter II.1.1, pages 111--112), while for $y\leq b$ we have introduced a dependence on the position $y$ such that the modified scattering length $- \big(\beta+\frac{\gamma}{y}\big)^{-1}$ goes to zero as $y \rightarrow 0$. 

\n
In Sect. \ref{secAppr} (see Theorem \ref{t:KKK}), we prove a uniform lower bound  for the spectrum of $H_{\ve}$, i.e., we show that there exists $\la_1 >0$, independent of $\ve$, such that $\inf\sigma(H_\ve)>-\lambda_1$. 

\n
Finally,  we prove the norm resolvent convergence  of $H_\ve,\D(H_\ve)$ to $H, \D(H)$ as $\ve \to0$. More precisely, the following result holds true and it is proved in Sect. \ref{s:src}.
 
\begin{theorem}\label{t:main} {Assume  \eqref{cos} and $\gamma>\max\{\gamma_0,2\}$}. Moreover, let us define $\lambdamin := \max\{\lambda_1, -\inf\sigma(H)\}$.
Then, for all $z\in \CO\backslash [-\lambdamin,\infty)$  there holds true
\[
 \|(H_\ve-z)^{-1}  -(H-z)^{-1} \| \leq c \, \ve^{\de} \qquad 0<\de<1/2.
\] 
\end{theorem}

\vs

\section{Analysis of the Quadratic Form}\label{secQuad}

\n
In this section, we prove closure and boundedness of the quadratic form $F$ defined by \eqref{f1}--\eqref{f4} for $\gamma >\gamma_c$ with $\gamma_c$ defined in \eqref{eq:gamma0}. To this end we first study the quadratic form $\Phi^\la$ in $L^2(\RE^3)$ given by \eqref{f3} and acting on the domain $\D(\Phi^\la)=H^{1/2}(\RE^3)$.
Recalling the definition of $\Gamma_\mathrm{reg}$  given  in \eqref{gamreg}  and using the fact that
\[
 \int\!\! d\yv \, \f{1}{y} |\xi(\yv)|^2 = \f{1}{2 \pi^2}  \int \!\! d\p\, d\q\,  \frac{\overline{\hat{\xi}(\p)} \hat{\xi}(\q)}{|\p - \q|^2 } 
\]
we write
\beq	\label{eq:phil}
\Phi^{\lambda}(\xi) = \Phi^{\lambda}_{\mathrm{diag}}({\xi})+\Phi^{\lambda}_\mathrm{off}({\xi})+ \Phi_\mathrm{reg} (\xi)  
\eeq

\n
where

\begin{align*}
	\Phi^{\lambda}_\mathrm{diag}({\xi})&=    \int \!\! d\p\, \sqrt{ \f{3}{4} p^2 +\lambda }\; |\hat{\xi}(\p)|^2 , 
	\\
	\Phi^{\lambda}_\mathrm{off}({\xi})&=-\frac1{\pi^2} \int \!\! d\p\, d\q\,  \frac{\overline{\hat{\xi}(\p)} \hat{\xi}(\q)}{p^2+q^2+\p \cdot \q +\lambda}, 
	 \\
	\Phi_\mathrm{reg} (\xi)&=\, \Phi^{(1)}_{\mathrm{reg}}(\xi)+\Phi_\mathrm{reg}^{(2)}(\xi) 
\end{align*}
with
\begin{align}\label{ay}
\Phi^{(1)}_{\mathrm{reg}}(\xi)&=\int\!\! d\y\,  a(y)  |\xi(\y)|^2 , \;\;\;\;\;\;\; 	a(y)= \beta + \f{\gamma}{y} (\theta (y)-1),\\
\Phi_\mathrm{reg}^{(2)}(\xi)&= \gamma \int\!\! d\yv \, \f{1}{y} |\xi(\yv)|^2 = \f{\gamma}{2 \pi^2}  \int \!\! d\p\, d\q\,  \frac{\overline{\hat{\xi}(\p)} \hat{\xi}(\q)}{|\p - \q|^2 }. \nonumber
\end{align}

\n
{Note that $a \in L^{\infty} (\RE^3)$ if we assume \eqref{cod} and $a, \nabla a \in L^{\infty}(\RE^3)$ if we choose \eqref{cos}}. 

\n
We will show that $\Phi^{\la}$ is equivalent to the $H^{1/2}$-norm. 
 First, we prove that $\Phi^\la(\xi)$ can be bounded from above by $\|\xi\|_{H^{1/2}}^2$. This is the content of the next proposition which ensures that $\Phi^\la$ is well defined on $\D(\Phi^\la)=H^{1/2}(\RE^3)$.

\begin{proposition}\label{prop:upper}
Assume \eqref{cod},  $\la >0$ and $\gamma>0$.  Then there exists $c>0$ such that
	\[
		\Phi^\la(\xi)\leq c\, \|\xi\|_{H^{1/2}}^2.
	\] 
\end{proposition}
\begin{proof}
	Using the bound (see \cite[ Remark 5.12]{K} or \cite{LS,Y}) 
	\begin{equation}\label{3.8a}
		 \int d\yv\,  \frac{| \xi(\yv)|^2}{y}\leq \frac{\pi}{2}\int d\pv\, p|\hat{\xi}(\pv)|^2
	\end{equation}
	we immediately get
	\[
		\Phi_{\mathrm{reg}}(\xi)\leq \|a\|_{L^{\infty}}  \|\xi\|^2 +\gamma\int d\yv\,  \frac{| \xi(\yv)|^2}{y}\leq \big( \|a\|_{L^{\infty}} +\gamma\f\pi2\big)\|\xi\|_{H^{1/2}}^2.
	\]
	Moreover (see \cite[Lemma 2.1]{FT2}),
	\be\label{stiof}
		\Phi^\la_\mathrm{off}({\xi})\leq \frac{2}{\pi^2}\int d\k_1 d\k_2\, \frac{k_1^{1/2}|\hat{\xi}(\k_1)|k_2^{1/2}|\hat{\xi}(\k_2)|}{k_1^{1/2}(k_1^2+k_2^2)k_2^{1/2}}\leq4\int d\k\, k|\hat{\xi}(\k)|^2\leq 4 \, \|\xi\|_{H^{1/2}}^2.
	\ee
	Since $\Phi^\la_\mathrm{diag}$ is clearly bounded by $\|\xi\|_{H^{1/2}}^2$, the thesis immediately follows from \eqref{eq:phil}.
	
\end{proof}

\vs

\n
The next step is  to bound from below $\Phi^\lambda$, which is our main technical result for the construction of the  Hamiltonian. Our main tool is the decomposition of the function $\hat\xi$ into partial waves (for an alternative  approach see, e.g., \cite{MS}).  
Then, we write 
\[
	\hat\xi(\p)=\sum_{\ell=0}^{+\infty}\sum_{m=-\ell}^\ell\hat\xi_{\ell m}(p)Y^\ell_m(\theta_\p,\varphi_\p)
\]
 where $Y^\ell_m$ denotes the Spherical Harmonics of order $\ell,m$ and $\p=(p,\theta_\p,\varphi_\p)$ in spherical coordinates. Accordingly, we find the following decomposition of the quadratic form $\Phi^\la$
\beq\label{eq:phi_dec}
	\Phi^\la(\xi)=\Phi^{(1)}_\mathrm{reg}(\xi)+\sum_{\ell=0}^{+\infty}\sum_{m=-\ell}^\ell \phi^\la_\ell(\hat\xi_{\ell m})
\eeq
where $\phi^\la_\ell$ is the quadratic form whose action on $g\in L^2((0,+\infty),p^2\sqrt{p^2+1} dp)$ is given by (see, e.g., \cite[Lemma 3.1]{CDFMT})
\[
	\phi^\la_\ell(g)=\phi^\la_{\mathrm{diag}}(g)+\phi^\la_{\mathrm{off},\ell}(g)+\phi^{(2)}_{\mathrm{reg},\ell}(g),
\]
where
\beq\label{eq:Fla}\begin{aligned}
	\phi^\la_{\mathrm{diag}}(g)=& \,  \int_0^{+\infty}dp\,p^2\sqrt{\frac{3}{4}p^2+\la}|g(p)|^2\\
	\phi^\la_{\mathrm{off},\ell}(g)=&\,-\f2\pi\int_0^{+\infty}dp_1\int_0^{+\infty}dp_2\,p_1^2\overline{g(p_1)}p_2^2 g(p_2)\int_{-1}^1dy\,\frac{P_\ell(y)}{p_1^2+p_2^2+p_1p_2y+\la}\\
	\phi_{\mathrm{reg},\ell}(g)=&\frac\gamma\pi \int_0^{+\infty}dp_1\int_0^{+\infty}dp_2\, p_1^2\overline{g(p_1)}p_2^2 g(p_2)\int_{-1}^1dy\, \frac{P_\ell(y)}{p_1^2+p_2^2-2p_1p_2y}
\end{aligned}\eeq
with $P_\ell(y)=\frac1{2^\ell \ell!}\frac {d^\ell}{dy^\ell}(y^2-1)^\ell$ the Legendre polynomial of degree $\ell$.
\\In the next lemma, we investigate the sign of $\phi^\la_{\mathrm{off},\ell.}$

\begin{lemma}\label{lm:phi_off} 
	Let $g\in L^2(\RE^+ ,p^2\sqrt{p^2+1}\,dp)$ and $\lambda >0$. Then,
\beq\label{eq:bound_theta_off}\begin{gathered}
	\phi_{\mathrm{off},\ell}^{0}(g)\geq \phi_{\mathrm{off},\ell}^{\la}(g)\geq 0\quad\text{for $\ell$ odd,}\\
	0\geq \phi_{\mathrm{off},\ell}^{\la}(g)\geq \phi_{\mathrm{off},\ell}^{0}(g)\quad\text{for $\ell$ even}
\end{gathered}\eeq
\end{lemma}
\begin{proof}
	The proof follows \cite{CDFMT}. For the sake of completeness, we give the details below. 
First, we rewrite
\[\begin{aligned}
		\phi^\la_{\mathrm{off},\ell}(g)=&-\frac2\pi\sum_{j=0}^{+\infty}(-1)^j\int_0^{+\infty}dp_1\int_0^{+\infty}dp_2\,\frac{p_1^{2+j}\overline{g(p_1)}p_2^{2+j}g(p_2)}{(p_1^2+p_2^2+\la)^{j+1}}\int_{-1}^1dy\,y^jP_\ell(y)\\
			=&-\frac2{\pi 2^\ell \ell!}\sum_{j=0}^{+\infty}(-1)^j\int_0^{+\infty}dp_1\int_0^{+\infty}dp_2\,\frac{p_1^{2+j}\overline{g(p_1)}p_2^{2+j}g(p_2)}{(p_1^2+p_2^2+\la)^{j+1}}\int_{-1}^1dy\,y^j\frac{d^\ell}{dy^\ell}(y^2-1)^\ell\\
			=&-\frac2{\pi 2^\ell \ell!}\sum_{j=0}^{+\infty}(-1)^{j}\int_0^{+\infty}dp_1\int_0^{+\infty}dp_2\,\frac{p_1^{2+j}\overline{g(p_1)}p_2^{2+j}g(p_2)}{(p_1^2+p_2^2+\la)^{j+1}}\int_{-1}^1dy\,\Big(\frac{d^\ell}{dy^\ell}y^j\Big)(1-y^2)^\ell
	\end{aligned}\]
	where in the last line we integrated by parts $\ell$ times. Next, we note that
	\[
		\frac1{(p_1^2+p_2^2+\la)^{j+1}}=\frac1{j!}\int_0^{+\infty}d\nu\,\nu^j e^{-(p_1^2+p_2^2+\la)\nu},
	\]
	hence,
	\beq\label{eq:theta_off_B}
		\phi^\la_{\mathrm{off},\ell}(g)=\sum_{j=0}^{+\infty}B_{\ell j}\int_0^{+\infty}d\nu\,\nu^j e^{-\la\nu}\left|\int_0^{+\infty}dp\,p^{2+j}g(p)e^{-p^2\nu}\right|^2
	\eeq
	with
	\[
		B_{\ell j}=-\frac2{\pi 2^\ell \ell!}\frac{(-1)^{j}}{j!}\int_{-1}^1dy\,\Big(\frac{d^\ell}{dy^\ell}y^j\Big)(1-y^2)^\ell.
	\]
	It's easy to see that $B_{\ell j}=0$ if $\ell$ and $j$ do not have the same parity. Moreover, $B_{\ell j}\leq0$ if $\ell,j$ are even and $B_{\ell j}\geq0$ if $\ell,j$ are odd. Then, \eqref{eq:theta_off_B} yields \eqref{eq:bound_theta_off}.
	
\end{proof}

\vs

\n
Thanks to Lemma \ref{lm:phi_off}, in order to obtain a lower bound we can neglect $\phi^\la_{\mathrm{off},\ell}$ with $\ell$ odd and focus on $\phi^0_{\mathrm{off},\ell}$ which control $\phi^\la_{\mathrm{off},\ell}$ with $\ell$ even. In the next lemma we show that $\phi^0_{\mathrm{off},\ell}$ and $\phi_{\mathrm{reg},\ell}$ can be diagonalized, note that we also include the analysis of $\phi^0_{\mathrm{off},\ell}$ with $\ell$ odd for later convenience.

\begin{lemma} \label{lemmavecchio}
Let $g$ be a real analytic function whose Taylor expansion near the origin 
\[
g(y) =\sum_{n=0}^\infty c_n \, y^n \qquad c_n\geq 0
\]
has  a radius of convergence bigger or equal than one. Define
\[
a_\ell = \int_{-1}^{1} P_\ell (y) \, g(y) \, dy, \quad \ell \in \NA
\] 
with $P_\ell$ being the Legendre polynomials. Then for any $\ell \in \NA$, we have
\[
a_\ell \geq 0 \qquad \qquad a_{\ell +2}\leq a_{\ell}.
\]
\end{lemma}
\begin{proof}
Using  $P_\ell(y)=\frac1{2^\ell \ell!}\frac {d^\ell}{dy^\ell}(y^2-1)^\ell$ and integrating by parts, one easily
prove the first claim. Integrating by parts again, one finds
\[
a_{\ell+2}=a_\ell
+(-1)^\ell \frac{1+2(\ell+1)}{2^{\ell+1}(\ell+1)!}\int_{-1}^1dy\,(y^2-1)^{\ell+1}\frac{d^\ell g}{dy^\ell}(y)
\]
and the monotonicity follows from the first claim.

\end{proof}

\vs 

\begin{lemma}\label{lm:S}
	Let $g\in L^2(\RE^+ ,p^2\sqrt{p^2+1}dp)$. Then,
	\[
		 \phi_{\mathrm{off},\ell}^0(g)=\int_\mathbb{R}dk|g^\sharp(k)|^2S_{\mathrm{off},\ell}(k), \qquad \phi_{\mathrm{reg},\ell}(g)=\int_\mathbb{R}dk |g^\sharp(k)|^2 S_{\mathrm{reg},\ell}(k)
	\]
	where 
\[
g^\sharp(k)=\frac1{\sqrt{2\pi}}\int_\mathbb{R}dx\, e^{-ikx}e^{2x}g(e^x), 
\] 

\beq\label{eq:Soff}
	S_{\mathrm{off},\ell}(k)=- 2\int_{-1}^1dy\,P_\ell(y)\f{ \sinh \Big( \!  k\arccos \big(\frac y2\big)\! \Big)}{ \sqrt{1- \frac {y^2}4 } \sinh (\pi k)}=
	\begin{cases}
		\displaystyle{-\int_{-1}^{1}\!\!\! dy\,  P_\ell(y) \, \,\f{  \cosh \big(k \arcsin (\frac y2) \big) }{ \sqrt{1- \frac{y^{2}}4 } \cosh (k \frac\pi 2)}}\quad& \text{for $\ell$ even,}\\
		\vspace{0.1cm}\\
		\displaystyle { \int_{-1}^{1}\!\!\! dy\,  P_\ell(y) \, \,\f{    \sinh (k \arcsin (\frac y2) )}{ \sqrt{1- \frac{y^{2}}4 }  \sinh (k \frac\pi2)}}  &  \text{for $\ell$ odd}
	\end{cases}
\eeq
and
\beq\label{eq:Sreg}
	S_{\mathrm{reg},\ell}(k)=\gamma\int_{-1}^1dy\,P_\ell(y)\frac{\sinh\big(k\arccos(-y)\big)}{\sqrt{1-y^2}\sinh(k\pi)}=
	\begin{cases}
		\displaystyle{\frac\gamma2\int_{-1}^1dy P_\ell(y)\frac{\cosh(k\arcsin(y))}{\sqrt{1-y^2}\cosh\big(k\frac\pi2\big)}}\quad&\text{for $\ell$ even,}\\
		\vspace{0.1cm}\\
		\displaystyle{\frac\gamma2\int_{-1}^1dy P_\ell(y)\frac{\sinh(k\arcsin(y))}{\sqrt{1-y^2}\sinh\big(k\frac\pi2\big)}}&\text{for $\ell$ odd}.
	\end{cases}
\eeq
Moreover,
\beq\label{mono1}
	S_{\mathrm{off},\ell}(k)\leq S_{\mathrm{off},\ell+2}(k)\leq0 \quad \text{for $\ell$ even,}\qquad S_{\mathrm{off},\ell}(k)\geq S_{\mathrm{off},\ell+2}(k)\geq0\quad \text{for $\ell$ odd}
\eeq
and
\beq \label{mono2}
	 S_{\mathrm{reg},\ell}(k)\geq S_{\mathrm{reg},\ell+2}(k)\geq0\quad \text{ $\forall\ell\geq0$ .}
\eeq
\end{lemma}
\begin{proof}
The proof is similar to \cite[Lemma 3.3, 3.5]{CDFMT}. 
For reader's convenience we give the details below. With the change of variables $p_1=e^{x_1}$ and $p_2=e^{x_2}$, 
we rewrite
	\[\begin{aligned}
		\phi^0_{\mathrm{off},\ell}(g)=&-\frac2\pi \int_{\mathbb{R}}dx_1\int_{\mathbb{R}}dx_2\, e^{3x_1}\overline{g(e^{x_1})}e^{3x_2}g(e^{x_2})\int_{-1}^1dy\,\frac{P_\ell(y)}{e^{2x_1}+e^{2x_2}+e^{x_1+x_2}y}\\
			=&-\frac1\pi\int_{\mathbb{R}}dx_1\,e^{2x_1}\overline{g(e^{x_1})}\int_{\mathbb{R}}dx_2\,e^{2x_2}g(e^{x_2})\int_{-1}^1dy\,\frac{P_\ell(y)}{\cosh(x_1-x_2)+\frac y2}.
	\end{aligned}\]
	Taking the Fourier transform we get
	\[\begin{aligned}
		\phi^0_{\mathrm{off},\ell}(g)=\int_{\mathbb{R}}dk\, |g^\sharp (k)|^2S_{\mathrm{off},\ell}(k)
	\end{aligned}\]
	with (see, e.g., \cite[p. 511]{GR})
\begin{equation}\label{Soffpg12}
		S_{\mathrm{off},\ell}(k)=-\frac1\pi\int_\mathbb{R} \!\! dx\, e^{-i kx}  \int_{-1}^1dy\,\f{P_\ell(y)}{\cosh (x) + \frac y2}  =  - 2\int_{-1}^1dy\,P_\ell(y)\f{ \sinh \Big( \!  k\arccos \big(\frac y2\big)\! \Big)}{ \sqrt{1- \frac {y^2}4 } \sinh (\pi k)}.
\end{equation}
Proceeding analogously for $\phi_{\mathrm{reg},\ell}(g)$ we find
\[
		\phi_{\mathrm{reg},\ell}(g)=\int_{\mathbb{R}}dk\, |g^\sharp (k)|^2S_{\mathrm{reg},\ell}(k)
	\]
	where
	\begin{equation}\label{Sregpg12}
		S_{\mathrm{reg},\ell}(k)=\frac\gamma{2\pi}\int_{\mathbb{R}}dx\,e^{-ipx}\int_{-1}^1dy\,\frac{P_\ell(y)}{\cosh(x)-y}=\gamma\int_{-1}^1dy\,P_\ell(y)\frac{\sinh\big(k\arccos(-y)\big)}{\sqrt{1-y^2}\sinh(k\pi)}.	\end{equation}
	Finally, noting that $\sinh(k\pi)=2\sinh(k\frac\pi2)\cosh(k\frac\pi2 )$ and
	\[\begin{aligned}
		\sinh \big(k\arccos(a)\big)&=\sinh\Big(k\frac\pi2-k\arcsin(a)\Big)\\
			&=\sinh\Big(k\frac\pi2\Big)\cosh(k\arcsin(a))-\cosh\Big(k\frac\pi2\Big)\sinh(k\arcsin(a))
	\end{aligned}\]
	and recalling that $P_\ell$ has the same parity of $\ell$ we get \eqref{eq:Soff} and \eqref{eq:Sreg}. 
	
	\n
	In order to prove the monotonicity properties \eqref{mono1}, it is sufficient to notice that the Taylor expansions of
\[
\f{ \cosh \big(k \arcsin (\frac y2) \big) }{ \sqrt{1- \frac{y^{2}}4 }} \qquad \qquad \f{ \sinh \lf( k \arcsin (\frac y2) \ri)  }{ \sqrt{1- \frac{y^{2}}4 } },
\]
have positive coefficients and invoke Lemma \ref{lemmavecchio}. A dilation of a factor 2 preserve the positivity of the coefficients and then
also \eqref{mono2} follows from Lemma \ref{lemmavecchio}.

\end{proof}

\vs

\noindent
Notice that 
\beq \begin{aligned} \label{coso}
	\phi^0_{\mathrm{diag}}(g)=&\frac{\sqrt{3}}{2}\int_0^{+\infty}dp\,p^3|g(p)|^2=\frac{\sqrt{3}}2 \int_{\mathbb{R}}dx\,e^{4x}|g(e^x)|^2=\frac{\sqrt{3}}2 \int_{\mathbb{R}}dk\,|g^\sharp(k)|^2.
\end{aligned}\eeq
Then we rewrite the quadratic form as 
\beq \label{diagonale}
\phi^0_\ell (g)=  \int_{\mathbb{R}}dk\,|g^\sharp(k)|^2 \, S_\ell (k),
\eeq
where
\beq \label{Stotale}
 S_\ell (k) = \f{\sqrt3}{2}+ S_{\mathrm{off},\ell}(k) + S_{\mathrm{reg},\ell}(k).
\eeq
Notice also, see \eqref{coso},  that if $\hat \xi (\p)= g(p) \, Y^{\ell}_m (\theta_\p, \varphi_\p)$ then $\| g^\sharp\|_{L^2(\RE)}= 
\| \xi\|_{\dot H^{1/2}(\RE^3)}$.
Equations \eqref{diagonale} and \eqref{Stotale} suggest that we need to ensure that $S_\ell$ is positive for $\ga>\ga_c$.

\n
From Lemma \ref{lm:S} and \cite[Lemma 3.2]{BT}, it follows that for any $\ell\geq2$ even
\beq\label{eq:S2}
	S_{\mathrm{off},\ell}(k)\geq S_{\mathrm{off},2}(k)\geq -B
\eeq
where 
\[
	B=\Big(\frac{50}{27}\pi-\frac{10}{3}\sqrt{3}+\frac{\sqrt{11}}{9}-\frac{10}{9}\arcsin\Big(\frac1{\sqrt{12}}\Big)\Big) \simeq 0.087.
\]
This is enough to control $\phi^0_{\mathrm{off},\ell}$ for $\ell\geq 2$ even with $\phi^\la_\mathrm{diag}$ as we will show in Proposition \ref{prop:lower}. It remains therefore to study what happens in $s-$wave. To this end, we introduce the auxiliary quadratic form $\Theta^\lambda_s(g)$, $s\in (0,1)$ acting on 
$L^2(\RE^+ ,p^2\sqrt{p^2+1}dp)$ and defined by
\[
\Theta^\lambda_s(g)=  s \, \phi^\la_\text{diag}(g)+ \phi^{0}_\text{off,0}(g) + \phi_{\mathrm{reg},0}(g). \\
\]
A key ingredient of the proof of closure and boundedness from below of $\Phi^\la$ is the following lemma.

\begin{lemma}\label{lemma:theta}
Let $g\in L^2(\RE^+ ,p^2\sqrt{p^2+1}dp)$ and $\ga>\ga_c$. Then there exists $s^* \in (0,1)$ such that $\Theta^{\lambda}_{s^*}(g) \geq 0$ for any $\la >0$.
\end{lemma}

\begin{proof}
Since $\phi^\la_{\mathrm{diag}}(g) \geq \phi^0_{\mathrm{diag}}(g)$, by Lemma \ref{lm:S} we have
\begin{equation}\label{te2}
\Theta^\la_s(g) \geq  \int_{\mathbb{R}}dk\,|g^\sharp(k)|^2\Bigg[s\frac{\sqrt{3}}2+\int_{-1}^{1}\!\!\! dy\,\bigg(\frac\gamma2\frac{\cosh(k\arcsin(y))}{\sqrt{1-y^2}\cosh\big(k\frac\pi2\big)}-\f{  \cosh \big(k \arcsin (\frac y2) \big) }{ \sqrt{1- \frac{y^{2}}4 } \cosh (k \frac\pi 2)}\bigg)\Bigg].
\end{equation}
The explicit computation of the two integrals in \eqref{te2} yields
\[
	\begin{aligned}
\Theta^\la_s(g) &\geq
 \frac{ \sqrt{3}}{2} \!\!\int\!\!dk\, |g^\sharp(k)|^2 \! \left[   s + \f{2 \gamma}{\sqrt{3} } \f{ \sinh \big( k\frac\pi2 \big) }{k \cosh\big(k\frac\pi2 \big)}  - \f{8}{\sqrt{3}} \f{  \sinh \big(k\frac\pi6  \big) }{k \cosh\big(k\frac\pi2 \big)}   \right]\\
  &=\colon \frac{ \sqrt{3}}{2} \!\!\int\!\!dk\, |g^\sharp(k)|^2 f(k),
\end{aligned}
\]
where
\[
f(k)= \f{f_1(k)}{f_2(k)}= \f{\sqrt{3}\, s\,  k \cosh \big(k\f{\pi}{2} \big) + 2 \gamma \sinh\big( k\f{\pi}{2}  \big)- 8 \sinh \big(k\f{\pi}{6} \big)}{\sqrt{3}\, k \, \cosh\big( k\f{\pi}{2} \big)}.
\]
It remains to show $f(k)\geq0$ to conclude.  We note that  $f(k)$ is even and thus it is enough to consider $k\geq 0.$ We have $f(0)>0$, $f_2(k) \geq 0$ and
\[\begin{aligned}
f_1'(k)&= (\sqrt{3} \, s + \pi \gamma) \cosh\Big(k \f{\pi}{2}\Big) +\frac{ \sqrt{3}\pi}{2} \, s \,k  \sinh\Big(k \f{\pi}{2}\Big) - \f{4 \pi}{3} \cosh\Big(k \f{\pi}{6}\Big)\\
&\geq (\sqrt{3} \, s + \pi \gamma) \cosh \Big(k\f{\pi}{2}\Big) - \f{4 \pi}{3} \cosh\Big(k \f{\pi}{6}\Big)\\
&\geq  \left(\sqrt{3} \, s + \pi \gamma  - \f{4 \pi}{3} \right)\cosh\Big(k \f{\pi}{6}\Big).
\end{aligned}\]
For $\gamma >\gamma_c$, let us choose $s^*$ such that $\max \{0,1-\f{\pi}{\sqrt{3}}(\gamma - \gamma_c) \} < s^* <1$.
Hence, we also have $f_1(k) \geq 0$ and the proof is complete. 

\end{proof}

\vs

\n
We are ready to prove the lower bound for $\Phi^\la(\xi)$. This is the content of the next proposition which together with Proposition \ref{prop:upper} shows that $\Phi^\la$ defines a norm equivalent to $\|\cdot\|_{H^{1/2}}$.

\begin{proposition}\label{prop:lower}
	Assume \eqref{cod} and $\gamma>\gamma_c$. Then, there exist $\la_0>0$ and $c_0>0$  such that
		\[ 
		\Phi^\la(\xi)>c_0\, \|\xi\|_{H^{1/2}}^2
	\]
	for any $\lambda > \la_0$.
\end{proposition}
\begin{proof}
	By \eqref{eq:phi_dec}, \eqref{eq:Fla} and Lemma \ref{lm:phi_off}, we get
	\[\begin{aligned}
		\Phi^\la(\xi)\geq&\Phi^{(1)}_{\mathrm{reg}}(\xi)+ \sum_{\substack{\ell=0\\\ell\mathrm{ even}}}^{+\infty}\sum_{m=-\ell}^\ell \big[\phi^\la_{\mathrm{diag}}(\hat\xi_{\ell m})+\phi^0_{\mathrm{off},\ell}(\hat\xi_{\ell m})+\phi_{\mathrm{reg},\ell}(\xi_{\ell m})\big]+ \sum_{\substack{\ell=1\\\ell\mathrm{ odd}}}^{+\infty}\sum_{m=-\ell}^\ell \phi^\la_{\mathrm{diag}}(\hat\xi_{\ell m})`.
	\end{aligned}\]
	Then, using Lemmata \ref{lm:S}, \ref{lemma:theta}, we obtain
	\[\begin{aligned}
		\Phi^\la(\xi)\geq& \Phi^{(1)}_\mathrm{reg}(\xi)+(1-s^*)\phi^\la_{\mathrm{diag}}(\hat\xi_{00})+\sum_{\substack{\ell=2\\ \ell\mathrm{ even}}}^{+\infty}\sum_{m=-\ell}^\ell \Big[\phi^\la_{\mathrm{diag}}(\hat\xi_{\ell m})+\phi^0_{\mathrm{off},\ell}(\hat\xi_{\ell m})\Big]+ \sum_{\substack{\ell=1\\\ell\mathrm{ odd}}}^{+\infty}\sum_{m=-\ell}^\ell \phi^\la_{\mathrm{diag}}(\hat\xi_{\ell m}).
	\end{aligned}\]
	Now, we note that Lemma \ref{lm:S} and the estimate \eqref{eq:S2} yield
	\[
		\sum_{\substack{\ell=2\\ \ell\mathrm{ even}}}^{+\infty}\sum_{m=-\ell}^\ell \Big[\phi^\la_{\mathrm{diag}}(\hat\xi_{\ell m})+\phi^0_{\mathrm{off},\ell}(\hat\xi_{\ell m})\Big]\geq\,\Big(1-\frac2{\sqrt3}B\Big)\sum_{\substack{\ell=2\\ \ell\mathrm{ even}}}^{+\infty}\sum_{m=-\ell}^\ell \phi^\la_\mathrm{diag}(\hat\xi_{\ell m})
\]
	where $\big(1-\frac2{\sqrt3}B\big)>0.$ Hence, setting $\Lambda=\min \{1-s^*,1-\frac2{\sqrt3}B\}>0$, we get
	\beq\label{eq:lower}
		\Phi^\la(\xi)\geq \Phi^{(1)}_\mathrm{reg}(\xi)+\Lambda\sum_{\ell=0}^{+\infty}\sum_{m=-\ell}^\ell \phi^\la_{\mathrm{diag}}(\hat\xi_{\ell m})=\Phi^{(1)}_{\mathrm{reg}}(\xi)+\Lambda\Phi^\la_{\mathrm{diag}}(\hat\xi).\,
	\eeq
	To conclude, we note
		\beq\label{infla}\begin{aligned}
			\Phi^{(1)}_\mathrm{reg}(\xi)+\Lambda\Phi^\la_\mathrm{diag}(\hat\xi)\geq &
				\Lambda\Phi^\la_\mathrm{diag}(\hat{\xi}) -\|a\|_{L^\infty} \int\!\! d\yv\,|\xi(\yv)|^2\\
				\geq& \Lambda\Phi^\la_\mathrm{diag}(\hat{\xi}) -  \f{\| a\|_{L^\infty}}{\sqrt{\la}} \int\!\! d \p \, \sqrt{\f{3}{4} p^2 + \la} \, |\hat{\xi}(\p) |^2 \\
				=& \left(\! \Lambda - \f{\| a\|_{L^\infty}}{\sqrt{\la}} \right) \Phi^\la_\mathrm{diag}(\hat\xi). 
		\end{aligned}\eeq
	From \eqref{eq:lower} and \eqref{infla} choosing $\la$ large enough we get the thesis since clearly $\Phi^\la_\mathrm{diag}(\hat\xi)\geq c\|\xi\|_{H^{1/2}}^2$.
	
\end{proof}

\vs

\n
We are now in position to  prove  Theorem \ref{thm:main} formulated in Sect. \ref{secMain}.

\begin{proof}[Proof of Theorem \ref{thm:main}] Point (i) is a consequence of Propositions \ref{prop:upper} and \ref{prop:lower}. For the proof of point (ii), we follow a standard strategy (see, e.g.,  \cite{CDFMT}). 
For the convenience of the reader, we give the details below.  By Proposition \ref{prop:lower}, we have

\[
F(\psi) = \mathcal F^{\la} (w^{\la}) + 12 \pi \,\Phi^{\la} (\xi) - \la \, \|\psi\|^2 \geq - \la \, \|\psi\|^2
\]
for any $\la >\la_0$ and then the form is bounded from below. Moreover, let us fix $\la>\la_0$ and define 
\beq\label{fla}
F^{\la}(\psi)= F(\psi) + \la \, \|\psi\|^2 = \mathcal F^{\la} (w^{\la}) + 12 \pi \,\Phi^{\la} (\xi)  ,
\eeq
on the domain $\D(F)$. 
Let us consider a sequence $\{\psi_n=w^\la_n+\mathcal G^\la\xi_n\}_{n\geq0} \subset \D(F)$, and $\psi \in  L^2_{\text{sym}}(\RE^6)$ such that $\lim_n \|\psi_n - \psi\|_{L^2}=0$ and $\lim_{n,m} F^{\la}(\psi_n - \psi_m)=0$. By \eqref{fla} we have  $\lim_{n,m} \mathcal F^{\lambda} (w^{\la}_n - w^{\la}_m)=0$ and $\lim_{n,m} \Phi^{\la}(\xi_n -\xi_m)=0$ or, equivalently,  $\{w^{\la}_n\}_{n\geq0}$  is a Cauchy sequence in $H_{\text{sym}}^1(\RE^6)$ and,  by Proposition \ref{prop:lower}, $\{\xi_n\}_{n\geq0}$ is a Cauchy sequence in $H^{1/2}(\RE^3)$. Then, there exist $w^{\la} \in H_{\text{sym}}^1(\RE^6)$ and $\xi \in H^{1/2}(\RE^3)$ such that
\beq\label{h1h12}
\lim_n \|w^{\la}_n - w^{\la}\|_{H^1}=0, \;\;\;\;\;\; \lim_n \|\xi_n - \xi\|_{ H^{1/2}}=0.
\eeq
Moreover, we also have
\beq\label{gxi2}
\lim_n \| \mathcal G^{\la} \xi_n - \mathcal G^{\la} \xi \|=0.
\eeq
Formulas \eqref{h1h12} and \eqref{gxi2} imply that $\psi_n = w^{\la}_n + \mathcal G^{\la} \xi_n$ converges in $L^2(\RE^6)$ to $w^{\la} + \mathcal G^{\la} \xi$. By uniqueness of the limit, we have that $\psi= w^{\la} + \mathcal G^{\la} \xi$ and then $\psi \in D(F)$. Furthermore, by \eqref{h1h12} we have 
\[
\lim_n F^{\la}(\psi - \psi_n)= \lim_n \mathcal F^{\la}( w^{\la} - w^{\la}_n) +12\pi \Phi^{\la} (\xi-\xi_n)=0.
\]
Thus,  we have shown that $F^{\la}$ and, a fortiori, $F$ are closed quadratic forms and this concludes the proof. 

\end{proof}

\begin{remark}
By \eqref{infla} one has 
\[
\la_0= \lf( \f{\| a\|_{L^\infty} }{\Lambda} \ri)^2.
\]
We underline that $\la_0$ depends on $\gamma$ both via $\| a\|_{L^\infty} $ and via  $\Lambda$.
In particular, as $\ga\to \ga_c$ we have $s^\ast \to 1$, so that $\Lambda \to 0$ and $\la_0\to \infty$.
In the concrete case \eqref{cafu}, we can take
\[
\la_0= \f{\gamma^2}{\Lambda^2 b^2}\quad  \text{if}\quad  \beta \geq 0, \qquad  
\la_0=\f{( |\beta|b + \gamma)^2}{\Lambda^2 b^2}\, \quad  \text{if} \quad  \beta < 0.
\]
From the proof of Theorem  \ref{thm:main}, it is clear that $-\lambda_0$ is a lower bound for the 
infimum of the spectrum of $H$. 
\end{remark}

\begin{remark}
We expect $\ga_c$ to be optimal that is, if $\ga < \ga_c$ one could argue as in \cite{FT2} and prove that $F$ is unbounded from below.
\end{remark}

\vs

\section{Hamiltonian}\label{secHam}

In this section,  we explicitly construct the Hamiltonian of our three bosons system. 
Let us first consider the quadratic form $\Phi^\la$, $\mathcal D (\Phi^\la) = H^{1/2}(\RE^3)$ in $L^2(\RE^3)$ (see \eqref{f3}). As a straightforward consequence of  Point (i) of Theorem \ref{thm:main},  such a quadratic form is closed and positive, and therefore it uniquely defines  a positive, self-adjoint operator $\Gamma^\la$ in $L^2(\RE^3)$ for $\lambda>\lambda_0$ characterized as follows

\begin{align}\label{doga}
\mathcal D(\Gamma^\la)=& \;\;\big\{ \xi \in H^{1/2}(\RE^3) \,|\, \exists g \in L^2(\RE^3) \;\;\text{s.t.}\;\; \Phi^\la(\eta,\xi)=(\eta,g) \;\;\;\text{for any }\; \eta\in H^{1/2}(\RE^3) \big\}\\
{\Gamma}^\la {\xi} =& \;\;{g} \;\;\;\; \text{for}\;\; \xi \in \mathcal D(\Gamma^\la).\label{azga}
\end{align}

\n
In the appendix,  we  prove that $\D(\Gamma^{\la}) = H^1(\RE^3)$ for $\gamma > \gamma_c^*$ (see Proposition \ref{p:fine}) and 

\begin{align}\label{opga}
(\hat{\Gamma}^{\la} \hat{\xi})(\p)&=  \sqrt{ \f{3}{4} p^2 +\lambda }\; \hat{\xi}(\p) 
-\frac1{\pi^2} \int \!\! d\q\,  \frac{ \hat{\xi}(\q)}{p^2+q^2+\p \cdot \q +\lambda}
 +(\widehat{a \,\xi})(\p) + \f{\gamma}{2 \pi^2}  \int \!\!  d\q\,  \frac{ \hat{\xi}(\q)}{|\p - \q|^2 } \nonumber\\
&=: (\hat{\Gamma}^{ \la}_{\text{diag}} \hat{\xi}) (\p)  +   (\hat{\Gamma}^{ \la}_{\text{off}} \hat \xi)(\p) +    (\hat{\Gamma}^{(1)}_{\text{reg}} \hat \xi)(\p) +    (\hat{\Gamma}^{(2)}_{\text{reg}} \hat \xi)(\p).
\end{align}

\n
Let us now consider the quadratic form $F$, $\D(F)$ in $L^2_{\text{sym}}(\RE^6)$. 
By Theorem \ref{thm:main}, such quadratic form uniquely defines a self-adjoint and bounded from below Hamiltonian  $H$, $\D(H)$ in $L^2_{\text{sym}}(\RE^6)$, next we prove Theorem \ref{prop:ham} which characterizes its domain and action. 
\begin{proof}[Proof of Theorem \ref{prop:ham}]
Let us assume that $\psi  = w^{\la} + \mathcal G^{\la} \xi \in \D(H)$. Then there exists $f \in L^2_{\text{sym}}(\RE^6)$ such that $F(v,\psi) = (v,f)$ for any $v = w^{\la}_v + \mathcal G^{\la} \xi_v \in \D(F)$ and $f=H\psi$.  Let us consider $v\in H^1(\RE^6)$, so that $\xi_v =0$ and 

\[
\int \!\! d\k d\p \, \Big( k^2 + \f{3}{4} p^2 + \la \Big) \overline{ \hat{v} (\k,\p)}  \, \hat{w}^{\la} (\k,\p) - \la (v , \psi) = (v,f).
\]

\n
Hence, $w^{\la} \in H^2(\RE^6)$ and $(H_0 + \la ) w^{\la} = f + \la \psi= (H + \la ) \psi$ which is equivalent to \eqref{azH}. 

\n
Let us consider $v \in \D(F)$ with $\xi_v \neq 0$.  Then

\begin{align}\label{eqcad}
\int \!\! d\k d\p \, \Big( k^2 + \f{3}{4} p^2 + \la \Big) \overline{ \hat{w}^{\la}_v (\k,\p)}  \, \hat{w}^{\la} (\k,\p) - \la (v , \psi)  +  12\, \pi \,\Phi^{\la}(\xi_v, \xi) = (v,f).
\end{align}

\n
Taking into account that

\begin{align*}
(v, f + \la \psi)&= (w^{\la}_v, (H + \la) \psi) + (\mathcal G^{\la} \xi_v, (H+\la)\psi) \nonumber\\
&= (w^{\la}_v, (H_0 + \la) w^{\la}) + (\mathcal G^{\la} \xi_v, (H_0+\la) w^{\la})
\end{align*}

\n
equation \eqref{eqcad} is rewritten as

\beq\label{anhu}
\Phi^{\la} (\xi_v, \xi)= \f{1}{12 \pi} (\mathcal G^{\la} \xi_v, (H_0+\la) w^{\la}).
\eeq

\n
It remains to compute the right hand side of \eqref{anhu}. Using \eqref{potxi} and the symmetry properties of $w^{\la}$ (see \eqref{hspaf}), we have

\begin{align*}
\f{1}{12 \pi} (\mathcal G^{\la} \xi_v, (H_0+\la) w^{\la}) &= \f{\sqrt{2}}{12 \pi^{3/2}} \int\!\! d\k d\p \, \Big( \hat{\xi}_v (\p) + \hat{\xi}_v (\k - \f{1}{2} \p) + \hat{\xi}_v (-\k - \f{1}{2} \p) \Big) \hat{w}^{\la} (\k, \p)\nonumber\\
&= \f{1}{3} (\xi_v, w^{\la}\big|_{\pi_{23}} ) + \f{\sqrt{2}}{6 \pi^{3/2}} \int\!\! d\k d\p \,  \hat{\xi}_v (\k - \f{1}{2} \p) \,  \hat{w}^{\la} (\k, \p)  \nonumber\\
&=\f{1}{3} (\xi_v, w^{\la}\big|_{\pi_{23}} ) + \f{\sqrt{2}}{6 \pi^{3/2}} \int\!\! d\k d\p \,  \hat{\xi}_v (\k - \f{1}{2} \p) \,  \hat{w}^{\la} (\f{1}{2}\k + \f{3}{4} \p, \k - \f{1}{2} \p) \nonumber\\
&= (\xi_v, w^{\la}\big|_{\pi_{23}} )
\end{align*}

\n
and by \eqref{anhu} we find the equation

\[
\Phi^{\la} (\xi_v, \xi)= (\xi_v, w^{\la}\big|_{\pi_{23}} )
\] 

\n
for any $\xi_v \in H^{1/2}(\RE^3)$. By definition of the operator $\Gamma^{\la}$ (see \eqref{doga}, \eqref{azga}), we conclude that $\xi \in \D(\Gamma^{\la})$ and $\Gamma^{\la} \xi= w^{\la}\big|_{\pi_{23}}$. 

\n
Let us now assume that $\psi  \in \D(F)$ with $  w^{\la} \in H^2(\RE^6)$,  $\xi \in \D(\Gamma^{\la})$ and $ \Gamma^{\la} \hat{\xi} = w^{\la}\big|_{\pi_{23}} $. For any $v = w^{\la}_v + \mathcal G^{\la} \xi_v \in \D(F)$ we have

\begin{align*}
F(v, \psi)&= (w^{\la}_v, (H_0 + \la)w^{\la}) - \la (v, \psi) + \phi^{\la}(\xi_v, \xi) \nonumber\\
&= (v, (H_0 + \la)w^{\la}) - (\mathcal G^{\la} \xi_v, (H_0 + \la)w^{\la}) - \la(v, \psi) + (\xi_v, \Gamma^{\la} \xi)\nonumber\\
&= (v, (H_0 w^{\la} - \la \mathcal G^{\la} \xi ) ) - (\xi_v, w^{\la}\big|_{\pi_{23}}) + (\xi_v, \Gamma^{\la} \xi)\nonumber\\
&=(v, (H_0 w^{\la} - \la \mathcal G^{\la} \xi ) ).
\end{align*}

\n
It is now sufficient to define $f = H_0 w^{\la} - \la \mathcal G^{\la} \xi $ to obtain that $\psi \in \D(H)$ and $f=H\psi$ and thus to conclude  the proof.

\end{proof}

\begin{remark}
We emphasize that the Hamiltonian $H, \D(H)$ is the rigorous counterpart of the formal regularized TMS Hamiltonian introduced in Sect. \ref{secIntro}. Indeed, for any $\psi \in L^2_{\textup{sym}}(\RE^6) \cap  C_0^{\infty}( \RE^6 \setminus \cup_{i<j} \pi_{ij})$ we have $\psi \in \D(H)$ and $H\psi=H_0\psi$, i.e., the Hamiltonian acts as the free Hamiltonian outside the hyperplanes. Moreover, we show that the boundary condition \eqref{bc23} is also satisfied. Let us consider $\psi \in \D(H)$  and let us recall that the corresponding charge $\xi$ belongs to $H^1(\RE^3)$. For $\x \neq 0$ we write

\be\label{wbc1}
\psi(\x,\y) - \f{\xi(\y)}{x} = (\mathcal G^{\la}_{23} \xi)(\x,\y) - \f{\xi(\y)}{x} +  (\mathcal G^{\la}_{31} \xi)(\x,\y) + (\mathcal G^{\la}_{12} \xi)(\x,\y) + w^{\la}(\x,\y)
\ee
and we compute the limit of the above expression for $x\rightarrow 0$ in the $L^2$-sense. Taking into account of \eqref{G23x}, we have
\[
\int\!\!d\y \left| (\mathcal G^{\la}_{23} \xi)(\x,\y) - \f{\xi(\y)}{x} + (\Gamma^{\la}_{\textup{diag}} \xi)(\y) \right|^2 = \int\!\!d\p \left| \f{e^{- \sqrt{\f{3}{4} p^2+ \la} x}-1}{x} \,\hat{\xi}(\p) + \sqrt{\f{3}{4}p^2 +\la} \,\hat{\xi}(\p)\right|^2
\]
which, by dominated convergence theorem, converges to zero for $x\rightarrow 0$. Moreover, for any $\eta \in C_0^{\infty}(\RE^3)$, with $\|\eta\|_{L^2}=1$, we estimate the difference $ (\mathcal G^{\la}_{31} \xi)(\x,\y)-(\mathcal G^{\la}_{31} \xi)(0,\y)$ proceeding as in \eqref{stiof}

\[\left| \int\!\!d\y \, \overline{\eta(\y)} \!\left[ (\mathcal G^{\la}_{31} \xi)(\x,\y) - (\mathcal G^{\la}_{31} \xi)(0,\y)\right] \right|  = \f{1}{2\pi^2} \left|    \int\!\!d\p \!\!\int\!\!d\q\, \f{ \overline{\hat{\eta}(\p)} \left(e^{i (\q + \f{1}{2}\p) \cdot \x} - 1\right) \hat{\xi}(\q) }{p^2+ q^2 +\p \cdot \q + \la} \right|
\]
Notice, see Remark \ref{afa},that 
\[
   \int\!\!d\p \!\!\int\!\!d\q\, \f{| {\hat{\eta}(\p)}|\, | \hat{\xi}(\q)| }{p^2+ q^2 +\p \cdot \q + \la} \leq c\, \| \eta\|\, \|\xi\|_{H^1},
\]
then we conclude that $ (\mathcal G^{\la}_{31} \xi)(\x,\y)-(\mathcal G^{\la}_{31} \xi)(0,\y) \rightarrow 0$ for $x \rightarrow 0$ in $L^2(\RE^3)$ by dominated convergence theorem
and the same is true for $ (\mathcal G^{\la}_{12} \xi)(\x,\y)-(\mathcal G^{\la}_{12} \xi)(0,\y)$. Note that $(\mathcal G^{\la}_{31} \xi)(0,\y) + (\mathcal G^{\la}_{12} \xi)(0,\y) = -( \Gamma^{\la}_{\textup{off}} \xi)(\y)$. 
For the last term in \eqref{wbc1} we have 

\begin{align*}
&\int\!\!d\y\, |w^{\la}(\x,\y) - w^{\la}(0,\y)|^2 = \f{1}{(2\pi)^3} \int\!\!d\p \left| \int\!\!d\k\, (e^{i \k \cdot \x} -1) \, \hat{w}^{\la} (\k,\p) \right|^2 \nonumber\\
& \leq  \f{1}{(2\pi)^3}  \left( \int\!\!d\k\, \f{|e^{i \k \cdot \x}-1|^2}{(k^2 + 1)^2} \right) \int\!\!d\p d\k \left| (k^2 +1)  \hat{w}^{\la} (\k,\p) \right|^2
\end{align*}
and then  $w^{\la}(\x,\y) - w^{\la}(0,\y) \rightarrow 0$ for $x \rightarrow 0$ in $L^2(\RE^3)$. Taking into account the above estimates, the condition $\Gamma^{\la} \xi = w^{\la}\big|_{\pi_{23}}$ and the decomposition $\Gamma^{\la}= \Gamma^{\la}_{\textup{diag}} + \Gamma^{\la}_{\textup{off}} + \Gamma_{\textup{reg}} $ we conclude 

\[
\lim_{x \rightarrow 0}\,  \left\| \, \psi(\x,\cdot) - \f{\xi}{x} - \Gamma_{\textup{reg}}\, \xi \,  \right\| =0
\]
which is precisely the boundary condition \eqref{bc23} satisfied in the $L^2$-sense.
\end{remark}

\vs

\n
Let us characterize the resolvent of our Hamiltonian.  We first introduce the shorthand notation  for the operator 
$\mathcal G^{\la}_{23} = G^{\la}$ (see \eqref{pot23}), 
i.e., 

\be\label{Gla}
G^\lambda  : L^2(\RE^3) \to L^2(\RE^6), \;\;\;  \qquad  (\hat{G}^{\lambda}\hat \xi) (\k,\p) := \sqrt\frac{2}{\pi} \frac{1}{k^2 + \frac{3}{4} p^2 + \lambda } \hat \xi (\p)
\ee
its adjoint is 
\[
G^{\lambda*}  : L^2(\RE^6) \to L^2(\RE^3), \;\;\;  \qquad  (\hat{G}^{\lambda*}\hat f) (\p) := \sqrt\frac{2}{\pi} \int\!\!d\k\, \frac{1}{k^2 + \frac{3}{4} p^2 + \lambda } \hat f (\k, \p). 
\]
Next, we prove the following preliminary result.

\begin{proposition}\label{p:GG*} For any $\lambda >0 $, there holds
 $G^\lambda \in \Bou(H^{s-\frac12}(\RE^3),H^s(\RE^6))$ for all $s<1/2$, hence, $G^{\lambda*} \in \Bou(H^{-s}(\RE^6), H^{-s+\frac12}(\RE^3))$ for all $s<1/2$. In particular, $G^\lambda \in \Bou(L^2(\RE^3),L^2(\RE^6))$ and $G^{\lambda*} \in \Bou(L^2(\RE^6),L^2(\RE^3))$. 
\end{proposition}
\begin{proof} We note that 
\[
\frac{1}{(k^2 +\frac34 p^2 + \lambda )^2} \leq \max\{(4/3)^2,1/\lambda^2\} \frac{1}{(k^2 + p^2 + 1)^2}.
\]
Hence 
\[
 \int d\k  \frac{(k^2 + p^2 +1)^s}{(k^2 +\frac34 p^2 + \lambda )^2}  \leq C_\lambda  \int_0^\infty dk \,k^2  \frac{(k^2 + p^2 +1)^s}{(k^2 +p^2 + 1 )^2}  =  C_{\lambda,s}  (p^2 +1)^{s-\frac12} .
\]
So that 
\[
\|G^\lambda \hat \xi\|^2_{H^s} = \frac2\pi \int d\k \, d\p \frac{(k^2 + p^2 +1)^s}{(k^2 +\frac34 p^2 + \lambda )^2} |\hat\xi(\p)|^2 \leq C_{\lambda,s} \|\xi\|^2_{H^{s-\frac12}}.  
\]
\end{proof}

\n
Let us recall the definition of the operator $S$ given in \eqref{S}, additionally we  
notice that 
\be\label{S2}
S^2 \hat \phi(\k,\p) = \hat\phi\Big(-\frac{3}{4} \p - \frac{1}{2}\k ,-\frac{1}{2}\p + \k\Big).
\ee
If $\hat \psi \in L^2_{\text{sym}}(\RE^6)$, it holds true 
\[ 
\hat \psi(\k,\p)  = S \hat \psi(\k,\p) = S^2 \hat \psi(\k,\p) .
\]
We note that the second equality is a consequence of the first one. Furthermore, we have 
 $S^* = S^2$. Taking into account of \eqref{potxi}, \eqref{Gla}, \eqref{S}, \eqref{S2}, we can write
\[
\mathcal G^{\la} = \sum_{j=0}^2 S^j G^{\la}.
\]

\n
We claim that  the resolvent $(H+\la)^{-1}$ of  $H, \D(H)$ computed in $z = -\lambda < -\lambda_0 $ is given by 
\[ 
R^\lambda  
= R_0^\lambda +\frac{1}{4\pi} \sum_{j=0}^2 S^j \, G^\lambda \, (\Gamma^\la)^{-1} \,  G^{\lambda*} \;\;\;\;\;\;\; \;\;\;\;\;
\]
where $R_0^{\la} = (H_0 + \la)^{-1}$ and $(\Gamma^\la)^{-1}$ is a well defined and bounded operator in $L^2(\RE^3)$ since $\Phi^\lambda$ is coercive. Indeed, let us consider $\;R^{\la} f$ for $f \in L^2_{\text{sym}}(\RE^6)$. We have  $R^{\la} f = w^{\la} + \mathcal G^{\la} \xi$, where 
$\;w^{\la} \equiv R_0^{\la} f   \in H^2(\RE^6)$, $\;\xi \equiv (4 \pi)^{-1} (\Gamma^\la)^{-1} \,  G^{\lambda*}f  \in \D( \Gamma^{\la})$ and $\Gamma^{\la} \xi = (4 \pi)^{-1} G^{\la*} f = R_0^{\la} f \big|_{23}$. Hence, $R^{\la}f \in \D(H) $ and $(H+\la) R^{\la}f= (H_0 + \la) R_0^{\la} f =f$. Therefore, we conclude  that $R^{\la}= (H+\la)^{-1}$.

\vs

\section{Approximating Hamiltonian}\label{secAppr}

In this section,  we prove a uniform  bound on the infimum of the spectrum of $H_\ve$  introduced in Sect. \ref{secMain}  and obtain the Konno--Kuroda  formula for its  resolvent (Theorem \ref{t:KKK}). 

\begin{remark}\label{a:chi}
Let us recall the scaled function $\chi_\ve$ defined in \eqref{chiep}, and the definitions of the constants $\ell$ and $\ell'$ in \eqref{ellellprime}. 
The assumptions on $\chi$ imply that $\hat\chi$ is real valued, Lipschitz,  that  $\ell,\ell'<\infty$, and that 
\[ \int d\k \frac{|\hat \chi( k) - \hat \chi( 0) |^2}{k^4}<\infty .\] 

\n Moreover, we recall  the definition of the  infinitesimal, position-dependent coupling constant in \eqref{gve}. 
In the position space $g_\ve$ is just the multiplication operator for the function (which we denote by the same symbol)
\[ 
g_\ve(y) = -4\pi \frac\ve\ell \Big( 1 + \frac{\ve}{\ell} \beta  + \frac{\ve}{\ell} \frac{\gamma}{y} \theta(y)   \Big)^{-1} = -4\pi \frac\ve\ell \Big( 1 + \frac{\ve}{\ell} a(y) + \frac{\ve}{\ell} \frac{\gamma}{y} \Big)^{-1}, 
\]
where $a(y)$ was introduced in \eqref{ay}. From now on,  we always  assume that $\ve< \ell/(2\|a\|_{L^\infty})$ so that $ 1 + \frac{\ve}{\ell} a(y) + \frac{\ve}{\ell} \frac{\gamma}{y}>1/2$ and $g_\ve$, as a function, is bounded, in particular  $\|g_\ve\|_{L^\infty}\leq 8\pi \ve/\ell$. 
\end{remark} 

\n
Let us consider the Hamiltonian $H_\ve$ defined in \eqref{afa2}. 
We remark that the term $\sum_{j=0}^2 S^j \big(|\chi_\ve\rangle \langle \chi_\ve|   \otimes g_\ve \big)   {S^j}^*$ is bounded (although not uniformly in $\ve$) in $L^2(\RE^6)$, with norm bounded by $3 \|\chi_\ve\|^2 \|g_\ve\|_{L^{\infty}}\leq (24\pi/\ell) \ve^{-2} \|\chi\|^2$, and therefore $H_{\ve}$, $\D(H_{\ve})$ is self-adjoint and bounded from below for any $\ve>0$.  

\n
As a first  step, we introduce the following operators which will play a crucial role in writing the Konno--Kuroda formula for the resolvent of $H_\ve$ (see Theorem \ref{t:KKK}).

\begin{definition}For any $\lambda>0$, let  us define

\[
\Gamma^{\la}_{\ve}   : D(\Gamma_{\reg})\subset L^2(\RE^3)\to  L^2(\RE^3)  \qquad 
\Gamma_\ve^\lambda :=   \Gamma_{\reg}    +   \Gamma^\lambda_{\diag,\ve}  + \Gamma^\lambda_{\off,\ve}
\]

\n
where $\Gamma^\lambda_{\diag,\ve} $ and $\Gamma^\lambda_{\off,\ve}$ are the bounded operators (see the remark below) 
\[
\Gamma^{\la}_{\diag,\ve}   : L^2(\RE^3)\to  L^2(\RE^3) \qquad 
(\hat \Gamma^{\la}_{\diag,\ve} \hat{\xi})(\p) : =   \sqrt{ \f{3}{4} p^2 +\lambda }\,  r\bigg(\ve \sqrt{\frac34 p^2 +\lambda}\bigg)\, \hat{\xi}(\p) ,
\]
with 
\[
r(s) := 4\pi \int  d\k \, \frac{|\hat\chi( s k )|^2 }{k^2(k^2 + 1)};
\]
and 
\be\label{Gaoe}
\Gamma^{\la}_{\off,\ve}   : L^2(\RE^3)\to  L^2(\RE^3) \qquad 
(\hat \Gamma^{\la}_{\off,\ve}  \hat{\xi})(\p) :=   
-8\pi \int \!\! d\q\,  \frac{\hat\chi\big(\ve\big|\frac12 \p +\q\big|\big)  \hat\chi\big(\ve\big| \p +\frac12 \q\big|\big)}{p^2+q^2+\p \cdot \q +\lambda} \hat{\xi}(\q). 
\ee
We also define the quadratic form associated with $\Gamma^\la_{\ve}$
\begin{equation}\label{Phive1}
\D(\Phi_\ve^\lambda) := \{\xi\in L^2(\RE^3) \; \textrm{\normalfont s.t.} \; |\cdot|^{-\frac12} \xi \in L^2(\RE^3)\} 
\end{equation}
\begin{equation}\label{Phive2}
 \Phi_\ve^\lambda ( \xi)  = ( \xi, \Gamma^\lambda_{\diag,\ve}  \xi) +  (  \xi, \Gamma_{\reg}  \xi)    + ( \xi, \Gamma^\lambda_{\off,\ve}  \xi) . 
\end{equation}

\end{definition}

\n
We will show in the proof of Lemma \ref{l:GammaveGamma} that the operators $ \Gamma^\lambda_{\diag,\ve} $  and $ \Gamma^\lambda_{\off,\ve}$ converge, as $\ve\to 0$, to the corresponding limiting operators $ \Gamma^\lambda_{\diag} $  and $ \Gamma^\lambda_{\off}$, defined in Eq. \eqref{opga}.

\vs

\begin{remark}\label{r:Gadiagve} We observe that $ \|\hat\chi\|_{L^\infty} \leq (2\pi)^{-3/2} \|\chi\|_{L^1} = (2\pi)^{-3/2} $, so that  $r(s)\leq 1$ and also $r(0)=1$. Additionally, we note the trivial bound $s r(s) \leq  \ell$. The latter implies  $\|\Gamma_{\diag,\ve}\|_{\Bou(L^2)}\leq  \ell/\ve$. 

\n
We also note that $sr(s)$ as a function of $s\in[0,+\infty)$ is strictly increasing, this is an immediate consequence of the identity 
\[
s\, r(s)  = 4\pi \int  d\k \, \frac{|\hat\chi( k )|^2 }{k^2} \frac{s^2}{k^2 + s^2}. 
\]
Since $\hat \chi$ is a Lipschitz function, interpolating with the $L^\infty$ estimate, we immediately obtain
\be \label{hummus}
|r(s)-1| \leq c \, |s|^\de \qquad  0\leq \de <1/ 2.
\ee

\n
Moreover, by  the Cauchy--Schwarz inequality , we find
\[\begin{aligned}
\|\Gamma^{\la}_{\off,\ve}{\xi}\|^2 \leq &     
(8\pi)^2 \| \chi_\ve\|^2 \int d\p\, d\q\,  \frac{\big|\hat\chi\big(\ve\big| \p +\frac12 \q\big|\big)\big|^2}{\big(p^2+q^2+\p \cdot \q +\lambda\big)^2} |\hat{\xi}(\q)|^2   \\ 
\leq &  
\frac{(8\pi)^2 \| \chi_\ve\|^2 \|\xi\|^2}{(2\pi)^3} \int d\p\,   \frac{4}{\big(p^2+2 \lambda\big)^2} = \frac{C}{\lambda^{1/2}} \frac{1}{\ve^3} \|\chi\|^2 \|\xi\|^2.
\end{aligned}
\]
Hence, $\|\Gamma_{\off,\ve}\|_{\Bou(L^2)} \leq C \lambda^{-1/4} \ve^{-3/2} \|\chi\|$ for some numerical constant $C$. 

\end{remark}

\n
We want to obtain a lower bound for $\Phi_{\ve}^\la(\xi)$. In the next lemma, we first analyze $(  \xi, \Gamma^\lambda_{\off,\ve}   \xi)$.

\begin{lemma}\label{56}
Let $\xi \in \D(\Phi_{\ve}^\la)$, $\lambda>0$ and $\gamma_0$ as in \eqref{ellellprime}. Then, 
\begin{equation}\label{lowergammaoff}
(  \xi, \Gamma^\lambda_{\off,\ve}   \xi) \geq - \gamma_0  \int d\y \frac{|\xi(\y)|^2}{y}. 
\end{equation}
\end{lemma}

\begin{proof}
\n
By \eqref{Gaoe}, the change of variable $\k=-\q - \f{1}{2} \p$ and the action of the operator $S$ (see \eqref{S}), we find

\begin{align}\label{Gaoer0}
( \xi, \Gamma^\lambda_{\off,\ve}  \xi)=  - 8\pi \! \int \!\!d\p\,d\k \,   \overline{\hat\chi(\ve k )\hat{\xi}(\p)} \,\,\f{ \hat\chi\big(\ve \big |\frac34 \p -\frac12 \k \big|\big) \hat{\xi} \big(-\frac12 \p -\k \big) }
{k^2+\frac34 p^2 +\lambda}=-8\pi \big(  \chi_\ve \xi , S R_0^\lambda   \chi_\ve \xi  \big) .
\end{align}

\n
It is convenient to write the r.h.s. of \eqref{Gaoer0} in the position space. To this aim, we denote by $R_0^\lambda(\x,\y;\x',\y')$ the integral kernel of the operator $R_0^\lambda$. Its explicit expression is given by the formula
\begin{align}
R_0^\lambda(\x,\y;\x',\y') = &  \frac1{(2\pi)^6} \int d\k\, d\p 
\frac{e^{i \k \cdot (\x - \x')+i \p \cdot (\y -\y')}}{k^2 + \frac34 p^2 +\lambda }  \label{dropped}\\ 
=  & \frac{\lambda}{ 4 \sqrt{3} \pi^3} \, \frac{1}{\frac34 	|\x - \x'|^2 + |\y-\y'|^2} \mathscr{K}_2 \Big(\sqrt{\frac43 \lambda } \sqrt{\frac34|\x - \x'|^2 + |\y-\y'|^2 }\Big)  \nonumber
\end{align}
where $\mathscr{K}_2$ is the modified Bessel function of the third kind and it is a nonnegative  function. 
By  the definition of $S$, see \eqref{S},  we obtain the formula 
\[
 (  \xi, \Gamma^\lambda_{\off,\ve}   \xi)  =  - 8\pi 
\int d\x\,d\y \, \chi_\ve(x) \overline{\xi(\y)} \int d\x'\,d\y' R_0^\lambda\Big(- \frac{1}{2} \x +\y  , -\frac{3}4\x - \frac12\y; \x',\y'\Big) \chi_\ve(x') \xi(\y'). 
\] 
To proceed, we add and subtract to $\xi(\y')$ the function $\xi(\y)$ and obtain 
\be\label{bibop}
 (  \xi, \Gamma^\lambda_{\off,\ve}   \xi)  =  - 
\int d\y |\xi(\y)|^2 J_\ve(\y)  - 
\int d\y \, d\y' \overline{\xi(\y)} \widetilde J_\ve(\y,\y')\big(\xi(\y') - \xi(\y)\big)
\ee
where 
\[
J_\ve(\y) := 8\pi  \int d\x\, d\x' \,  \chi_\ve(x)  \chi_\ve(x') \int d\y' R_0^\lambda\Big(- \frac{1}{2} \x +\y  , -\frac{3}4\x - \frac12\y; \x',\y'\Big)
\]
and 
\[
\widetilde J_\ve(\y,\y') := 8\pi  \int d\x\, d\x' \,  \chi_\ve(x) \chi_\ve(x') R_0^\lambda\Big(- \frac{1}{2} \x +\y  , -\frac{3}4\x - \frac12\y; \x',\y'\Big) . 
\]
We claim that 
\begin{equation}\label{voodoo}
 - 
\int d\y \, d\y' \overline{\xi(\y)} \widetilde J_\ve(\y,\y')\big(\xi(\y') - \xi(\y)\big) \geq 0.
\end{equation}

\n
To prove inequality \eqref{voodoo}, we reason as follows. The integral kernel of $R_0^\lambda$ is (pointwise) positive   and $\chi$ is a  nonnegative function; hence, $\widetilde J_\ve > 0$. Moreover,    the expression
\[
\frac34\Big|- \frac{1}{2} \x +\y - \x'\Big|^2 + \Big| -\frac{3}4\x - \frac12\y-\y'\Big|^2
\]
 is invariant if one changes $(\x,\y;\x',
\y') \to (-\x',\y';-\x,\y)$ (as one can check with a straightforward calculation), and   $R_0^\lambda\Big(- \frac{1}{2} \x +\y  , -\frac{3}4\x - \frac12\y;\x',\y'\Big)$ shares the same property. Taking  also into account the fact that   $\chi$ is spherically symmetric, it follows that $\widetilde J_\ve(\y,\y') = \widetilde J_\ve(\y',\y)$. The symmetry of $\widetilde J_\ve$ implies 
\[\begin{aligned}
&\int d\y \, d\y' \overline{\xi(\y)} \widetilde J_\ve(\y,\y')\big(\xi(\y') - \xi(\y)\big) \\ 
= & 
\frac1{2} \int d\y \, d\y' \overline{\xi(\y)} \widetilde J_\ve(\y,\y')\big(\xi(\y') - \xi(\y)\big)
+
\frac1{2} \int d\y \, d\y' \overline{\xi(\y')} \widetilde J_\ve(\y',\y)\big(\xi(\y) - \xi(\y')\big) \\
= & -  \frac1{2} \int d\y \, d\y' \widetilde J_\ve(\y,\y')\big|\xi(\y') - \xi(\y)\big|^2  \leq  0,
\end{aligned}\]
from which the inequality \eqref{voodoo} immediately follows. By \eqref{bibop} and \eqref{voodoo}, we find

\begin{equation}\label{inge}
 (  \xi, \Gamma^\lambda_{\off,\ve}   \xi) \geq  - 
 \int d\y |\xi(\y)|^2 J_\ve(\y).
\end{equation}

\n
To conclude the proof,  we are left to show that \eqref{inge} implies \eqref{lowergammaoff}. 
The following identity can be obtained  by integration  starting from identity \eqref{dropped}
\[
\int d\y' \, R_0^\lambda\Big(- \frac{1}{2} \x +\y  , -\frac{3}4\x - \frac12\y; \x',\y'\Big) = \frac{e^{-\sqrt\lambda \,  |\y -  (\frac12 \x +\x')|}}{4\pi  \big|\y - (\frac12 \x +\x')\big|}. 
\]
Hence, 
\begin{align}
\int d\y |\xi(\y)|^2 J_\ve(\y) =  &8\pi 
\int d\y |\xi(\y)|^2  \int d\x\, d\x'\,  \chi_\ve(x)  \chi_\ve(x') \frac{e^{-\sqrt\lambda \,  |\y -  (\frac12 \x +\x')|}}{4\pi  \big|\y - (\frac12 \x +\x')\big|} \nonumber \\ 
\leq   &  8\pi  
\int d\y |\xi(\y)|^2  \int d\x\, d\x'\,  \chi_\ve(x)  \chi_\ve(x') \frac{1}{4\pi  \big|\y -  (\frac12 \x +\x')\big|} \nonumber  \\ 
= & 8\pi  \int d\y |\xi(\y)|^2 \int d\k\, \frac{e^{i\k\cdot \y}}{k^2} \hat\chi(\ve k )\hat\chi(\ve k/2) \nonumber  \\ 
= & 32\pi^2  \int d\y |\xi(\y)|^2 \int_0^\infty dk\, \frac{\sin(ky)}{ky} \hat\chi(\ve k )\hat\chi(\ve k/2). \label{apa}
\end{align}
To proceed, we use the identity 
\[
 \int_0^\infty dk\, \frac{\sin(ky)}{ky} \hat\chi(\ve k )\hat\chi(\ve k/2) = -  \int_0^\infty dk\, \int_0^k ds \frac{\sin(sy)}{sy}  \frac{d}{dk} \big(\hat\chi(\ve k )\hat\chi(\ve k/2)\big). 
\]
It is easy  to check that for $k>0$ the function $\int_0^k ds \frac{\sin s}{s}$ has maxima in $k= n\pi$ for $n\in\NA$ odd, minima in $k=n\pi$ for $n\in\NA$ even, it is positive and has an absolute maximum in $k = \pi$. Hence,  $|\int_0^k ds \frac{\sin s}{s}| \leq \int_0^\pi ds \frac{\sin s}{s} \leq \pi$. By the latter considerations, we infer
\[\begin{aligned}
\int_0^\infty dk\, \frac{\sin(ky)}{ky} \hat\chi(\ve k )\hat\chi(\ve k/2) \leq &  \ve  \int_0^\infty dk\, \left|\int_0^k ds \frac{\sin(sy)}{sy} \right| \Big|\hat\chi'(\ve k )\hat\chi(\ve k/2) + \frac12\hat\chi(\ve k)\hat\chi'(\ve k/2)\Big|  \\ 
\leq &   \frac{\pi}{y}  
\Big( \big\| \hat\chi' \, \hat\chi(\cdot/2)\big\|_{L^1(0,\infty)} + 
 \frac12\big\| \hat\chi \, \hat\chi'(\cdot/2)\big\|_{L^1(0,\infty)} 
\Big) \\ 
\leq &   \frac{\pi}{y}  \frac{3\sqrt2}{2} 
\big\| \hat\chi\big\|_{L^2(0,\infty)}\big\| \hat\chi'\big\|_{L^2(0,\infty)} 
=  \frac{\pi}{(4\pi)^2}  \frac{3\sqrt{2}}{2}  \sqrt{\ell \ell'} \, \frac{1}{y}. 
\end{aligned}\]
Using the latter bound  in \eqref{apa} gives 

\begin{equation}\label{apa0}
 \int d\y |\xi(\y)|^2 J_\ve(\y) \leq \gamma_0  \int d\y \frac{|\xi(\y)|^2}{y}, 
\end{equation}

\n
with the  explicit expression for $\gamma_0$. By \eqref{inge} and \eqref{apa0} we conclude the proof of the lemma. 

\end{proof}

\n
Using the previous result, we can now establish a uniform  lower bound for $\Phi_\ve^\lambda ( \xi)$.

\begin{lemma}\label{l:Phive}
Assume \eqref{cod}, $\xi \in \D(\Phi_{\ve}^\la)$  and 
 $\gamma> \gamma_0$ (see  \eqref{ellellprime}). 
Then, there exist $\ve_0,\lambda_1,c>0$ such that 
\begin{equation}\label{lowerbound}
\Phi_\ve^\lambda ( \xi)  \geq c \bigg(\xi,\Big(\ID + \frac1{|\cdot|}\Big)\xi\bigg) 
\end{equation}
for all $\lambda>\lambda_1$ and $0<\ve<\ve_0$. 
\end{lemma}
\begin{proof} 

\n
We recall that in the position space the operator $\Gamma_{\reg}$ (see \eqref{gamreg} and \eqref{ay}) is just the multiplication by the function (denoted by the same symbol)
\[ 
\Gamma_{\reg}(y)    = \frac{\gamma}{y}  + a(y).
\] 
So that 
\[
(  \xi, \Gamma_{\reg}  \xi) \geq \gamma  \int d\y \frac{|\xi(\y)|^2}{y}
- \|a\|_{L^\infty}  \|\xi\|^2.
\]

\n
By the  above inequality and Lemma \ref{56}, we infer  
\begin{equation}\label{machine}
\Phi_\ve^\lambda ( \xi)  \geq \int d\p \, \left(
\sqrt{ \f{3}{4} p^2 +\lambda }\,  r\bigg(\ve \sqrt{\frac34 p^2 +\lambda}\bigg) - \|a\|_{L^\infty}\right) |\hat\xi(\p)|^2
+ (\gamma - \gamma_0)\int d\y \frac{|\xi(\y)|^2}{y}.
\end{equation}
Since   $s \to s \,r(s)$ is strictly increasing (see Remark \ref{r:Gadiagve}) one has that for all $\lambda \geq \lambda_1$
\[
\sqrt{ \f{3}{4} p^2 +\lambda }\, r\bigg(\ve \sqrt{\frac34 p^2 +\lambda}\bigg) \geq 
  \sqrt{ \lambda_1 }\, r\big(\ve \sqrt{\lambda_1}\big). 
\]
Now, fix $\lambda_1$ so that $\sqrt{\lambda_1} > \|a\|_{L^\infty}$ and $\ve_0$ so that for all $0<\ve<\ve_0$ there holds $\sqrt{\lambda_1} \,r \big(\ve \sqrt{\lambda_1}\big) >  \|a\|_{L^\infty}$ (which is possible since $\lim_{s\to0} r(s) = 1$). Then 
\[
\sqrt{ \f{3}{4} p^2 +\lambda }\, r\bigg(\ve \sqrt{\frac34 p^2 +\lambda}\bigg)- \|a\|_{L^\infty} \geq \tilde c
\]
for some positive constant $\tilde c$. Hence, from inequality \eqref{machine},  taking $\gamma>\gamma_0$, we  infer the lower bound \eqref{lowerbound} with $c = \min\{\tilde c, \gamma -\gamma_0\}$. 

\end{proof}

\n
To proceed,  we  need some further notation. We denote by $\frac{\gamma}{|\,\cdot\,|}$ the multiplication operator for  $\frac{\gamma}{y}$ and  define the operator 
\[
\nu_\ve : L^2(\RE^3) \to L^2(\RE^3) \qquad \nu_\ve :=  \bigg( \Big(\ID +\frac{\ve}\ell \Big(\Gamma_{\reg} - \frac{\gamma}{|\cdot|}\Big) \Big)^{1/2}+ i \Big(\frac{\ve}{\ell} \frac{\gamma}{|\cdot|}\Big)^{1/2}
\bigg)^{-1} . 
\]
Similarly to $g_\ve$, in the position space the operator $\nu_\ve$ acts as   the multiplication by the function (denoted  by the same symbol)
\[
\nu_\ve(y) =  \bigg( \Big(1 + \frac{\ve}{\ell} a(y)\Big)^{1/2}  + i  \Big(\frac{\ve}{\ell}\frac{\gamma}{y}  \Big)^{1/2} \bigg)^{-1}. 
\]
Obviously, we have
\[
\nu_\ve^* = \bigg( \Big(\ID +\frac{\ve}\ell \Big(\Gamma_{\reg} - \frac{\gamma}{|\cdot|}\Big) \Big)^{1/2}- i \Big(\frac{\ve}{\ell} \frac{\gamma}{|\cdot|}\Big)^{1/2}
\bigg)^{-1} .
\]
Moreover,  there holds the  identity 
\[
g_\ve = -4\pi \frac\ve\ell \nu_\ve \nu_\ve^*,
\]
and the bounds $\|\nu_\ve\|_{\Bou(L^2)} = \|\nu_\ve^*\|_{\Bou(L^2)} \leq \sqrt2 $ for all $\ve\leq \ell/(2\|a\|_{L^\infty})$.

\n
In the next  lemma we study the invertibility of the operator $\nu_\ve^*  \Gamma_\ve^\lambda \nu_\ve$. 

\begin{lemma} \label{l:apriori}For any $\lambda >0 $, there holds true the identity 
\begin{equation}\label{BveGave}
\frac{1}{4\pi} \nu_\ve^*  \Gamma_\ve^\lambda \nu_\ve = \frac{1}{4\pi}\frac{\ell}{\ve}- B_{\ve}^\lambda  \,
\end{equation}

\n
where   $B_\ve^\lambda$ is the bounded (although not uniformly in $\ve$) operator 
\[
 B_\ve^\lambda  : L^2(\RE^3)\to  L^2(\RE^3) \qquad B_\ve^\lambda  := \sum_{j=0}^2\big(  \langle \chi_\ve|   \otimes \nu_\ve^*  \big) S^j R_0^\lambda   \big(|\chi_\ve\rangle  \otimes \nu_\ve \big).\]
Moreover, assume \eqref{cod} and $\gamma> \gamma_0$. 
Then, $\nu_\ve^*  \Gamma_\ve^\lambda \nu_\ve$ is invertible in $L^2(\RE^3)$ for all $\lambda>\lambda_1$ and $\ve$ small enough, with inverse uniformly bounded in $\ve$ and  $\lambda$.
\end{lemma}
\begin{proof}
We start by  pointing out the identity (note that, obviously, $\nu_\ve$, $\nu_\ve^*$, and $\Gamma_{\reg}$ commute) 
\begin{equation}\label{1922}
\nu_\ve^*  \Gamma_{\reg}\nu_\ve = - \frac1{4\pi} \frac{\ell}{\ve} \Gamma_{\reg} g_\ve = 
 \frac\ell\ve \Big(\ID + \frac{1}{4\pi} \frac\ell\ve g_\ve \Big),
\end{equation}
hence, $\nu_\ve^*  \Gamma_{\reg}\nu_\ve$ is well defined in $L^2(\RE^3)$ with norm bounded by $3\ell/\ve$.  
Let us consider the quadratic form associated with  $\frac{1}{4\pi}\frac{\ell}{\ve}- B_{\ve}^\lambda $, 
\[
D(\widetilde\Phi_\ve^\lambda) = L^2(\RE^3) \qquad  
\widetilde\Phi_\ve^\lambda ( \xi) = \frac{1}{4\pi}\frac{\ell}{\ve} \| \xi\|^2 -  (\xi,  B_{\ve}^\lambda \xi). 
\]
We set 
\[
B_{\diag,\ve}^\lambda : =  \big(  \langle \chi_\ve|   \otimes \nu_\ve^*  \big) R_0^\lambda   \big(|\chi_\ve\rangle  \otimes \nu_\ve \big) \quad \text{and} \qquad 
B_{\off,\ve}^\lambda : =  \sum_{j=1}^2\big(  \langle \chi_\ve|   \otimes  \nu_\ve^*  \big) S^j R_0^\lambda   \big(|\chi_\ve\rangle  \otimes \nu_\ve \big) 
\]
so that 
\[B_\ve^\lambda  =B_{\diag,\ve}^\lambda + B_{\off,\ve}^\lambda ,
\]
and 
\[
\widetilde\Phi_\ve^\lambda ( \xi)  = \frac{1}{4\pi}\frac{\ell}{\ve} \| \xi\|^2 - (\xi,  B_{\diag,\ve}^\lambda \xi) - (\xi,  B_{\off,\ve}^\lambda \xi).
\]

\n
We first study  the term $ \frac{1}{4\pi}\frac{\ell}{\ve} \| \xi\|^2 - (\xi,  B_{\diag,\ve}^\lambda \xi)$ making use of  two identities. The first one is   
\[
\int d\k \, \frac{|\hat\chi(\ve k)|^2}{k^2 + \frac34 p^2 +\lambda }  =\frac{1}{4\pi} \frac\ell\ve  -   \frac{1}{4\pi} \sqrt{\frac34 p^2 +\lambda } \,  r\bigg(\ve \sqrt{\frac34 p^2 +\lambda}\bigg),
\]
and  gives
\[\begin{aligned}
\big(\chi_\ve \xi, R^\lambda_0  \chi_\ve \xi\big) =   & \int d\p \,  d\k \, \frac{|\hat\chi(\ve k)|^2 | \hat \xi (\p)|^2}{k^2 + \frac34 p^2 +\lambda }  \\ 
= & \frac{1}{4\pi}\frac\ell\ve \|\xi\|^2 -\frac{1}{4\pi}\int d\p \,    \sqrt{\frac34 p^2 +\lambda } \,  r\bigg(\ve \sqrt{\frac34 p^2 +\lambda}\bigg) |\hat \xi(\p)|^2 \equiv 
\frac{1}{4\pi}\frac\ell\ve \|\xi\|^2 -  \frac{1}{4\pi}(\xi, \Gamma^\lambda_{\diag,\ve} \xi) . 
\end{aligned}\]
The second one is 
$
\nu_\ve^* \, \nu_\ve =  
 \ID - \frac\ve\ell \nu_\ve^* \, \Gamma_{\reg} \,  \nu_\ve .
$
Now, we can compute 
\be\label{bdie}
\begin{aligned}
\frac{1}{4\pi} \frac\ell\ve \|\xi\|^2 - (\xi,  B_{\diag,\ve}^\lambda \xi)  = & \frac{1}{4\pi}  \frac\ell\ve \|\xi\|^2 - \Big(\xi, \big(  \langle \chi_\ve|   \otimes \nu_\ve^*  \big) R_0^\lambda   \big(|\chi_\ve\rangle  \otimes \nu_\ve \big) \xi\Big) \\
 = &\frac{1}{4\pi} \frac\ell\ve \|\xi\|^2 - \Big(\chi_\ve (\nu_\ve\xi),R_0^\lambda  \,  \chi_\ve(\nu_\ve  \xi) \Big)\\  
 = & \frac{1}{4\pi} \frac\ell\ve \|\xi\|^2 -\frac{1}{4\pi}  \frac\ell\ve \| \nu_\ve  \xi \|^2+ \frac{1}{4\pi}(\nu_\ve  \xi, \Gamma^\lambda_{\diag,\ve} \nu_\ve  \xi)  \\ 
 = &  \frac{1}{4\pi} (\nu_\ve  \xi, \Gamma_{\reg} \nu_\ve  \xi)  + \frac{1}{4\pi}(\nu_\ve  \xi, \Gamma^\lambda_{\diag,\ve} \nu_\ve  \xi) .
\end{aligned}\ee

\n
Next, we study the term $(\xi,  B_{\off,\ve}^\lambda \xi)$. We have that 
\begin{equation}\label{Phioffve}
\begin{aligned}
-(\xi,B_{\off,\ve}^\lambda \xi)  = & - \Big(\xi, \sum_{j=1}^2\big(  \langle \chi_\ve|   \otimes \nu_\ve^*  \big) S^j R_0^\lambda   \big(|\chi_\ve\rangle  \otimes \nu_\ve \big) \xi\Big) \\ 
= & - \sum_{j=1}^2  \Big(  \chi_\ve (\nu_\ve\xi) , S^j R_0^\lambda   \chi_\ve(\nu_\ve \xi ) \Big) \\
= &   - \int d\p\,d\k \frac{   \overline{\hat\chi(\ve k )}\overline{\widehat{\nu_\ve\xi}(\p)} \, \hat\chi\big(\ve \big |\frac34 \p -\frac12 \k \big|\big) \widehat{\nu_\ve\xi} \big(-\frac12 \p -\k \big) }
{k^2+\frac34 p^2 +\lambda}\\ 
&  - \int d\p\,d\k \frac{   \overline{\hat\chi(\ve k )}\overline{ \widehat{\nu_\ve \xi}(\p)} \, \hat\chi\big(\ve \big |\frac34 \p +\frac12 \k \big|\big) \widehat{\nu_\ve\xi}\big(-\frac12 \p + \k \big) }
{k^2+\frac34 p^2 +\lambda}\\ 
= & 
-2  \int d\p\,d\q \frac{   \overline{\hat\chi\big(\ve \big |\p +\frac12 \q \big|\big)} \hat\chi\big(\ve \big |\frac12 \p + \q \big|\big) \overline{ \widehat{\nu_\ve \xi}(\p)} \,  \widehat{\nu_\ve \xi}(\q)  }
{p^2+ q^2 + \p\cdot \q+\lambda} \\
\equiv &   \frac{1}{4\pi} (\nu_\ve  \xi, \Gamma^\lambda_{\off,\ve} \nu_\ve  \xi) . 
\end{aligned}
\end{equation}

\n
By \eqref{bdie} and \eqref{Phioffve}, we find 
\[
\begin{aligned}
\widetilde\Phi_\ve^\lambda ( \xi) =  & \frac{1}{4\pi}\frac{\ell}{\ve} \| \xi\|^2 -  (\xi,  B_{\ve}^\lambda \xi)   \\ 
= & \frac{1}{4\pi} \big((\nu_\ve  \xi, \Gamma_{\reg} \nu_\ve  \xi)  +  (\nu_\ve  \xi, \Gamma^\lambda_{\diag,\ve} \nu_\ve  \xi)   + (\nu_\ve  \xi, \Gamma^\lambda_{\off,\ve} \nu_\ve  \xi)  \big) \\
= &  \frac{1}{4\pi}(\xi,  \nu_\ve^*  \Gamma_\ve^\lambda \nu_\ve\xi)
\end{aligned}
\]
which concludes the proof of identity \eqref{BveGave}. 

\n
We are left to prove the second part of the lemma.  We  use Lemma \ref{l:Phive}  by noticing the identity $\widetilde\Phi_\ve^\lambda ( \xi) =  \frac{1}{4\pi} \Phi_\ve^\lambda ( \nu_\ve\xi)$ with $\Phi_\ve^\lambda (\xi)$ defined in  \eqref{Phive1}--\eqref{Phive2}. We stress that the identity makes sense for all $\xi\in L^2(\RE^3)$ since $\nu_\ve \xi \in D(\Phi_\ve^\lambda)$ (recall the remark after  \eqref{1922}).

\n
Then, from Lemma \ref{l:Phive} we obtain 
\begin{equation} \label{LBwidetildePhi}
\widetilde\Phi_\ve^\lambda ( \xi)  \geq \frac{1}{4\pi}  c_0 \bigg(\nu_\ve \xi,\Big(\ID + \frac1{|\cdot|}\Big)\nu_\ve\xi\bigg) = \frac{1}{4\pi} c_0  \int d\y \frac{1+ \frac{1}{y}}{1 + \frac\ve{\ell} \Big(\beta + \frac{\gamma}{y} \theta(y)\Big)}|\xi(\y)|^2 \geq c\,  \|\xi \|^2, 
\end{equation}
where we used the inequality 
\[
\frac{1+ \frac{1}{y}}{1 + \frac\ve{\ell} \Big(\beta + \frac{\gamma}{y} \theta (y)\Big)} 
 \geq \,  \min 
\left\{\f{2}{3} , \f{2|\beta|}{\gamma} \right\} .
\]
To see that the latter inequality holds true, recall that we are assuming $\ve < \ell/(2\|a\|_{L^\infty})$ so that $0<1 + \frac\ve{\ell} \beta + \frac\ve{\ell} \frac{\gamma}{y} \theta (y) = 1 + \frac\ve{\ell} a(y) + \frac\ve{\ell} \frac{\gamma}{y}  \leq \frac32+ \frac\ve{\ell} \frac{\gamma}{y}  $, and 
\[
\frac{1+ \frac{1}{y}}{1 + \frac\ve{\ell} \Big(\beta + \frac{\gamma}{y} \theta (y)\Big)} 
\geq  
\frac{1+ \frac{1}{y}}{\frac32 + \frac\ve{\ell} \frac{\gamma}{y} } = 2\,  \f{y+1}{3y + 2 \frac{\ve}{\ell}\gamma} \geq \,  \min 
\left\{\f{2}{3} , \f{\ell}{\ve}\frac{1}{\gamma} \right\} \geq \min 
\left\{\f{2}{3} , \frac{2\|a\|_{L^\infty}}{\gamma} \right\} .
\]
Since $\widetilde \Phi_\ve^\lambda$ is the quadratic form associated with $\frac{1}{4\pi}\frac{\ell}{\ve}- B_{\ve}^\lambda $, the bound \eqref{LBwidetildePhi} implies that $\frac{1}{4\pi}\frac{\ell}{\ve}- B_{\ve}^\lambda $ is invertible, with an inverse bounded by $1/c$, and this concludes the proof. 

\end{proof}

\n
The last preparatory step is the definition of the bounded operator 
\[ 
A_\ve^\lambda  : L^2(\RE^3)\to  L^2(\RE^6) \qquad  A_\ve^\lambda  := {4\pi}R_0^\lambda  \big(|\chi_\ve\rangle  \otimes \nu_\ve \big) \qquad \lambda >0.
\] 
In Fourier transform 
\[
(\hat A_\ve^\lambda \hat\xi)(\k,\p)  = 4\pi \frac{\hat \chi(\ve k)}{k^2 +\frac34 p^2 +\lambda} \widehat{\nu_\ve \xi}(\p) .
\]
Hence (recall that $|\hat \chi(k)| \leq  (2\pi)^{-3/2}$ so that $ 4\pi |\hat \chi(k)| \leq \sqrt{2/\pi}$)
\begin{equation}\label{1502}
\big|(\hat A_\ve^\lambda \hat\xi)(\k,\p)\big| \leq \sqrt{\frac2\pi} \frac{1}{k^2 + \frac34 p^2 +\lambda} \widehat{\nu_\ve \xi}(\p),
\end{equation}
$\| A_\ve^\lambda \xi \| \leq \|G^\lambda\|_{\Bou(L^2(\RE^3),L^2(\RE^6))}\|\nu_\ve\xi\|  \leq \sqrt2 \|G^\lambda\|_{\Bou(L^2(\RE^3),L^2(\RE^6))}\|\xi\| $, and $A_\ve^\lambda$ is uniformly bounded in $\ve$.

\n
In what follows, we will use the identity 
\[
|\chi_\ve\rangle \langle \chi_\ve|   \otimes g_\ve  = - 4\pi \frac\ve\ell \big(|\chi_\ve\rangle  \otimes \nu_\ve \big)\big(  \langle \chi_\ve|   \otimes  \nu_\ve^* \big)
\]
and write 
\[
H_\ve = H_0   - 4\pi \frac\ve\ell  \sum_{j=0}^2 S^j \big(|\chi_\ve\rangle  \otimes \nu_\ve \big)\big(  \langle \chi_\ve|   \otimes  \nu_\ve^* \big)  {S^j}^*.
\]

\n
We are now ready to formulate and prove the main result of this section.

\begin{theorem}\label{t:KKK}Assume \eqref{cod} and  $\gamma_0$, $\lambda_1$, and $\ve_0$  as in  Lemma  \ref{l:apriori}. Then, for all $\gamma>\gamma_0$, $0<\ve<\ve_0$,  and $\lambda>\lambda_1$, the operator $H_\ve +\lambda$ has a bounded inverse in $L^2_{\textup{sym}}(\RE^6)$. Moreover, denoting its inverse by $R_\ve^\lambda$, one has the Konno--Kuroda formula 
\begin{equation}\label{Rve}
R_\ve^\lambda = 
 R_0^\lambda   +\frac{1}{4\pi}  \sum_{j=0}^2 S^j A_\ve^\lambda  ( \nu_\ve^*  \Gamma_\ve^\lambda \nu_\ve )^{-1}  A_\ve^{\lambda*}. 
\end{equation}
\end{theorem}
\begin{remark}As a matter of fact, see, e.g., \cite[Theorem B.1]{GHL}, the Konno--Kuroda formula holds true for all the complex $\lambda$ such that $-\lambda \in \rho(H_0) \cap \rho(H_\ve)$. The relevant information of Th. \ref{t:KKK} is that there exists a real $\lambda_1$, independent from $\ve$, such that $R_\ve^\lambda$ is a well defined bounded operator for all $\lambda >\lambda_1$. This is equivalent to the lower bound $\inf \sigma(H_\ve) \geq - \lambda_1$. 
\end{remark}
\begin{proof}[Proof of Theorem \ref{t:KKK}]
The  action of $H_\ve$  on a (symmetric) wavefunction in its domain  is given by 
\[
H_\ve \psi = H_0 \psi + \sum_{j=0}^2 S^j  \big(|\chi_\ve\rangle \langle \chi_\ve|   \otimes g_\ve \big) \psi \qquad \forall \psi \in \D(H_\ve).
\]
We describe how to obtain formula \eqref{Rve}. For a given function $\phi \in L^2_{\sym}(\RE^6)$, and $\lambda$ large enough,  assume that  $ \psi_\ve\in \D(H_\ve)\subset L^2_{\sym}(\RE^6)$  is a solution in of the equation 
\[
(H_\ve +\lambda) \psi_\ve = \phi . 
\]
The latter gives 
\[
 (H_0+\lambda)\psi_\ve =  \phi - \sum_{j=0}^2 S^j  \big(|\chi_\ve\rangle \langle \chi_\ve|   \otimes g_\ve \big) \psi_\ve, 
\]
and, recalling that $R_0^\lambda =(H_0+\lambda)^{-1}$ is a well defined bounded operator, 
\[
\psi_\ve = R_0^\lambda  \phi - \sum_{j=0}^2 S^j R_0^\lambda \big(|\chi_\ve\rangle \langle \chi_\ve|   \otimes g_\ve \big) \psi_\ve,
\]
where we used the fact that $R_0^\lambda $ and $S$ commute. Hence, 
\begin{equation}\label{psive}
 \psi_\ve = R_0^\lambda  \phi +4\pi \frac\ve\ell  \sum_{j=0}^2 S^j R_0^\lambda \big(|\chi_\ve\rangle  \otimes \nu_\ve \big)\big(  \langle \chi_\ve|   \otimes \nu_\ve^* \big)\psi_\ve.
\end{equation}
Set 
\[ 
h_\ve := \big(  \langle \chi_\ve|   \otimes  \nu_\ve^* \big)\psi_\ve ,
\]
  and rewrite  \eqref{psive} as
\begin{equation}\label{psive2}
 \psi_\ve = R_0^\lambda  \phi + \frac\ve\ell  \sum_{j=0}^2 S^j A_\ve^\lambda h_\ve .
\end{equation}
We want to obtain a formula for $h_\ve$. To this aim,  apply the operator $\big(  \langle \chi_\ve|   \otimes \nu_\ve^* \big)$ to (the left of) identity \eqref{psive}. By simple algebraic manipulations, it follows that 
\[
\Big(\ID -4\pi   \frac\ve\ell \sum_{j=0}^2\big(  \langle \chi_\ve|   \otimes \nu_\ve^*  \big) S^j R_0^\lambda   \big(|\chi_\ve\rangle  \otimes \nu_\ve \big) \Big)h_\ve
= 
\big(  \langle \chi_\ve|   \otimes \nu_\ve^* \big) R_0^\lambda \phi.  
\]
By Lemma \ref{l:apriori}, the operator at the l.h.s. is invertible and 
\begin{equation}\label{hve2}
h_\ve =\frac{1}{4\pi}\frac{\ell}{\ve} \Big(\frac{1}{4\pi}\frac{\ell}{\ve}  -  B_\ve^\lambda \Big)^{-1} \big(  \langle \chi_\ve|   \otimes  \nu_\ve^*  \big) R_0^\lambda  \phi
=\frac{1}{4\pi} \frac{\ell}{\ve} \big(\nu_\ve^* \Gamma^\lambda_\ve \nu_\ve \big)^{-1} A^{\lambda*}_\ve  \phi,
\end{equation}
where we used the identity 
\[
A^{\lambda*}_\ve   =4\pi \big(  \langle \chi_\ve|   \otimes  \nu_\ve^*  \big) R_0^\lambda .
\]

\n
Using the identity   \eqref{hve2} in   \eqref{psive2} we obtain the formula 
\[
 \psi_\ve = 
 R_0^\lambda  \phi +\frac{1}{4\pi}  \sum_{j=0}^2 S^j A_\ve^\lambda  ( \nu_\ve^*  \Gamma_\ve^\lambda \nu_\ve )^{-1}  A_\ve^{\lambda*}\phi.
\]
Now, as suggested from the formula above, one can define $R_\ve^\lambda$ as in  \eqref{Rve} and show by a straightforward calculation  that $(H_\ve+\lambda)R_\ve^\lambda =\ID$ on $L^2(\RE^6)$ and $ R_\ve^\lambda (H_\ve+\lambda) = \ID$ on $\D(H_\ve)$, from which  it follows that $R_\ve^\lambda = (H_\ve +\lambda)^{-1}$.  

\end{proof}

\vs
\section{Norm Resolvent Convergence \label{s:src}}
In this section we prove that $H_\ve$ converges to $H$ in the norm resolvent sense and we give an estimate of the rate of convergence.

\begin{proof}[Proof of Theorem \ref{t:main}]It is enough to prove the statement for some fixed  $z = -\lambda < -\lambdamin$, then it holds
 true for a generic $z\in\CO\backslash   [-\lambdamin,\infty)$ by analytic continuation.

\n
Since 
\[
R_\ve^\lambda  - R^\lambda = \frac{1}{4\pi} \big(A_\ve^\lambda  ( \nu_\ve^*  \Gamma_\ve^\lambda \nu_\ve )^{-1}  A_\ve^{\lambda*} - G^\lambda (\Gamma^\la)^{-1}  G^{\lambda*}\big), 
\]
 we need to show that for some $\lambda$ large enough there holds 
 \[
 \|(A_\ve^\lambda  ( \nu_\ve^*  \Gamma_\ve^\lambda \nu_\ve )^{-1}  A_\ve^{\lambda*} - G^\lambda (\Gamma^\la)^{-1}  G^{\lambda*}) \phi\|  \leq c\, \ve^{\de} \|\phi\|\qquad \forall \phi\in L^2_{\sym}(\RE^6). 
\]
Without loss of generality, here as well as in the proof of the Lemmata \ref{l:AveCveGG*} and \ref{l:GammaveGamma},  we can assume $\lambda>1$. All the generic constants denoted by $c$ are independent from $\lambda$ for $\lambda>1$.  
 We start with the trivial identity 
\begin{equation*}
\begin{aligned}
A_\ve^\lambda  ( \nu_\ve^*  \Gamma_\ve^\lambda \nu_\ve )^{-1}  A_\ve^{\lambda*} - G^\lambda (\Gamma^\la)^{-1}  G^{\lambda*} =  & 
(A_\ve^\lambda -  G^\lambda ) ( \nu_\ve^*  \Gamma_\ve^\lambda \nu_\ve )^{-1}  A_\ve^{\lambda*}  \\
& + G^\lambda ( \nu_\ve^*  \Gamma_\ve^\lambda \nu_\ve )^{-1}  (A_\ve^{\lambda*} - G^{\lambda*}) \\ 
&+ G^\lambda \big(  ( \nu_\ve^*  \Gamma_\ve^\lambda \nu_\ve )^{-1} -   (\Gamma^\la)^{-1} \big) G^{\lambda*} . 
\end{aligned}
\end{equation*}
By Lemma \ref{l:AveCveGG*}, and since   $( \nu_\ve^*  \Gamma_\ve^\lambda \nu_\ve )^{-1}$ and $A_\ve^\lambda$ are uniformly bounded in $\ve$  (see Lemma \ref{l:apriori} and  the remark after  \eqref{1502}), we infer 
\[
\|(A_\ve^\lambda -  G^\lambda ) ( \nu_\ve^*  \Gamma_\ve^\lambda \nu_\ve )^{-1}  A_\ve^{\lambda*} \|_{\Bou(L^2(\RE^6))} \leq c \, \ve^{\de} \qquad 0<\de<1/2
\]
 and
\[
\|G^\lambda ( \nu_\ve^*  \Gamma_\ve^\lambda \nu_\ve )^{-1}  (A_\ve^{\lambda*} - G^{\lambda*})  \|_{\Bou(L^2(\RE^6))} \leq c\, \ve^{\de} \qquad 0<\de<1/2.
\]
We are left  to prove that 
\begin{equation}\label{stay}
 \|(G^\lambda \big(  ( \nu_\ve^*  \Gamma_\ve^\lambda \nu_\ve )^{-1} -   (\Gamma^\la)^{-1} \big) G^{\lambda*} \phi\|  \leq c\, \ve^{\de} \|\phi\|.
\end{equation}
We note the identity 
\[
  ( \nu_\ve^*  \Gamma_\ve^\lambda \nu_\ve )^{-1} -   (\Gamma^\la)^{-1} = -   ( \nu_\ve^*  \Gamma_\ve^\lambda \nu_\ve )^{-1} \big(\nu_\ve^*  \Gamma_\ve^\lambda \nu_\ve - \Gamma^\lambda \big) (\Gamma^\la)^{-1}. 
\]
So that,  taking into account the fact that $G^\lambda$ is bounded and $( \nu_\ve^*  \Gamma_\ve^\lambda \nu_\ve )^{-1}$ is uniformly bounded in $\ve$, we infer that  \eqref{stay} is a consequence of 
\[
 \big\| \big(\nu_\ve^*  \Gamma_\ve^\lambda \nu_\ve - \Gamma^\lambda \big)(\Gamma^\la)^{-1} G^{\lambda*}\phi\big\|  \leq c \ve^{\de} \|\phi\|
\]
which holds true by Lemma \ref{l:GammaveGamma}. 

\end{proof}

\vs

\begin{lemma} \label{l:AveCveGG*} There exists $c>0$ such that 
\[
 \|A_\ve^\lambda -  G^\lambda\|_{\Bou(L^2(\RE^3),L^2(\RE^6))}  \leq c\, \ve^{\delta} 
\]
for all $0<\delta < 1/2$  and for all $\lambda >1$.  
\end{lemma}
\begin{proof}
We use the identity 
\begin{equation}\label{idnuve}
\nu_\ve =  \ID +(\nu_\ve  - \ID)
\end{equation}
and write 
\[
A_\ve^\lambda = A_{1,\ve}^\lambda + A_{2,\ve}^\lambda 
\]
with 
\[
 A_{1,\ve}^\lambda  =  4\pi R_0^\lambda  \big(|\chi_\ve\rangle  \otimes \ID  \big) \qquad \text{and} \qquad  A_{2,\ve}^\lambda =    4\pi  R_0^\lambda  \big(|\chi_\ve\rangle  \otimes  (\nu_\ve-\ID)    \big) . 
\]
For all $\xi\in L^2(\RE^3)$, one has  
\[
\|(A_{1,\ve}^\lambda - G^\lambda)\xi  \|^2 
= (4\pi)^2 \int d\k\,d\p \frac{|\hat \chi(\ve k) - \hat\chi(0)|^2}{(k^2 +\frac34 p^2 +\lambda)^2} |\hat \xi(\p)|^2 
\leq \ve  \int d\k \frac{|\hat \chi( k) - \hat \chi(0)|^2}{k^4}   \|\xi\|^2 = c \,  \ve \|\xi\|^2
\] 
Concerning $A_{2,\ve}^\la$, we note that one has 
\[\begin{aligned}
\|A_{2,\ve}^\lambda \xi  \|^2 = &(4\pi)^2 \int d\k\,d\p \frac{|\hat \chi(\ve k) |^2}{(k^2 +\frac34 p^2 +\lambda)^2} \Big|\Fou\Big(( \nu_\ve -1)    \big)\xi\Big)(\p) \Big|^2 \\ 
 \leq  & (4\pi)^2\big\|(-\Delta +1)^{-\frac14} ( \nu_\ve -1)   \xi\big\|^2 \,  \sup_{p>0} \int d\k \frac{|\hat \chi(\ve k) |^2 (p^2+1)^{\frac12}}{(k^2 +\frac34 p^2 +\lambda)^2}  .
\end{aligned}
\]
We have
\[
 \nu_\ve(y)-1= \f{
1-\sqrt{  1 + \frac{\ve}{\ell}a(y)}    - i\sqrt{ \frac{\ve}{ \ell} \frac{\gamma}{y} } 
}{
\sqrt{  1 + \frac{\ve}{\ell} a(y) } + i\sqrt{ \frac{\ve}{ \ell} \frac{\gamma}{y} }.
}
\]
Hence, taking into account Remark \ref{a:chi} and the inequality $\big|1- (1+s)^\frac12\big| \leq |s| $ for all $s\geq -1 $,  one has 
\[
|  \nu_\ve(y)-1| \leq c \sqrt\ve\lf(1+\f{1}{|\cdot|^\f12}\ri).
\]
Setting  $O := |\cdot|^{-1/2} (-\Delta +1)^{-\frac14} $, inequality \eqref{3.8a} reads $\| Of \|^2 \leq  \frac{\pi}{2} \|f\|^2. $ which implies also
 \[
 \Big\|(-\Delta +1)^{-\frac14}  \frac{1}{|\cdot|^\frac12} f \Big\|  
\leq 
\|O^* \|_{\Bou(L^2(\RE^3))}  \|f\|  \leq  \Big(\frac{\pi}{2}\Big)^{\frac12} \|f\|,
 \]
 and we arrive at  
\[
\big\|(-\Delta +1)^{-\frac14} (\nu_\ve-1)  \xi\big\|
\leq C \ve^\frac12 \|\xi\|. 
\]
Moreover
\[
 \sup_{p>0} \int d\k \frac{|\hat \chi(\ve k) |^2 (p^2+1)^{\frac12}}{(k^2 +\frac34 p^2 +\lambda)^2}  
\leq \frac{1}{\ve^{1-\delta}}  \;  \sup_{p>0} \frac{(p^2+1)^{\frac12}}{(\frac34 p^2 +\lambda)^{1-\frac\delta2 }}\int d\k \frac{|\hat \chi(k) |^2}{k^{2+\delta}}  \leq \frac{c}{\ve^{1-\de}} ,
\]
so that 
\[
\|A_{2,\ve}^\lambda \xi  \| \leq c\, \ve^{\delta} \|\xi\| \qquad  \;0 < \delta < 1/2.
\]
Hence, 
\[
 \|(A_\ve^\lambda -  G^\lambda) \xi \| \leq 
  \|(A_{1,\ve}^\lambda -  G^\lambda) \xi \| +  
   \|A_{2,\ve}^\lambda \xi \| \leq  c\,  \ve^{\delta} \|\xi\|,
\]
which concludes the proof of the lemma. 

\end{proof}

\vs

\n
Clearly we have also
\[
 \|A_\ve^{\lambda*} - G^{\lambda*}\|_{\Bou(L^2(\RE^6),L^2(\RE^3))}  \leq c\,  \ve^{\delta} 
\]
for all $0<\delta < 1/2$  and for all $\lambda >1$.  
\begin{lemma}\label{l:GammaveGamma} Assume \eqref{cos} and $\lambda>1$.  For any $\phi \in L^2_{\sym}(\RE)$ there holds true
\[
 \big\| \big(\nu_\ve^*  \Gamma_\ve^\lambda \nu_\ve - \Gamma^\lambda \big)(\Gamma^\la)^{-1} G^{\lambda*}\phi   \big\| \leq c\, \ve^{\de}  \|\phi\|
\qquad 0<\de<1/2.
\]
\end{lemma}
\begin{proof}
In this proof we set
\[
\xi := (\Gamma^\la)^{-1} G^{\lambda*}\phi,
\]
and we know by Proposition \ref{finefine} that
\begin{equation}\label{together}
\| \xi\|_{H^{\frac32}} \leq c \,  \|\phi\|\qquad \text{and}\qquad \| \Gamma^\lambda  \xi\|_{H^{\frac12}}  =  \| G^{\lambda*}\phi\|_{H^{\frac12}} \leq  c\,  \|\phi\|. 
\end{equation}
\n
We use again the identity \eqref{idnuve} and obtain 
\[ 
\nu_\ve^*  \Gamma_\ve^\lambda \nu_\ve - \Gamma^\lambda = 
 \Gamma_\ve^\lambda  - \Gamma^\lambda +  ( \nu_\ve^{*} -\ID) \Gamma_\ve^\lambda + \Gamma_\ve^\lambda  ( \nu_\ve-\ID) + ( \nu_\ve^{*}-\ID ) \Gamma_\ve^\lambda  ( \nu_\ve -\ID) .
\]
We are going to prove that for $0<\de<1/2$, we have:
\begin{align}\label{strawberry1}
&\big\|\big(\Gamma_{ \ve}^\lambda  - \Gamma^\lambda\big) \xi\big\|\leq c \, \ve^{\de} \|\phi\|,  \\
&\big\| ( \nu_\ve^{*} -\ID) \Gamma_\ve^\lambda \xi\big\| \leq c \, \ve^{\de}\|\phi\| ,\label{strawberry2a} \\
&\big\|  \Gamma_\ve^\lambda ( \nu_\ve -\ID)   \xi\big\|\leq c \, \ve^{\de}\|\phi\| , \label{strawberry2b} \\
& \label{strawberry3}
 \big\| ( \nu_\ve^{*} -\ID) \Gamma_\ve^\lambda ( \nu_\ve -\ID)  \xi\big\|\leq c \, \ve^{\de}\|\phi\| .
\end{align}
Let us prove \eqref{strawberry1}. We have that 
\[ 
\Gamma_\ve^\lambda  - \Gamma^\lambda  =  \Gamma_{\diag, \ve}^\lambda  - \Gamma_{\diag}^\lambda +  \Gamma_{\off,\ve}^\lambda  - \Gamma_{\off}^\lambda.  
\] 
Concerning the first couple of operators, using \eqref{hummus}, we have for $0<\de<1/2$ 
\begin{align*}
\big\|\big(\Gamma_{\diag, \ve}^\lambda  - \Gamma_{\diag}^\lambda\big) \xi\big\|^2 &
= \int d\p\, \bigg(\frac34  p^2 +\lambda\bigg) \bigg( r\Big(\ve \sqrt{\frac34 p^2 +\lambda}\Big) - 1 \bigg)^2|\hat \xi(\p)|^2  \\
&\leq c\, \ve^{2\de} \int d\p\, \bigg(\frac34  p^2 +\lambda\bigg)^{1+\de}|\hat \xi(\p)|^2 \leq c\, \ve^{2\de} \|\phi\|^2. 
\end{align*}

\n
On the other hand, using \eqref{hummus} and the boundedness of $\hat \chi$, we have
\[ \begin{aligned}
\lf |\big(\hat \Gamma_{\off, \ve}^\lambda  - \hat \Gamma_{\off}^\lambda\big)\hat  \xi (\p) \ri|
&=  8\pi \lf| \int \,d\q\, \frac{\hat \chi\big(\ve\big|\p +\frac12 \q\big|\big) \hat \chi\big(\ve\big|\frac12 \p + \q\big|\big) -\big(\hat\chi(0)\big)^2 }{p^2 + q^2 + \p \cdot \q +\lambda } \hat \xi(\q) \ri| \\
&\leq 8\pi \int \,d\q\, \frac{\lf| \hat \chi\big(\ve\big|\p +\frac12 \q\big|\big) \hat \chi\big(\ve\big|\frac12 \p + \q\big|\big) -\big(\hat\chi(0)\big)^2\ri| }{p^2 + q^2 + \p \cdot \q +\lambda } | \hat \xi(\q)| \\
& \leq c \, \ve^{\de}  \int \,d\q\, \frac{\lf( p^2 +q^2\ri)^{\de /2} }{p^2 + q^2 + \p \cdot \q +\lambda } | \hat \xi(\q)| \qquad 0<\de < 1/2.
\end{aligned}\]

\n
Due to \eqref{together}, it is sufficient to prove that 
\[
T(\p, \q)=  \frac{\lf( p^2 +q^2\ri)^{\de} }{(p^2 + q^2 + \p \cdot \q +\lambda)(q^2+1)^{3/4} }
\]
is the  integral kernel of an $L^2$-bounded operator for $ 0<\de < 1/2$. If we put $f(\p)= (p^2+1)^{-3/4}$, it is straightforward to prove that for $0<\de<1/2$ we have
\[
\int T(\p, \q) \,f(\q)\, d\q \leq c_1 \, f(\p) \qquad \qquad \int T(\p, \q) \,f(\p)\, d\p \leq c_2 \, f(\q) ,
\]
then the claim follows from Schur's test.

\n
Let us prove \eqref{strawberry2a}. We claim that  for all $s\in[0,1/2)$ there exists $C>0$ such that  
\begin{equation}\label{easy}
\| \Gamma_\ve^\lambda \xi \|_{H^s}\leq c \|\xi\|_{H^{s+1}} \leq c   \|\phi\|. 
\end{equation}
Due to Remarks \ref{afa} and  \ref{owari}, it is sufficient to control each term in $ \Gamma_\ve^\lambda \xi$ with corresponding limit uniformly in $\ve$.
It holds true for  ${ \Gamma}^{ \la}_{\text{diag},\ve}\xi $, since $0\leq r(s) \leq 1$; see also \eqref{together}. 
It holds true for  ${\Gamma}^{ \la}_{\text{off},\ve}\xi $ since 
\[
\big|\hat \Gamma_{\off, \ve}^\lambda  \hat \xi (\p) \big| \leq
  \int \,d\q\, \frac{1}{p^2 + q^2 + \p \cdot \q +\lambda }\big| \hat \xi(\q) \big| = 
\hat \Gamma_{\off}^\lambda | \hat \xi | (\p).
\]
Therefore, estimate \eqref{easy} holds true. We have
\[
 \nu_\ve^\ast (y)-1= \f{
1-\sqrt{  1 + \frac{\ve}{\ell} a(y) } +i\sqrt{ \frac{\ve}{ \ell} \frac{\gamma }{y} } 
}{
\sqrt{  1 + \frac{\ve}{\ell} a(y) } - i\sqrt{ \frac{\ve}{ \ell} \frac{\gamma }{y} }
}=
f_1 (y) + f_2 (y).
\]
with
\[ 
f_1 (y)=\f{
1-\sqrt{  1 + \frac{\ve}{\ell} a(y) }
}{
\sqrt{  1 + \frac{\ve}{\ell} a(y) } - i\sqrt{ \frac{\ve}{ \ell} \frac{\gamma }{y} }
}
\qquad
f_2 (y)=\f{
i\sqrt{ \frac{\ve}{ \ell} \frac{\gamma }{y} } 
}{
\sqrt{  1 + \frac{\ve}{\ell} a(y) } - i\sqrt{ \frac{\ve}{ \ell} \frac{\gamma }{y} }
}
\]
and, due to Remark \ref{a:chi}, we have the estimates
\be \label{tenda2}
|f_1 (y) |\leq c \, \sqrt{\ve} \qquad |f_2 (y) |\leq c\, \sqrt{\ve}  \frac{ 1}{\sqrt{ y + 2 {\ve}\ga / \ell} }.
\ee
Moreover, with straightforward calculations, one has  for $p>6$
\be \label{tenda}
\| f_2 \|_{L^p }^p \leq c\int_0^\ve y^2 \, dy + c\, \ve^{p/2} \int_\ve^\infty \f{1}{y^{p/2-2}} dy= c \, \ve^3. 
\ee
By Sobolev embedding and \eqref{easy}, we have$ \|  \Gamma_\ve^\lambda \xi \|_{L^q}\leq C \|\phi\|$ for $\forall q \in [2,3)$.
Using H\"older inequality, \eqref{tenda2} and \eqref{tenda}, we have
\[
\big\| ( \nu_\ve^{*} -\ID) \Gamma_\ve^\lambda \xi\big\| \leq 
\big\| f_1 \Gamma_\ve^\lambda \xi\big\| +
\big\| f_2 \Gamma_\ve^\lambda \xi\big\|\leq c (\sqrt\ve + \ve^{3/p} ) \| \phi \| \qquad p>6,
\]
and \eqref{strawberry2a} is proved.

\n
Let us prove \eqref{strawberry2b}. 
Due to estimate \eqref{together}, it is sufficient estimate $\| (\nu_\ve-1)\xi\|_{H^1}$.
We start from
\be \label{hajimari}
\big\| ( \nu_\ve-1 )  \xi\big\|_{H^{1}} \leq 
\big\| ( \nu_\ve-1) \xi\big\|
+\big\|  (\nu_\ve -1)  \nabla\xi\big\|
+
\big\| \nu_\ve'  \xi\big\|. 
\ee
The second term in \eqref{hajimari} can be estimated as \eqref{strawberry2a}; also the first term, even if it is more regular, can be estimated in the same way.
We discuss the third term. First notice that
\[
\nu_\ve^\prime (y) = - \frac{1}{ \lf(\sqrt{  1 + \frac{\ve}{\ell} a(y) } +i\sqrt{ \frac{\ve}{ \ell} \frac{\gamma}{y} }\ri)^2} 
\lf(\frac{\ve}{\ell} \frac{a'(y)}{2\sqrt{1+\frac\ve\ell a(y)}} -i \sqrt{\frac\ve \ell \gamma} \frac1{2y^\frac32} \ri) = f_3(y) + f_4(y)
\]
with
\[
 f_3(y) =  - \frac{\ve}{\ell} \frac{a'(y)}{2  \lf(\sqrt{  1 + \frac{\ve}{\ell} a(y) } +i\sqrt{ \frac{\ve}{ \ell} \frac{\gamma}{y} }\ri)^2 \sqrt{1+\frac\ve\ell a(y)}} 
\qquad \qquad
 f_4(y) =  \frac{i \sqrt{\frac\ve \ell \gamma} }{ \lf(\sqrt{  1 + \frac{\ve}{\ell} a(y) } +i\sqrt{ \frac{\ve}{ \ell} \frac{\gamma}{y} }\ri)^2} 
 \frac1{2y^\frac32} . 
\]
Taking into account that 
\[
a'(y) = - \frac{\theta(y) - 1 - y \theta'(y) }{y^2} = \frac1{y^2} \int_0^y s \theta''(s) ds\] is bounded, we have $\| f_3\|_{L^\infty} \leq c \, \ve$. Moreover we also have for  $2<p <6 $
\[
\| f_4 \|^{p}_{L^{p}} \leq c\, \ve^{-p /2} \int_0^\ve y^{2- p /2} dy + \ve^{p /2} \int_\ve^\infty \f{1}{y^{3 p /2-2}}dy = c\, \ve^{3-p}.
\]
By Sobolev embedding $\xi \in L^q$ for $2\leq q <\infty$, then using H\"older inequality, we have
\[
\big\| \nu_\ve'  \xi\big\| \leq c \| f_3 \xi \|+ \| f_4 \xi \| \leq \|f_3\|_{L^{\infty} } \| \xi\| +\|f_4\|_{L^{p}} \| \xi\|_{L^{q}}
\leq c (\ve +\ve^{3/ p -1} ) \|\xi\|_{H^{3/2} }.
\]
where $p^{-1}+q^{-1}= 1/2$, and \eqref{strawberry2b} is proved.

\n
Let us prove \eqref{strawberry3}.  
This estimate immediately reduces to \eqref{strawberry2b} since $\| \nu_\ve^\ast-1\|_{L^\infty} \leq c $ uniformly in $\ve$.

\end{proof}

\vs

\appendix

\section{Regularity of the Charge}\label{secApp}

In this appendix we characterize $\D (\Ga^\la)$ and we prove a regularity result for the charge associated with $\psi\in\D(H)$.

\subsection{Domain of $\Ga^\la$} \label{app:A1}
Here, we prove that $\D(\Gamma^{\la}) = H^1(\RE^3)$. 
\begin{remark}
By Proposition \ref{prop:lower}, the spectrum of $\Gamma^{\lambda}$ is contained in $[c_0, \infty)$, $c_0 >0$. Therefore, $(\Gamma^{\lambda} )^{-1}$ exists and it is a bounded operator in $L^2(\RE^3)$ with norm less than $c_0^{-1}$. That is, for any  $\genf\in L^2 (\RE^3)$ there exists $ \xi\in L^2(\RE^3)$ solution of the equation $\Ga^\la \xi = \genf$ and $\|\xi\|\leq c_0^{-1} \|\genf\|$.  
\end{remark}
\n
We want to prove that $\xi\in H^1(\RE^3)$,  that is
$(\Gamma^{\lambda} )^{-1}\in \B (L^2(\RE^3),H^1(\RE^3)).$

\n
For the sake of notation, we introduce the operator $T$ defined as: 
\begin{equation}\label{T}
\hat T \hat{\xi}(\p): =  \f{\sqrt 3}{2} p\, \hat{\xi}(\p) 
-\frac1{\pi^2} \int_{\RE^3} \!\! d\q\,  \frac{ \hat{\xi}(\q)}{p^2+q^2+\p \cdot \q}
 + \f{\gamma}{2 \pi^2}  \int_{\RE^3} \!\!  d\q\,  \frac{ \hat{\xi}(\q)}{|\p - \q|^2 }.
\end{equation}
Let 
\[
\ga_c^\ast =  \frac74 \sqrt 3 - 2\simeq 1.031.
\]
\begin{proposition} \label{p:fine} Assume \eqref{cod}  and $\ga>\ga_c^\ast$, let $\genf\in L^2(\RE^3)$ and let $\xi$ be the solution of 
\beq \label{dominio}
\Ga^\la \xi = \genf ,
\eeq
then $\xi \in H^1$ and $ \| \xi\|_{H^1}\leq c \| \genf\|_{ L^2(\RE^3)}$.
\end{proposition}
\begin{proof}
Set $\genf^\lambda := \genf - (\Gamma^\lambda - T) \xi$, and recast  \eqref{dominio} as 
\beq \label{dominio1}
T\xi = \genf^\la.
\eeq
By  \eqref{opga} and \eqref{T}, there follows 
\begin{align} \label{effela1}
\hat{\genf}^\la(\p)=\hat{\genf}(\p) -(\widehat{a \,\xi} )(\p) - \f{2 \la}{\sqrt{3}} \, \f{\hat{\xi}(\p)}{p + \sqrt{p^2 + 4 \la /3}} 
- \f{\la}{\pi^2} \!\int_{\RE^3}\!\! d\q \, \f{ \hat{\xi} (\q)}{(p^2 + q^2 + \p \cdot \q + \la)(p^2 + q^2 + \p \cdot \q)} .
\end{align}
We start by noticing that $\genf^\la \in L^2(\RE^3)$. To see that this is indeed the case recall that by our previous remark $\xi\in L^2(\RE^3)$ and notice that all the terms at the  r.h.s. of identity \eqref{effela1} are in $L^2$. To convince oneself that this is true also for the integral term (for all the others it is obvious), it is sufficient to notice that the integral kernel is an Hilbert--Schmidt operator. To this aim, one can use the inequality 
\[
\int_{\RE^3} d\p d\q \, \f{ 1}{(p^2 + q^2 + \p \cdot \q + \la)^2(p^2 + q^2 + \p \cdot \q)^2} \leq 16
\int_{\RE^3} d\p d\q \, \f{ 1}{(p^2 + q^2 + 2\la)^2(p^2 + q^2 )^2},  
\]
and introducing polar coordinates in $\RE^6$ verify that the latter integral is finite.

\n
To conclude the proof of the proposition we need to show that 
\[
\|\nabla \xi \| \leq c \|\genf^\lambda\|.
\] 
Decomposing $\hat \xi$ and $\hat \genf^\lambda$ on the basis of Spherical Harmonics, and setting $p=e^x$,  $\zeta_{\ell m}(x) = e^{5/2 \,x} \hat \xi_{\ell m}(e^x)$, and $h_{\ell m}(x)= e^{3/2 \,x} \hat \genf_{\ell m}(e^x)$ we obtain:
\[
\| \nabla \xi \|^2 = \sum_{\ell m} \int_{\RE^+} |\hat \xi_{\ell m}(p) |^2 p^4\,dp = 
\sum_{\ell m} \int_\RE |\zeta_{\ell m}(x) |^2 \, dx 
\]
and 
\[
\| \genf^\la \|^2 = \sum_{\ell m} \int_{\RE^+} |\genf^\la_{\ell m}(p) |^2 p^2 dp =
\sum_{\ell m} \int_\RE |h_{\ell m}(x) |^2 \, dx;
\]
here, clearly,  $ \hat{\genf}_{\ell m}^\la(p)\in L^2(\RE^+, p^2\,dp)$ for $\ell \in \NA$ and , $m=-\ell, \ldots , \ell\,$. We look for an  inequality  between the $L^2$-norms of the functions $\zeta_{\ell m}$ and $h_{\ell m}$ of the form $\|\zeta_{\ell m} \|\leq c\|h_{\ell m}\|$ with $c$ independent on $\ell$ and $m$. To proceed, we decompose  \eqref{dominio1} on the basis of Spherical Harmonics and obtain:
\beq \label{dominioonda}
\begin{aligned}
\f{\sqrt{3}}{2}\, p \, \hat{\xi}_{lm}(p) 
-\frac{2}{\pi}  \! \int_0^{\infty}\! \!\! dq\, q^2 \, \hat{\xi}_{\ell m} (q) \!\! \int_{-1}^{1} \!\!  & d\nu\, \f{P_\ell(\nu)}{p^2+q^2+p q \nu }
  + \f{\gamma}{ \pi} \! \int_0^{\infty} \! \!\!  dq\, q^2 \hat{\xi}_{\ell m} (q) 
   \!\! \int_{-1}^{1} \!\! d\nu\, \f{P_\ell(\nu)}{p^2 + q^2 -2 p  q \nu } 
   = \hat{\genf}_{\ell m}^\la (p). 
\end{aligned}\eeq
Then,  we multiply the latter equation by $p^{3/2}$ and change variables as above, with  $p=e^x$ and $q=e^y$, to obtain:
\beq\begin{aligned} \label{dominio3}
\f{\sqrt{3}}{2} \, \zeta_{\ell m}(x) - \f{1}{\pi} \int_{\RE}\!\! dy\, \zeta_{\ell m}(y) \,e^{(x-y)/2}\!\!\int_{-1}^{1}\!\!\! d\nu\, 
& \f{P_\ell(\nu)}{\cosh (x-y) + \nu/2}+ \\
&+ \f{\gamma}{2 \pi} \int_{\RE}\!\! dy\, \zeta_{\ell m}(y) \, e^{(x-y)/2}\!\!\int_{-1}^{1}\!\!\! d\nu\,  \f{P_\ell(\nu)}{\cosh (x-y) - \nu} = h_{\ell m}(x).
\end{aligned}\eeq
The latter equation can be seen as a  convolution equation on $L^2 (\RE)$ and discussed by Fourier transform, to this aim we note the identities:  
\[
 - \f{1}{\pi} \int_\RE dx \, e^{-ikx} e^{x/2} \int_{-1}^{1}  d\nu\, 
 \f{P_\ell(\nu)}{\cosh (x) + \nu/2} = S_{\mathrm{off},\ell}\left(k+\frac{i}2\right) \\  
\]
and 
\[ \f{\gamma}{2 \pi}  \int_\RE dx \, e^{-ikx}  e^{x/2}\int_{-1}^{1} d\nu\,  \f{P_\ell(\nu)}{\cosh (x) - \nu} 
= S_{\mathrm{reg},\ell}\left(k+\frac{i}2\right)  
\]
(which hold true because  $S_{\mathrm{off},\ell}(k) $ and $S_{\mathrm{reg},\ell}(k)$, defined in  \eqref{Soffpg12} and \eqref{Sregpg12},  admit an holomorphic extension to the strip $\{ |\Im k| <1\}$); therefore (see  \eqref{Stotale}) 
\begin{equation}\label{mercy}
S_\ell \lf( k +\f{i}{2} \ri) =  \f{\sqrt{3}}{2} \, + S_{\off,\ell}\lf( k +\f{i}{2} \ri)  + S_{\reg,\ell} \lf( k +\f{i}{2} \ri)  .
\end{equation}
Then,  \eqref{dominio3} is equivalent to 
\[
S_\ell \lf( k +\f{i}{2} \ri) \hat \zeta_{\ell m} (k) = \hat h_{\ell m} (k).
\]
To conclude the proof of the proposition,  it is  sufficient to prove that $\lf| S_\ell \lf( k+ \f{i}{2} \ri) \ri|\geq c>0$.
We shall focus on the real part of $S_\ell$ and prove  that 
\begin{equation}\label{mjrtom}
\Re  S_\ell \lf( k+ \f{i}{2} \ri) \geq c>0.
\end{equation}
Starting from \eqref{eq:Soff} and \eqref{eq:Sreg}, with some straightforward calculations
one arrives at
\begin{align}
\Re S_{\mathrm{off},\ell} \lf( k+ \f{i}{2} \ri) =&  - 2\int_{-1}^1dy\,P_\ell(y)\f{ \cosh \Big( \!  k\arccos \big(\frac y2\big)\! \Big) \sin \lf( \f12 \arccos(\frac y2) \ri)}{ \sqrt{1- \frac {y^2}4 } \cosh (\pi k)} \nonumber \\
=&  - \sqrt{2}\int_{-1}^1dy\,P_\ell(y)\f{ \cosh \Big( \!  k\arccos \big(\frac y2\big)\! \Big)}{ \sqrt{1+ \frac {y}2 } \cosh (\pi k)} \label{middle1}
\end{align}
and 
\begin{align}
\Re S_{\mathrm{reg},\ell} \lf( k+ \f{i}{2} \ri) 
=&{\gamma} \int_{-1}^1dy\,P_\ell(y)\frac{\cosh\big(k\arccos(-y)\big)\sin \lf( \f12 \arccos(-y) \ri)}{\sqrt{1-y^2}\cosh(k\pi)} \nonumber
  \\
= &\f{\gamma}{\sqrt2} \int_{-1}^1dy\,P_\ell(y)\frac{\cosh\big(k\arccos(-y)\big)}{\sqrt{1-y}\cosh(k\pi)} \label{middle2}.
\end{align}
We analyze separately the cases $\ell\geq 1$ and $\ell =0$. We start with $\ell \geq 1$. Notice that $\frac{\cosh\big(k\arccos(-y)\big)}{\sqrt{1-y}}$ has a series expansion with positive coefficients. To see that this is indeed the case recall that: (see \cite[(1.112.4)]{GR})
\[
\f{1}{\sqrt{1-y}} = 1 + \sum_{k=0}^\infty \frac{(2k+1)!!}{(2k+2)!!} y^{k+1};
\]
and $\arccos(-y)= \pi/2 +\arcsin(y)$ as well as $\cosh(x)$  have  positive coefficients series. Then,
\begin{equation}\label{rollin}
\frac{\cosh\big(k\arccos(-y)\big)}{\sqrt{1-y}} = \sum_{n=0}^\infty c_n y^n \qquad c_n \geq 0
\end{equation}
and 
\beq \label{pluto}
\Re S_{\mathrm{reg},\ell} \lf( k+ \f{i}{2} \ri) \geq0
\eeq
by Lemma \ref{lemmavecchio}.

\n
Now, we prove
some suitable lower bounds for  $\Re S_{\mathrm{off},\ell}$. By \eqref{rollin}, we infer 
\[
\frac{\cosh\Big( \!  k\arccos \big(\frac y2\big)\! \Big)}{ \sqrt{1+ \frac {y}2 } }  =  \sum_{n=0}^\infty (-1)^n c_n \Big(\frac{y}{2}\Big)^n 
\]
so that we cannot apply directly Lemma \ref{lemmavecchio} to $\Re S_{\mathrm{off},\ell}$. However,  due to the parity properties of the Legendre polynomials, we can apply the lemma for $\ell$ even and $\ell$ odd separately. We start with the analysis of the case $\ell$ odd. We have that 
\begin{align}
\Re S_{\mathrm{off},\ell} \lf( k+ \f{i}{2} \ri) = &   -\frac{\sqrt{2}}{\cosh (\pi k)}  \int_{-1}^1dy\,P_\ell(y)\f{ \cosh \Big( \!  k\arccos \big(\frac y2\big)\! \Big)}{ \sqrt{1+ \frac {y}2} } \nonumber  \\
=&  
\frac{\sqrt{2}}{\cosh (\pi k)}  \int_{-1}^1dy\,P_\ell(y) \sum_{\substack{n=0\\ n \text{ - odd}}}^\infty c_n \Big(\frac{y}{2}\Big)^n  \geq 0  \qquad \ell \text{ - odd} \label{pluto1}
\end{align}
where the latter inequality is a consequence of  Lemma \ref{lemmavecchio}.

\n  On the other hand, for $\ell$ even we have  
\[\begin{aligned}
\Re S_{\mathrm{off},\ell} \lf( k+ \f{i}{2} \ri) = -\frac{\sqrt{2}}{\cosh (\pi k)}  \int_{-1}^1dy\,P_\ell(y) \sum_{\substack{n=0\\ n \text{- even}}}^\infty c_n \Big(\frac{y}{2}\Big)^n  \qquad \ell \text{ - even}.
\end{aligned}\]
Hence, using again  Lemma \ref{lemmavecchio}, we infer
\begin{align}
0\geq & \Re S_{\mathrm{off},\ell} \lf( k+ \f{i}{2} \ri) 
\geq   -\frac{\sqrt{2}}{\cosh (\pi k)}  \int_{-1}^1dy\,P_2(y)\f{ \cosh \Big( \!  k\arccos \big(\frac y2\big)\! \Big)}{ \sqrt{1+ \frac {y}2} } \nonumber\\
=  &  -\frac{1}{\sqrt{2} \cosh (\pi k)}  \int_{-1}^1dy\,(3y^2-1)\f{ \cosh \Big( \!  k\arccos \big(\frac y2\big)\! \Big)}{ \sqrt{1+ \frac {y}2} } \nonumber\\
\geq   &  -\frac{1}{\sqrt{2} \cosh (\pi k)}  \int_{\frac{1}{\sqrt{3}} \leq |y| \leq 1}dy\,(3y^2-1)\f{ \cosh \Big( \!  k\arccos \big(\frac y2\big)\! \Big)}{ \sqrt{1+ \frac {y}2} } \nonumber\\
\geq   &  -\frac{\cosh (\frac23 \pi k)}{ \cosh (\pi k)}  \int_{\frac{1}{\sqrt{3}} \leq |y| \leq 1}dy\,(3y^2-1)  \geq - \frac{4}{3\sqrt 3}   \qquad\qquad\qquad \ell \text{ - even}. \label{pari}
\end{align}
To get the lower bound in the second to last line, we restricted the integral where the second Legendre polynomial is positive, so that we can infer the bound \eqref{pari} by using the monotonicity properties of the integrand. 
The remaining steps are elementary inequalities. 

\n
Therefore, by  \eqref{mercy}, together with the lower bounds \eqref{pluto}, \eqref{pluto1}, and \eqref{pari} we have
\[\begin{aligned}
&\text{for $\ell\geq 1$  odd } \qquad \Re S_{\ell} \geq \frac{\sqrt{3}}{2} \\ 
&\text{for $\ell\geq 2$  even } \qquad \Re S_{\ell} \geq \frac{\sqrt3}{18} .
\end{aligned}\]
\n
These give the lower bound \eqref{mjrtom} for $\ell\geq 1$.  We remark that the lower bound for $\ell \geq 1$ holds true whenever $\gamma \geq 0$, hence, for these values of $\ell$  the regularizing  three-body interaction does not play any role.  

\n
To obtain the bound for $\ell = 0$, we reason like in the proof of Lemma \ref{lemma:theta}. We prove that for any fixed $\gamma > \gamma_c^*$ there  exists $s >0 $ small enough, such that 
\begin{equation}\label{modern}
\Re\left( \f{\sqrt{3}}{2} (1-s) + S_{\off,0}\lf( k +\f{i}{2} \ri)  + S_{\reg,0} \lf( k +\f{i}{2} \ri)\right)  \geq 0. 
\end{equation}
If this is the case, by  \eqref{mercy}, we obtain the needed bound by noticing that 
\[
\Re S_0 \lf( k +\f{i}{2} \ri) = \f{\sqrt{3}}{2}  s +\Re\left( \f{\sqrt{3}}{2} (1-s) + S_{\off,\ell}\lf( k +\f{i}{2} \ri)  + S_{\reg,\ell} \lf( k +\f{i}{2} \ri)\right)  \geq  \f{\sqrt{3}}{2}  s.
\]
Changing variables in  \eqref{middle1} and \eqref{middle2} we obtain 
\[
	\Re S_{\mathrm{off},0}\Big(k+\frac i2\Big)=- \frac{4}{\sqrt 2} \int_{-\frac12}^{\frac12} dy\,\frac{ \cosh \big(   k\arccos \big(y\big) \big)}{ \sqrt{1+ y } \, \cosh (k\pi )} = - \frac{4}{\sqrt 2} \int_{\frac\pi3}^{\frac23 \pi} dt \,
	\sqrt{1-\cos t} \, \frac{ \cosh (   k t )}{  \cosh (k\pi )} 
\]
and
\[
	\Re S_{\mathrm{reg},0}\Big(k+\frac i2 \Big)=\frac \gamma{\sqrt2}\int_{-1}^1dy\,\frac{ \cosh \big(   k\arccos (y) \big)}{ \sqrt{1+ y } \, \cosh (k\pi )} = \frac \gamma{\sqrt2}\int_{0}^\pi dt \,
	\sqrt{1-\cos t} \, \frac{ \cosh (   k t )}{  \cosh (k\pi )} .
\]
Next we use the identities 
\[
\int_{\frac\pi3}^{\frac23 \pi} dt \,
	\sqrt{1-\cos t} \,  \cosh \big(   k t \big) = \sqrt 2 \,  \frac{  \sqrt 3  \cosh\big(k\frac\pi3 \big)  - 2k \sinh\big(k\frac\pi3 \big) + 2 \sqrt 3  k \sinh\big(k\frac 23 \pi \big)  - \cosh\big(k\frac23 \pi \big) }{1+4k^2}
\]
and
\[
\int_{0}^\pi dt \,
	\sqrt{1-\cos t} \,  \cosh \big(   k t \big) =  2 \sqrt 2 \,  \frac{ 1 + 2k \sinh(k\pi ) }{1+4k^2},
\]
to obtain the formulae
\[
	\Re S_{\mathrm{off},0}\Big(k+\frac i2\Big)= -4 \frac{  \sqrt 3  \cosh\big(k\frac\pi3 \big)  - 2k \sinh\big(k\frac\pi3 \big) + 2 \sqrt 3  k \sinh\big(k\frac 23 \pi \big)  - \cosh\big(k\frac23 \pi \big)  }{(1+4k^2) \cosh (k \pi )}
\]
and
\[
	\Re S_{\mathrm{reg},0}\Big(k+\frac i2 \Big)=2 \gamma   \frac{ 1 + 2k \sinh(k\pi ) }{(1+4k^2)  \cosh (k\pi )}.
\]
So that, 
\[
\Re\left( \f{\sqrt{3}}{2} (1-s) + S_{\off,\ell}\lf( k +\f{i}{2} \ri)  + S_{\reg,\ell} \lf( k +\f{i}{2} \ri)\right)
 = \f{f_0(k) +  f_1(k)}{(1+4k^2) \cosh(k\pi)}
\]
with 
\begin{align*}
 f_0(k) &= \frac{\sqrt3}{2} (1-s) \cosh(k\pi)    + 4 \cosh \lf(k\frac{2\pi}{3} \ri)  -  4 \sqrt3 \cosh \lf(k\frac{\pi}{3} \ri)+2\ga, \\
 f_1(k) &=  2 \sqrt 3 (1-s) k^2 \, \cosh(k\pi)  + 4k \left(
 \ga \sinh(k\pi) - 2\sqrt3 \sinh \lf(k\frac{2\pi}{3} \ri) + 2 \sinh \lf(k\frac{\pi}{3} \ri)\right) .
\end{align*}
To prove the bound \eqref{modern}, it is enough to show that $f_0 + f_1 \geq 0$. Since $f_0$ and $f_1$ are  even functions of $k$, it is enough to consider  $k\geq 0$.
We have
\[
f_0(k)= \sum_{j=0}^\infty \f{1}{ (2j)!}\lf( \f{k\pi}{3}\ri)^{2j} \lf( \frac{\sqrt3}{2}  (1-s) \cdot 9^j + 4 \cdot 4^j - 4 \sqrt3  \ri) +2\ga.
\]
Noticing that  for $s >0 $ small enough  one has $ \sqrt3 (1-s)  \cdot  9^j/2 + 4\cdot  4^j - 4 \sqrt3  >0$ for  all $j\geq 1$, we have the lower bound (it is convenient to keep the first two terms of the series)
\[
f_0(k) \geq  2\left(\gamma - \gamma_c^* - s \frac{\sqrt 3}{2}\right) +  \lf( \f{k\pi}{3}\ri)^{2}\bigg( \frac{\sqrt{3}}{4} (1 - 9s) +8  \bigg)  
\]
 with  $\ga_c^\ast = \frac74 \sqrt 3 - 2 $. Similarly, we have
\[
f_1(k)= 2 \sqrt3 (1-s) k^2\sum_{j=0}^\infty \f{1}{ (2j)!}\lf( \f{k\pi}{3}\ri)^{2j}   \cdot 9^j + 4k  \sum_{j=0}^\infty \f{1}{ (2j+1)!}\lf( \f{k\pi}{3}\ri)^{2j+1} \lf( \gamma \cdot 3^{2j+1}- 2\sqrt3 \cdot 2^{2j+1} +2\ri).
\]
Since $ \gamma \cdot 3^{2j+1}- 2\sqrt3 \cdot 2^{2j+1} +2>0$  for $\ga>\ga_c^\ast>1 $  and $j \geq 1$ we obtain a lower by keeping only the term $j=0$ 
\[
f_1(k)\geq 
k \lf( \f{k\pi}{3}\ri) \bigg( \frac{6\sqrt3}{\pi} (1-s)+ 12 \gamma - 16 \sqrt3 +8 \bigg).
\]
Hence, taking into account the fact that we are assuming $\gamma > \gamma_c^*$,  we have 
\[
f_0(k) + f_1(k) \geq  2\left(\gamma - \gamma_c^* - s \frac{\sqrt 3}{2}\right) +  k \lf( \f{k\pi}{3}\ri)\bigg(  \frac{\sqrt{3}}{12}  \pi(1-9s) + \frac{8}{3}\pi  +\frac{6\sqrt3}{\pi} (1-s)+ 12 \gamma_c^* - 16 \sqrt3 +8  \bigg)  
\]
and the r.h.s.  is positive for $s>0$ small enough because $\gamma - \gamma_c^*>0$ and $\frac{\sqrt{3}}{12}  \pi  + \frac{8}{3}\pi  +\frac{6\sqrt3}{\pi} + 12 \gamma_c^* - 16 \sqrt3 +8  \simeq 4.8>0$. 

\end{proof}
\begin{remark} \label{afa}
As a consequence, we have
\begin{align*}
(\hat{\Gamma}^{\la} \hat{\xi})(\p)&=  \sqrt{ \f{3}{4} p^2 +\lambda }\; \hat{\xi}(\p) 
-\frac1{\pi^2} \int_{\RE^3} \!\! d\q\,  \frac{ \hat{\xi}(\q)}{p^2+q^2+\p \cdot \q +\lambda}
 +(\widehat{a \,\xi})(\p) + \f{\gamma}{2 \pi^2}  \int_{\RE^3} \!\!  d\q\,  \frac{ \hat{\xi}(\q)}{|\p - \q|^2 } \nonumber\\
&=: (\hat{\Gamma}^{ \la}_{\text{diag}} \hat{\xi}) (\p)  +   (\hat{\Gamma}^{ \la}_{\text{off}} \hat \xi)(\p) +    (\hat{\Gamma}^{(1)}_{\text{reg}} \hat \xi)(\p) +    (\hat{\Gamma}^{(2)}_{\text{reg}} \hat \xi)(\p),
\end{align*}
since every term belong to $L^2$ for $\xi \in H^1(\RE^3)$, that is all four operators are bounded from $H^1$ to $L^2$. 
The claim is obvious for ${\Gamma}^{ \la}_{\text{diag}} $, it was proved in \cite[Proposition 5]{miche8} for ${\Gamma}^{ \la}_{\text{off}} $, while regarding $ {\Gamma}^{(2)}_{\text{reg}}$ it amounts to Hardy's inequality.
It is trivial for  $ {\Gamma}^{(1)}_{\text{reg}}$, since $a\in L^\infty$.
\end{remark}

\subsection{Regularity of the Charge\label{app:A2}}
Now, we study the regularity of the charge associated with $\psi\in H$, that is the solution of the equation appearing in \eqref{domH}.
\begin{proposition} \label{finefine} 
Assume  \eqref{cos} and $\ga>2$, let $\genf\in H^{1/2}(\RE^3)$ and let $\xi\in \D(\Ga^\la)$ be the solution of 
\beq \label{dominio4}
\Ga^\la \xi = \genf ,
\eeq
then $\xi \in H^{3/2}(\RE^3)$ and $ \| \xi\|_{H^{3/2}}\leq c \| \genf\|_{H^{1/2}}$.
\end{proposition}
\begin{proof}
We argue as in Proposition \ref{p:fine} and recast  \eqref{dominio4} as $T\xi = \genf^\la$ (see  \eqref{T} and \eqref{effela1}). We note that $\genf^\lambda \in H^{1/2}(\RE^3)$. To convince oneself that this is the case, recall that by Proposition \ref{p:fine},  $\xi \in  H^{1}(\RE^3)$  then $a\xi \in H^{1}(\RE^3)$ since $a$ and $a'$ are bounded; moreover, it can be checked that  the last term at the r.h.s. of Eq. \eqref{effela1} also belongs to $H^{1/2}$:
\begin{align}
&\int\!\!d\p\, p\,  \Big| \int\!\! d\q\, \f{\hat{\xi}(\q)}{(p^2 + q^2 + \p \cdot \q + \la)(p^2 + q^2 + \p \cdot \q)} \Big|^2 \nonumber\\
\leq &16 \int\!\!d\p\, p\,  \left(  \int\!\! d\q\, \f{ q^{1/2} |\hat{\xi}(\q)| }{ q^{1/2} (p^2 + q^2  +2 \la)(p^2 + q^2 )} \right)^{\!\!2} \leq 16 \, \|\xi\|^2_{H^{1/2}}  \! 
\int\!\!d\p\! \int\!\!d\q \, \f{p}{q (p^2 + q^2  + 2\la)^2 (p^2 + q^2 )^2}\nonumber\\
= & 256 \, \pi^2 \, \|\xi\|^2_{H^{1/2}}\int_0^{\infty} \!\!dp\int_0^{\infty} \!\!dq\, \f{p^3 \, q}{ (p^2 + q^2  + 2\la)^2 (p^2 + q^2 )^2}, \nonumber
\end{align}
introducing polar coordinates, one easily sees that the last integral is finite.

\n
To conclude the proof we are going to show that 
\begin{equation}\label{A18a}
\| \Delta^{3/2} \xi \| \leq c \|\genf^\lambda\|_{H^{1/2}}.
\end{equation}
Decomposing  $\xi$ and $\genf^\lambda$ on the basis of the spherical harmonics we obtain   
\[
\| \Delta^{3/2} \xi \|^2= \sum_{\ell m} \int_0^\infty |\hat \xi_{\ell m}(p) |^2 p^5\,dp ,\]
\[
\| \genf^\la \|^2 = \sum_{\ell m} \int_0^\infty |\genf^\la_{\ell m}(p) |^2 p^2 dp ,
\]
and 
\[
\| \Delta^{1/2} \genf^\la \|^2 = \sum_{\ell m} \int_0^\infty |\genf^\la_{\ell m}(p) |^2 p^3 dp ,
\]
where $ \hat{\genf}_{\ell m}^\la(p)$ satisfies the condition
\[
\int_0^{\infty}  |\hat{\genf}_{\ell m}^\la(p) |^2   p^2 (1 + p) dp < \infty,  \qquad
\ell \in \NA, \quad m=-\ell, \ldots , \ell.
\]
To prove the bound \eqref{A18a} we will show that 
\begin{equation}\label{release}
\int_0^\infty |\hat \xi_{\ell m}(p) |^2 p^5\,dp \leq c  \int |\genf^\la_{\ell m}(p) |^2 p^3 dp \qquad \forall \ell\in\NA,\; \ell\geq 1, \quad m = -\ell,\dots,\ell,
\end{equation}
where $c$ does not depend on $\ell$ and $m$, and 
\begin{equation}\label{release1}
\int_0^\infty |\hat \xi_{00}(p) |^2 p^5\,dp \leq c  \int |\genf^\la_{00}(p) |^2 p^2(1+p) dp.
\end{equation}
We remark that for $\ell = 0$ we have a slightly weaker bound involving both $\|\genf^\lambda\|$ and $\| \Delta^{1/2} \genf^\la \|$. \\

\n
i)  Case $\ell \geq 1$. 

\n
Decomposing  the equation $T\xi = \genf^\lambda$ on the basis of the spherical harmonics we obtain   \eqref{dominioonda}. We multiply it   by $p^2$, change variables as $p=e^x$, and $q=e^y$, and set $\zeta_{\ell m}(x) = e^{3 \,x} \hat \xi_{\ell m}(e^x)$, 
and $h_{\ell m}(x)= e^{2 \,x} \hat \genf_{\ell m}(e^x)$. In this way we obtain the equation. 
\beq\begin{aligned} \label{dominio2}
\f{\sqrt{3}}{2} \, \zeta_{\ell m}(x) - \f{1}{\pi} \int_{\RE}\!\! dy\, \zeta_{\ell m}(y) \,e^{(x-y)}\!\!\int_{-1}^{1}\!\!\! d\nu\, 
& \f{P_\ell(\nu)}{\cosh (x-y) + \nu/2}+ \\
&+ \f{\gamma}{2 \pi} \int_{\RE}\!\! dy\, \zeta_{\ell m}(y) \, e^{(x-y)}\!\!\int_{-1}^{1}\!\!\! d\nu\,  \f{P_\ell(\nu)}{\cosh (x-y) - \nu} = h_{\ell m}(x).
\end{aligned}\eeq
Since 
\[
\int_{\RE^+} |\hat \xi_{\ell m}(p) |^2 p^5\,dp =  \int_\RE |\zeta_{\ell m}(x) |^2 \, dx
\]
and 
\[
 \int |\genf^\la_{\ell m}(p) |^2 p^3 dp = \int_\RE |h_{\ell m}(x) |^2 \, dx,
\]
the bound \eqref{release} is equivalent to the  inequality
\begin{equation}\label{tell}
\|\zeta_{\ell m}\|_{L^2(\RE)} \leq c \|h_{\ell m}\|_{L^2(\RE)}.
\end{equation}
The integral equation \eqref{dominio2} can be conveniently studied via Fourier transform.  To proceed we start by noticing  that, taking into account the identity $\!\int_{-1}^{1} \! d\nu\, P_\ell(\nu)=0$, $\ell\geq1$, we have 
\begin{align}
S_{\mathrm{off},\ell}(k)
&= -\frac{1}{\pi} \int_\RE dx e^{-ikx} \int_{-1}^{1}\!\!\! d\nu\,   \f{ P_\ell(\nu) }{\cosh x + \nu/2} \nonumber \\
&= -\frac{1}{\pi} \int_\RE dx e^{-ikx} \int_{-1}^{1}\!\!\! d\nu\, P_\ell(\nu) \left( \f{1}{\cosh x + \nu/2} - \f{1}{\cosh x}\right) \nonumber \\ 
& =  -\frac{1}{2\pi} \int_\RE dx e^{-ikx}  \!\int_{-1}^{1}\!\!\! d\nu\, \f{P_\ell(\nu) \, \nu}{ \cosh x\, ( \cosh x + \nu/2 )}  \qquad \ell\geq 1. \label{tacchino}
\end{align}
The representation formula \eqref{tacchino} shows that $S_{\mathrm{off},\ell}(k)$ can be holomorphically extended to the strip $\{ |\Im k| <2\}$. The same argument can be repeated for $S_{\mathrm{reg},\ell}(k)$.
Therefore,  in Fourier transform,  \eqref{dominio2} reads  
\begin{equation}\label{cliente}
S_\ell (k+i) \hat \zeta_{\ell m} = \hat h_{\ell m},
\end{equation}
with $S_\ell$ given as in  \eqref{Stotale}, i.e., 
\[
S_\ell (k+i) = \frac{\sqrt3}{2} + S_{\off,\ell} (k+i) + S_{\reg,\ell} (k+i) .
\]
 By the unitarity of the Fourier transform to prove the bound \eqref{tell} it is enough to find a lower bound for $|S_\ell (k+i)|$ (see also the similar argument used in the proof of Proposition \ref{p:fine}). We concentrate on the real part of $S_\ell (k+i)$. 

\n
With straightforward calculations, starting from \eqref{eq:Soff}, one finds:
\[\begin{aligned} 
	\Re S_{\mathrm{off},\ell} (k+i) 
 &= 
\begin{cases} \displaystyle - \int_{-1}^{1}\!\!\! d\nu\,  P_\ell(\nu) \, \nu \,\f{    \sinh (k \arcsin (\nu/2) )}{  2 \sqrt{1- \nu^{2}/4 }  \, \sinh (k \pi/2)} & \;\;\;\;\; \text{for} \;  l \text{ even},
    \\
    \\
\displaystyle \;\;  \int_{-1}^{1}\!\!\! d\nu\,  P_\ell(\nu) \, \nu \,\f{  \cosh (k \arcsin (\nu/2) ) }{ 2 \sqrt{1- \nu^{2}/4 }  \, \cosh (k \pi /2)} & \;\;\;\;\;  \text{for} \;  l \text{ odd}.
\end{cases}
\end{aligned}
\]
Analogously
\[
\begin{aligned}
\Re S_{\mathrm{reg},\ell} (k+i)  &= 
\begin{cases} \displaystyle \gamma \int_{-1}^{1}\!\!\! d\nu\,  P_\ell(\nu) \, \nu \,\f{   \sinh (k \arcsin \nu )}{ 2 \sqrt{1- \nu^{2}}  \, \sinh (k \pi/2)} & \;\;\;\;\; \text{for} \;  l \text{ even},
    \\
    \\
\displaystyle \gamma \int_{-1}^{1}\!\!\! d\nu\,  P_\ell(\nu) \, \nu \,\f{  \cosh (k \arcsin \nu ) }{ 2 \sqrt{1- \nu^{2}}  \, \cosh (k \pi /2)} & \;\;\;\;\;  \text{for} \;  l \text{ odd}.
\end{cases}
\end{aligned}
\]
Let us observe that, using the recurrence formula for the Legendre polynomials $(\ell+1) P_{\ell+1}(\nu) = (2\ell+1) \nu P_\ell(\nu) - \ell P_{\ell-1}(\nu)$, we can rewrite
\beq\label{eq:QS}
	\Re S_{\mathrm{off},\ell} (k+i) =-\frac{\ell+1}{2(2\ell+1)}S_{\mathrm{off},\ell+1}(k)-\frac{\ell}{2(2\ell+1)}S_{\mathrm{off},\ell-1}(k)
\eeq
and analogously
\[
	\Re S_{\mathrm{reg},\ell} (k+i) =\frac{\ell+1}{2\ell+1}S_{\mathrm{reg},\ell+1}(k)+\frac{\ell}{2\ell+1}S_{\mathrm{reg},\ell-1}(k).
\]
Then,  using Lemma \ref{lm:S}, for any $k \in \RE$ we obtain that 
\[
\Re S_{\mathrm{off},\ell} (k+i)  \leq 0 \;\;\;\; \text{ for} \;\;\ell \;\; \text{even}, \;\;\;\;\;\;  \Re S_{\mathrm{off},\ell} (k+i)  \geq 0  \;\;\;\; \text{for} \;\;  \ell \;\; \text{odd} 
\]
and 
\[
\Re S_{\mathrm{reg},\ell} (k+i)  \geq 0 \;\;\;\; \text{ for any} \;\; \ell .
\]
Hence, 
\begin{equation}\label{1810}
\Re S_\ell (k+i) \geq  \frac{\sqrt3}{2} \qquad \text{for $\ell$ odd}. 
\end{equation}
Next,  we  focus attention on $\Re S_{\mathrm{off},\ell} (k+i) $ with $\ell$ even. Notice that  \eqref{eq:QS} and Lemma \ref{lm:S} imply
\[
	\Re S_{\mathrm{off},\ell} (k+i) \geq -\frac12\Big[\frac{\ell+1}{2\ell+1}+\frac\ell{2\ell+1}\Big]S_{\mathrm{off},1}(k)\geq -\frac12S_{\mathrm{off},1}(k).
\]
Then, using \cite[Lemma 3.5]{CDFMT}, we obtain
\begin{equation}\label{Qofd}
	\Re S_{\mathrm{off},\ell} (k+i) \geq 2\frac{\sqrt{3}}3-\frac4\pi=-\frac{\sqrt{3}}2\Big(\frac8{\sqrt{3}\pi}-\frac43\Big) =: -\frac{\sqrt{3}}2 d
\end{equation}
where $ d \in (0,1)$.  
\n
Therefore, we find
\[
\Re S_\ell (k+i) \geq  \frac{\sqrt3}{2}(1-d) \qquad \text{for $\ell$ even}. 
\]
The latter bound, together with the one in \eqref{1810}, give \eqref{release} with $c = (\frac{\sqrt3}{2}(1-d))^{-2}$ for all $\ell \geq 1$. 

\n
Next, we proceed with the analysis of the case $\ell =0$. \\

\n
ii) Case $\ell=0$

\n  For $\ell=0$, the kernels in  \eqref{dominio2} are too singular, in particular the regularization described in  \eqref{tacchino} does not apply; hence, we cannot extend $S_{\mathrm{off},0} (k)$ and $S_{\mathrm{reg},0} (k)$ to $\Im k =1$ and proceed starting from an equation of the form \eqref{cliente}. 

\n
In order to circumvent this difficulty, we  
define 
\be\label{deze}
\zeta_{t}(x) := e^{(3-t)x} \, \hat{\xi}_{00}(e^x), \qquad  x \in \RE,
\ee
for $t \in (0,1)$. 
Then, we multiply   \eqref{dominioonda} by $p^{2-t}$ and introduce the change of variables $p=e^x$, $q=e^y$ to obtain 
\begin{align*}
\f{\sqrt{3}}{2} \, \zeta_{t}(x) - \f{1}{\pi} \int_{\RE}\!\! dy\, \zeta_{t}(y) \,e^{(1-t)(x-y)}\!\! &\int_{-1}^{1}\!\!\! d\nu\,  \f{1}{\cosh (x-y) + \nu/2} \nonumber\\
&+ \f{\gamma}{2 \pi} \int_{\RE}\!\! dy\, \zeta_{t}(y) \, e^{(1-t)(x-y)}\!\!\int_{-1}^{1}\!\!\! d\nu\,  \f{1}{\cosh (x-y) - \nu} = h_{t}(x)
\end{align*}
where 
\[
h_{t}(x)= e^{(2-t)x}  \hat{\genf}_{00}^\la(e^x) .
\]
The integral kernels in  \eqref{eqca0} are regular for $t\in(0,1)$. Let us rewrite the equation in such a way to isolate the term that becomes singular for $t \rightarrow 0$.
\begin{align}\label{eqca0}
\f{\sqrt{3}}{2} \, \zeta_{t}(x) + \f{\gamma-2}{2\pi} \! \int_{\RE}\!\! dy\, \zeta_{t}(y) \, \, & \f{e^{(1-t)(x-y)}}{\cosh (x-y)}  +  \f{1}{2\pi}\! \int_{\RE}\!\! dy\, \zeta_{t}(y) \,\f{e^{(1-t)(x-y)}}{ \cosh (x-y)} \int_{-1}^1 \!\!d \nu\, \f{\nu}{  \cosh (x-y) + \nu/2 } \nonumber\\
   & +    \f{\gamma}{2\pi}\! \int_{\RE}\!\! dy\, \zeta_{t}(y) \,\f{ \, e^{(1-t)(x-y)}}{\cosh (x-y)}  \int_{-1}^1 \!\!d \nu\, \f{\nu}{ \cosh (x-y) - \nu} 
   = h_{t}(x)
\end{align}
In the Fourier space equation \eqref{eqca0} is
\be\label{eqca01}
\f{\sqrt{3}}{2} \hat{\zeta}_{t}(k) +   \Big(Q^0_{t}(k)+ Q^1_{t}(k)  +  Q^2_{t}(k) \Big)   \hat{\zeta}_{t}(k)= \hat{h}_{t} (k)
\ee
where 
\begin{align}
 Q^0_{t}(k) &= \f{\gamma\!-\!2}{2\pi}  \int_{\RE} \!\! dx\, e^{-ikx} \, \f{e^{(1-t)x}}{\cosh x} = \frac12 \f{\gamma\!-\!2}{\cosh \Big( \f{\pi}{2}k + i \f{\pi}{2} (1-t) \Big) }= \frac12 \f{\gamma\!-\!2}{\sin \f{\pi}{2} t \, \cosh \f{\pi}{2} k \, + i \, \cos \f{\pi}{2} t \, \sinh \f{\pi}{2} k} \, \nonumber\\
 & = \frac{\gamma-2}{2} \, \f{\sin \f{\pi}{2} t \, \cosh \f{\pi}{2} k \, - i \, \cos \f{\pi}{2} t \, \sinh \f{\pi}{2} k}{\Big( \sin \f{\pi}{2} t \, \cosh \f{\pi}{2} k \Big)^2 + \Big( \cos \f{\pi}{2} t \, \sinh \f{\pi}{2} k \Big)^2}\,, \nonumber \\
 \label{q1ep} Q^1_{t}(k)&= \f{1}{2\pi} \int_{-1}^1 \!\!d\nu\, \nu \!\! \int_\RE\!\!dx\, e^{-ikx} \f{e^{(1-t)x}}{\cosh x \,(\cosh x + \nu/2)}, \\
 Q^2_{t}(k)&= \f{\gamma}{2\pi} \int_{-1}^1 \!\!d\nu\, \nu \!\! \int_\RE\!\!dx\, e^{-ikx} \f{e^{(1-t)x}}{\cosh x \,(\cosh x - \nu)}. \label{q2ep}
\end{align}
From equation \eqref{eqca01}, we have
\begin{align}\label{zeq}
|\hat{\zeta}_{t}(k) |^2 &=
\f{ |\hat{h}_{t} (k)|^2 }{\Big| \f{\sqrt{3}}{2} +  Q^0_{t}(k)+ Q^1_{t}(k)  +  Q^2_{t}(k)  \Big|^2 } 
\leq 
\f{ |\hat{h}_{t} (k)|^2 }{\Big[ \f{\sqrt{3}}{2} +  \Re Q^0_{t}(k)+ \Re Q^1_{t}(k)  + \Re  Q^2_{t}(k)  \Big]^2 }.
\end{align}
We notice that 
\be\label{Q0po}
\Re Q^0_{t}(k) \geq 0 \;\;\;\;\;\; \text{for}\;\; \gamma>2.
\ee
Moreover, 
\[\begin{aligned}
 \Re Q^1_{0}(k)=& \f{1}{2\pi} \frac12 \left(\int_{-1}^1 \!\!d\nu\, \nu \!\! \int_\RE\!\!dx\, e^{-ikx} \f{e^{x}}{\cosh x \,(\cosh x + \nu/2)} +\int_{-1}^1 \!\!d\nu\, \nu \!\! \int_\RE\!\!dx\, e^{ikx} \f{e^{x}}{\cosh x \,(\cosh x + \nu/2)}  \right) \\ 
&  \text{(in the second integral we change $x \to-x$)}
 \\
 =& \f{1}{2\pi}  \int_{-1}^1 \!\!d\nu\, \nu \!\! \int_\RE\!\!dx\,  \f{e^{-ikx}}{\cosh x + \nu/2}   = -\frac12 S_{\off,1}(k) \geq -\frac{\sqrt3}{2} d
 \end{aligned}
\]
by \cite[Lemma 3.5]{CDFMT}, see also  \eqref{Qofd}. Similarly, 
\[
\Re  Q^2_{0}(k)= \f{\gamma}{2\pi} \int_{-1}^1 \!\!d\nu\, \nu \!\! \int_\RE\!\!dx\,  \f{e^{-ikx}}{\cosh x - \nu} = S_{\reg,1}(k) \geq 0 
\]
Hence, we have
\begin{align}\label{q1ne}
\Re Q^1_{t}(k) = \Re  Q^1_{0}(k) + \Re \big( Q^1_{t}(k) -  Q^1_{0}(k) \big) \geq -\f{\sqrt{3}}{2} \, d - \| Q^1_{t} - Q^1_{0}\|_{L^{\infty}}
\end{align}
where (see \eqref{q1ep})
\begin{align*}
 \| Q^1_{t} -  Q^1_{0}\|_{L^{\infty}} \leq \f{1}{2\pi} 
 \!\int\!\!dx\, \Big| \f{( e^{-t x} -1) \, e^x}{ \cosh x}  \int_{-1}^1 \!\! d \nu \, \f{\nu}{\cosh x + \nu/2} \Big| \leq  \f{1}{2\pi}\!
 \int\!\!dx\,  \f{| e^{-t x} -1| \, e^x}{  \cosh x  \, (\cosh x - 1/2)} .
\end{align*}
By dominated convergence theorem, we find  that  $\| Q^1_{t} -  Q^1_{0}\|_{L^{\infty}} \rightarrow 0$ for $t \rightarrow 0. \;$ 
Analogously, we have
\begin{align}\label{q2ne}
\Re Q^2_{t}(k) = \Re  Q^2_{0}(k) + \Re \big( Q^2_{t}(k) - Q^2_{0}(k) \big) \geq - \| Q^2_{t} - Q^2_{0}\|_{L^{\infty}}
\end{align}
where (see \eqref{q2ep})
\begin{align}\label{q2reg}
\| Q^2_{t} -Q^2_{0}\|_{L^{\infty}}
\leq \f{\gamma}{2\pi} \int\!\!dx\,  \f{ | e^{-t x} -1| \, e^x}{ \cosh x}  \int_{-1}^1 \!\! d \nu \, \f{1}{\cosh x - \nu}   .
\end{align}
The last integral in \eqref{q2reg} can be estimated as follows. For $|x|<1$ we have
\begin{align*}
  \int_{-1}^1 \!\! d \nu \, \f{1}{\cosh x - \nu}  &=
 \int_0^1 \!\! d\nu\, \f{1}{\cosh x + \nu} +  \int_0^1 \!\! d\nu\, \f{1}{\cosh x - \nu} \leq  \int_0^1 \!\! d\nu\, \f{1}{1 + \nu} + \int_0^1 \!\! d\nu\, \f{1}{1+\f{x^2}{2}  - \nu} \nonumber\\
 &= \log 2 + \log \big( 1 + \f{x^2}{2} \big) - \log \f{x^2}{2} \leq \log \f{1}{x^2} + c
\end{align*}
where we have used the inequality $\cosh x \geq 1 + \f{x^2}{2}$. For $|x|\geq 1$ we have
\begin{align*}
  \int_{-1}^1 \!\! d \nu \, \f{1}{\cosh x - \nu}  \leq    \int_{-1}^1 \!\! d\nu\, \f{1}{\cosh x -1} = \f{2}{\cosh x -1}.
  \end{align*}
Using the above estimates, we can apply the dominated convergence theorem in \eqref{q2reg} and we find that   $\| Q^2_{t} - Q_{\mathrm{reg},0}\|_{L^{\infty}} \rightarrow 0$ for $t \rightarrow 0.\;$
Taking into account  \eqref{Q0po}, \eqref{q1ne}, \eqref{q2ne} and considering  $t$ sufficiently small, from \eqref{zeq} we obtain
\begin{align*}
\int\!\!dx\, |\zeta_{t}(x)|^2 &= \!
\int\!\!dk\, |\hat{\zeta}_{t} (k)|^2 \leq c\! \int \!\!dk\, |\hat{h}_{t} (k)|^2
= c \!
\int\!\!dx\, |h_{t} (x)|^2 =c\!\int_0^{\infty} \!\!\!dp\, p^2 p^{1-2t} |\genf_{00}^\la(p) |^2 \nonumber\\
&\leq c\!\int_0^{\infty} \!\!\!dp\, p^2 (1+p)  |\genf_{00}^\la(p) |^2 
\end{align*} 
where  $\,c\,$ is a constant independent of $t$. Moreover, by \eqref{deze} we have $|\zeta_{t}(x)|^2 = |e^{(3-t)x} \, \xi_{00}(e^x)|^2 \rightarrow    |e^{3x} \, \xi_{00}(e^x)|^2 $ for $t \rightarrow 0$ a.e.. Then, applying Fatou's lemma, we find 
\[ 
\int_0^{\infty} \!\! dp\, p^5\,  |\hat{\xi}_{00} (p)|^2 = \int\!\! dx \, |e^{3x} \hat{\xi}_{00} (e^x)|^2  \leq \liminf_{t\to0} \int \!\! dx |\zeta_{t}(x)|^2 \leq  c\!\int_0^{\infty} \!\!\!dp\, p^2 (1+p)  |\genf_{00}^\la(p) |^2 ,
\]
this gives the bound \eqref{release1} and concludes the proof of the proposition.

\end{proof}

\vs

\begin{remark} \label{owari}
Notice also that $\Ga^\la : H^{s+1}\to H^s$ for $s\in(0,1/2)$.The claim is obvious for ${\Gamma}^{ \la}_{\text{diag}} $ and it was proved in \cite[Proposition 5]{miche8} for ${\Gamma}^{ \la}_{\text{off}} $.
Regarding  ${\Gamma}^{ (2)}_{\text{reg}}$, due to Hardy's inequality, it is sufficient to prove that it is a bounded operator from $\dot H^{s+1}$ to  $\dot H^{s}$ for $0<s<1/2$.
Let us consider 
\[
T(\p, \q)=  \frac{p^s}{(p^2 + q^2 + \p \cdot \q +\lambda) q^{s+1} }.
\]
If we put $f(\p)= p^{-3/2}$, it is straightforward to prove that for $0<s<1/2$ we have
\[
\int T(\p, \q) \,f(\q)\, d\q \leq c_1 \, f(\p) \qquad \qquad \int T(\p, \q) \,f(\p)\, d\p \leq c_2 \, f(\q) ,
\]
and then $T$ is the  integral kernel of an $L^2$-bounded operator by Schur's test and the claim on ${\Gamma}^{ (2)}_{\text{reg}}$ follows. For  ${\Gamma}^{ (2)}_{\text{reg}}$ the claim is trivial since $\theta \in C^1$.
\end{remark}

\vs

\vs\vs\vs

\end{document}